\documentclass{article}

\usepackage{arxiv}

\usepackage[utf8]{inputenc} 
\usepackage[T1]{fontenc}    
\usepackage{hyperref}       
\usepackage{url}            
\usepackage{booktabs}       
\usepackage{amsfonts}       
\usepackage{nicefrac}       
\usepackage{microtype}      
\usepackage{lipsum}		
\usepackage{graphicx}
\usepackage{natbib}
\usepackage{doi}

\usepackage{graphicx}
\usepackage{amssymb}
\usepackage{mathtools}
\usepackage{quiver}
\usepackage{tikz}
\usepackage{float}
\usepackage{tasks}
\usepackage{stix}
\usepackage{booktabs}
\usepackage{subcaption}
\usepackage{csquotes}
\usepackage{comment}
\usepackage{amsthm}
\usepackage{amsmath}
\usepackage{subcaption}
\usepackage{proof}
\usetikzlibrary{positioning}
\usetikzlibrary{shapes.geometric}
\pgfmathtruncatemacro\distance{1}
\RequirePackage{tikz-qtree}

\newtheorem{therm}{Theorem}[section]
\newtheorem{lemma}[therm]{Lemma}
\newtheorem{corollary}[therm]{Corollary}
\newtheorem{proposition}[therm]{Proposition}
\newtheorem{remark}[therm]{Remark}
\newtheorem{definition}[therm]{Definition}

\usepackage{natbib}
\usepackage[shortlabels]{enumitem}

\newcommand{\sub}                         {{\mathsf{sub}}}
\newcommand{\up}                         {{\uparrow}}
\newcommand{\down}                         {{\downarrow}}

\newcommand{\tuple}[1]                         {{\langle #1\rangle}}

\newcommand{\dmneg}{{\sim}}

\newcommand{\todo}[1]{{\color{red}{#1}}}

\newcommand{\MoisilRule}[1]{\mathsf{r}^{\mathsf M}_{#1}}
\newcommand{\PPTopRule}[1]{\mathsf{r}^{\top}_{#1}}

\newcommand{\la}{\langle}
\newcommand{\ra}{\rangle}
                               
\newcommand{\A}[1][A]{\ensuremath{\mathbf{#1}} }
\newcommand{\De}{\Delta}
\newtheorem{example}{Example}
\newtheorem{problem}{Problem}
\DeclareSymbolFont{symbolsC}{U}{txsyc}{m}{n}
\DeclareMathSymbol{\strictif}{\mathrel}{symbolsC}{74}

\newcommand{\AlgA}{\mathbf{A}}
\newcommand{\AlgB}{\mathbf{B}}
\newcommand{\AlgC}{\mathbf{C}}
\newcommand{\AlgD}{\mathbf{D}}

\newcommand{\PPVar}{\mathbb{PP}}
\newcommand{\PPImpResVar}{\Variety(\PPSixResImp)}
\newcommand{\PPImpResTempVar}{\mathbb{PP}^{\ResImp}}
\newcommand{\PPSix}{\mathbf{PP_6}}
\newcommand{\PPSixDesImp}{\mathbf{PP_6^{\DesImp}}}
\newcommand{\PPSixResImp}{\mathbf{PP_6^{\ResImp}}}

\newcommand{\PPFourResImp}{\mathbf{PP_4^{\ResImp}}}
\newcommand{\PPThreeResImp}{\mathbf{PP_3^{\ResImp}}}
\newcommand{\PPTwoResImp}{\mathbf{PP_2^{\ResImp}}}

\newcommand{\Imp}{\GenImp}

\newcommand{\SDH}{\mathbb{SDH}_1}
\newcommand{\ISSix}{\mathbf{IS_6}}

\newcommand{\ISAlg}{\mathbf{IS}}
\newcommand{\PPlogic}{\mathcal{PP}_{\!\leq}}

\newcommand{\PartialMatValues}{\mathcal{V}_{10}}
\newcommand{\PartialMatDes}{D}
\newcommand{\PartialMat}{\mathfrak{M}}

\newcommand{\TotComp}[1]{\mathbb{T}(#1)}

\usepackage{transparent}
\newcommand{\semitransp}[2][0.35]{{\transparent{#1}#2}}
\newcommand{\HighlDes}[1]{{\color{black}{#1}}}
\newcommand{\NotHighlPic}[1]{{\semitransp{#1}}}

\newcommand{\PPResImplogicSingle}{{\mathcal{PP}^{\ResImp}_{\!\leq}}}
\newcommand{\PPResImplogicOrderMC}{{\mathcal{PP}^{\sequent,\ResImp}_\mathsf{up}}}
\newcommand{\PPResImplogicPrimeMC}{{\mathcal{PP}^{\sequent,\ResImp}_{\!\leq}}}

\newcommand{\PPResAsslogicSingle}{{\mathcal{PP}^{\ResImp}_\top}}

\newcommand{\InterpValSet}{\mathcal{U}}

\newcommand{\IdempEqName}{$\Delta$-idemp}

\newcommand{\PPResImplogic}{\mathcal{PP}^{\ResImp}_{\!\leq}}

\newcommand{\PPResAsslogic}{\mathcal{PP}^{\ResImp}_\top}

\newcommand{\PPlogicMC}{\mathcal{PP}_{\!\leq}^{\mathsf{\sequent}}}
\newcommand{\PPlogicPrimeMC}{\mathcal{PP}_{\!\leq}^{\mathsf{\sequent}}}
\newcommand{\PPOnelogic}{\mathcal{PP}_\top}
\newcommand{\PPOnelogicMC}{\mathcal{PP}_{\top}^{\mathsf{\sequent}}}
\newcommand{\FDELogic}{\mathcal{B}}
\newcommand{\powerset}{\wp}
\newcommand{\sequent}{\rhd}
\newcommand{\notsequent}{\,\smallblacktriangleright}
\newcommand{\SetSet}{\textsc{Set-Set}}
\newcommand{\CalcVar}{\mathsf{R}}
\newcommand{\RuleA}{\mathsf{r}}
\newcommand{\SetFmla}{\textsc{Set-Fmla}}
\newcommand{\Props}{\mathsf{props}}

\newcommand{\notsequentx}[1]{\smallblacktriangleright_{#1}}

\newcommand{\cons}{{\circ}}

\newcommand{\inferx}[3][]{\frac{#2}{#3}\,#1}
\newcommand{\Hom}{\mathsf{Hom}}
\newcommand{\DefMat}{\mathfrak{M}}
\newcommand{\End}{\mathsf{End}}

\newcommand{\Variety}{\mathbb{V}}

\newcommand{\LangAlg}[1]{\mathbf{L}_{#1}(P)}
\newcommand{\LangSetProp}[2]{L_{#1}(#2)}
\newcommand{\LangSet}[1]{L_{#1}(P)}
\newcommand{\DefProp}{p}
\newcommand{\DeffProp}{q}
\newcommand{\Fm}{\varphi}

\newcommand{\DMNeg}{{\sim}}
\newcommand{\PPComp}{\neg}
\newcommand{\DefCon}{\copyright}
\newcommand{\LSig}{\Sigma^{\mathsf{bL}}}
\newcommand{\DMSig}{\Sigma^{\mathsf{DM}}}
\newcommand{\DMoSig}{\Sigma^{\mathsf{PP}}}
\newcommand{\DMnSig}{\Sigma^{\mathsf{IS}}}

\newcommand{\DMSigImp}{\Sigma^{\mathsf{DM}}_{\Imp}}

\newcommand{\HOp}{\mathbb{H}}
\newcommand{\SOp}{\mathbb{S}}
\newcommand{\POp}{\mathbb{P}}

\newcommand{\OneAssert}{\vdash^{\top}}

\newcommand{\PPAlg}{\mathbf{PP}}

\newcommand{\Reduce}[1]{{#1}^{\ast}}

\newcommand{\CL}{\mathcal{CL}}

\DeclareMathOperator{\ConSet}{\mathsf{Cng}}

\newcommand{\neither}{\ensuremath{\mathbf{n}}}
\newcommand{\both}{\ensuremath{\mathbf{b}}}
\newcommand{\fvalue}{\ensuremath{\mathbf{f}}}
\newcommand{\tvalue}{\ensuremath{\mathbf{t}}}
\newcommand{\efvalue}{\ensuremath{\hat{\mathbf{f}}}}
\newcommand{\etvalue}{\ensuremath{\hat{\mathbf{t}}}}

\newcommand{\FourSet}{\mathcal{V}_4}
\newcommand{\SixSet}{\mathcal{V}_6}
\newcommand{\SymbDef}{\,{\coloneqq}\,}
\newcommand{\FDEAlg}{\mathbf{DM}_4}
\newcommand{\KleeneAlg}{\mathbf{DM}_3}
\newcommand{\BoolAlg}{\mathbf{DM}_2}
\newcommand{\Upset}[1]{{{\uparrow}#1}}
\newcommand{\SymLogEquiv}{{\lhd\rhd}}

\newcommand{\LeibCong}[1]{\Omega^{#1}}

\newcommand{\GenImp}{\Rightarrow}
\newcommand{\ClassImpSymb}{\Rightarrow}
\newcommand{\HeyImpSymb}{\Rightarrow}

\newcommand{\NodeSet}{\mathsf{nds}}
\newcommand{\TreeRel}{\mathsf{\leq}}
\newcommand{\DefTree}{t}
\newcommand{\DefNode}{n}
\newcommand{\Label}[1]{l^{#1}}
\newcommand{\Disc}{\ast}
\newcommand{\Root}[1]{\mathsf{rt}(#1)}
\newcommand{\Children}[1]{\mathsf{chdr}}

\newcommand{\DesImp}{\ClassImpSymb_{\mathsf{A}}}
\newcommand{\ResImp}{\HeyImpSymb_{\mathsf{H}}}
\newcommand{\FullImp}{\ClassImpSymb_{\mathsf{A1}}}

\newcommand{\ArbDisj}{\lor}

\newcommand{\CondAOne}{(\mathsf{A}1)}
\newcommand{\CondATwo}{(\mathsf{A}2)}
\newcommand{\CondIOne}{(\mathsf{I}1)}
\newcommand{\CondITwo}{(\mathsf{I}2)}

\newcommand{\PPSixMat}{\DefMat_6}

\newcommand{\DiscSet}[1]{\Omega_{#1}}
\newcommand{\DiscNSet}[1]{\mho_{#1}}

\newcommand{\FmSetA}{\Phi}
\newcommand{\FmSetB}{\Psi}
\newcommand{\FmSetC}{\Pi}
\newcommand{\FmSetD}{\Theta}
\newcommand{\FmSetAnalytic}{\Xi}

\newcommand{\FmA}{\varphi}
\newcommand{\FmB}{\psi}
\newcommand{\FmC}{\xi}
\newcommand{\FmD}{\theta}

\newcommand{\OrderPresMC}[1]{\sequent_{#1}^{\leq}}
\newcommand{\NOrderPresMC}[1]{\notsequent_{#1}^{\leq}}

\newcommand{\IneqPresMC}[1]{\sequent_{#1}^{\leq}}

\newcommand{\ig}{\equiv}
\newcommand{\men}{\leq }

\title{Adding an Implication to Logics of Perfect Paradefinite Algebras}

\author{ \href{https://orcid.org/0000-0003-3240-386X}{\includegraphics[scale=0.06]{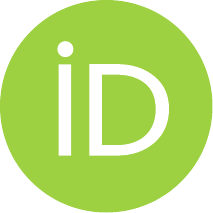}\hspace{1mm}Vitor Greati}\\
	Bernoulli Institute\\
	University of Groningen\\
	Groningen, The Netherlands\\
	\texttt{v.rodrigues.greati@rug.nl} \\
	\And
	\href{https://orcid.org/0000-0002-6941-7555}{\includegraphics[scale=0.06]{orcid.pdf}\hspace{1mm}Sérgio Marcelino} \\
	SQIG - Instituto de Telecomunicações\\
    Departamento de Matemática - Instituto Superior Técnico\\
    Universidade de Lisboa\\
	Lisboa, Portugal\\
	\texttt{smarcel@math.tecnico.ulisboa.pt} \\
    \And
	\href{https://orcid.org/0000-0003-2601-8164}{\includegraphics[scale=0.06]{orcid.pdf}\hspace{1mm}João Marcos} \\
	Department of Philosophy\\
	UFSC\\
	Santa Catarina, Brazil\\
	\texttt{botocudo@gmail.com} \\
    \And
	\href{https://orcid.org/0000-0003-1364-5003}{\includegraphics[scale=0.06]{orcid.pdf}\hspace{1mm}Umberto Rivieccio} \\
	Departamento de Lógica, Historia y Filosofía de la Ciencia\\
	Universidad Nacional de Educación a Distancia\\
	Madrid, Spain\\
	\texttt{umberto@fsof.uned.es} \\
}

\renewcommand{\shorttitle}{\textit{arXiv} Template}

\hypersetup{
pdftitle={A template for the arxiv style},
pdfsubject={q-bio.NC, q-bio.QM},
pdfauthor={David S.~Hippocampus, Elias D.~Striatum},
pdfkeywords={First keyword, Second keyword, More},
}

\begin{document}
\maketitle

\begin{abstract}
	Perfect paradefinite algebras are
De Morgan algebras expanded with a
perfection (or classicality) operation.
They form a variety that is term-equivalent
to the variety of involutive Stone algebras.
Their associated multiple-conclusion (\SetSet{})
and single-conclusion (\SetFmla{}) order-preserving logics are
non-algebraizable self-extensional
logics of formal inconsistency and undeterminedness
determined by
a six-valued matrix, studied in depth 
by Gomes et al. (2022)
from both the algebraic and the proof-theoretical
perspectives.
We continue hereby that study
by investigating directions
for conservatively expanding
these logics
with an implication connective
(essentially, one that
admits the deduction-detachment
theorem).
We first consider logics given by very simple and manageable non-deterministic
semantics whose implication (in isolation) is classical.
These, nevertheless, fail to be self-extensional.
We then consider the
implication
realized by the relative pseudo-complement
over the
six-valued perfect paradefinite
algebra. Our strategy is to
expand such algebra with this connective
and study the (self-extensional) \SetSet{} and \SetFmla{} 
order-preserving and $\top$-assertional
logics
of the variety induced by the new algebra.
We provide axiomatizations
for such new variety and for such logics,
drawing parallels with
the class of symmetric Heyting algebras
and with Moisil's `symmetric modal logic'.
For the \SetSet{} logic, in particular,
the axiomatization we obtain is analytic.
We close by studying interpolation
properties for these logics
and concluding that the new variety
has the Maehara amalgamation property.
\end{abstract}


\section{Introduction}
\label{sec:introduction}
The algebraic structures we now call \emph{involutive Stone algebras}
appear to have been first considered by Roberto Cignoli and Marta Sagastume
along their investigation of finite-valued \L ukasiewicz logics, in connection with 
\L ukasiewicz-Moisil algebras (\cite{cignoli1981,cignoli1983}).
Formally, involutive Stone algebras are usually presented 
as expansions of De Morgan algebras (whose languages consists of a conjunction $\land$, a disjunction $\lor$, a negation $\DMNeg$ and the lattice bounds $\bot, \top$)
in one of the following two alternative term-equivalent ways
(we shall soon see that a third one has been recently proposed):
\begin{enumerate}


\item by adding a unary `possibility' operator (usually denoted by~$\nabla$ in the literature);

\item by adding an `intuitionistic' (namely, a pseudo-complement) negation (denoted by~$\neg$) satisfying the well-known Stone equation: $\neg x \lor \neg \neg x \approx \top$.

\end{enumerate}

\noindent 
Though obviously arising as ``algebras of logic'', the class of involutive Stone algebras was not employed as an algebraic semantics for logical systems  in the above-mentioned seminal  works; rather, this has been pursued in a series
of 
recent papers which focused in particular
on the logic that preserves the lattice order of involutive Stone algebras (\cite{cantu:2019,cantu:2020,marcelino2021,cantu2022}).

In the paper~\cite{Gomes2022}, which is the immediate predecessor of the present one, 
we built on the observation (made in~\cite{cantu:2020}) that the logic of involutive Stone algebras may be viewed as a \emph{Logic of Formal Inconsistency}
in the sense of~\cite{jm2005thesis,03-CCM-lfi}, and axiomatized it as such. This was possible due to the fact that involutive Stone algebras may be presented in a third
term-equivalent way, namely: 
\begin{enumerate}
\setcounter{enumi}{2}
\item by adding a unary `perfection' operator (denoted by $\circ$) of the kind discussed in~\cite{jmarcos2005}.
\end{enumerate}
Reformulated in this language (i.e., $\{ \land, \lor, \DMNeg, \cons, \bot, \top \}$), involutive Stone algebras have been renamed   \emph{perfect paradefinite algebras} (PP-algebras) in~\cite{Gomes2022}.  In the latter paper
 we  axiomatized the $\SetFmla$ logic of order of PP-algebras, called $\PPlogic$,
and also showed that 
$\PPlogic$ is semantically determined by a finite matrix based
on the six-element algebra (there dubbed $\PPSix$) that generates the class of PP-algebras as a variety. Actually,
we developed most of our
results with respect to a $\SetSet{}$ 
version of this logic,
called $\PPlogicMC{}$.
The $\SetSet{}$ axiomatization we presented for it
is analytic, and thus
suitable
for automated reasoning.
We then used all those results
to aid in the study of $\PPlogic$.

In the present paper, 
further pursuing   this approach, we shall focus on the question of \emph{how to add an implication connective to 
$\SetSet{}$
and $\SetFmla{}$
logics associated to
PP-algebras}. In an effort to proceed in a systematic fashion, 
we shall be guided by the following main principles: 
\begin{itemize}
    \item the resulting system must be conservative over 
    the logics being extended;
    \item the new connective must indeed qualify as an implication according to some
     general standard. 
\end{itemize}
In order to narrow down our search and to eventually converge upon a short list of candidates, we shall presently consider  a third guiding principle as well, namely:
\begin{itemize}
    \item 
    the new logic must also qualify as the \emph{logic of order}
    of a suitable class of algebras.
\end{itemize}

\noindent In the present setting, as we shall see, the above principle turns out to be equivalent to requiring
the logic to be \emph{self-extensional}
{(see Definition~\ref{def:selfextensional},
Definition~\ref{def:order-pres-set-set-logic} and general results on
self-extensional logics in \citet*[Sec. 3]{Jansana2006})}.

The approach outlined above 
for the study of implicational
extensions
is reminiscent of --- indeed, directly inspired by--- Avron's proposal  for extending
Belnap's logic (\cite{avron2020}), 
in which two main requirements are entertained: 
first, that the new connective should be an implication
relative to a given set of designated elements 
(condition
$\CondAOne$  in
{Subsection~\ref{ss:impcon}}); second, that the resulting logic should be self-extensional $\CondATwo$.
We shall
discuss  these requirements at length in {Section~\ref{sec:classical-implication}}, where we show that
in our case there is unfortunately \emph{no implication connective that meets both of them}
(Theorem~\ref{thm:noexist}). 
This is where we shall choose to retain the latter requirement (self-extensionality)
instead of the former, which we find too restrictive, as we  
shall also discuss. As this choice  still leaves plenty of room for a large collection of alternative 
implications, we shall follow two North Stars: Algebra and Tradition.

Algebra suggests that,  for a wide family of logics,
a well-behaved implication connective 
may be obtained by considering, whenever available, the residuum of the algebraic operation that realizes the logical conjunction (see e.g.~\cite{GalJipKowOno07}): this leads us to expand the logics of PP-algebras by 
a relative pseudo-complement
implication,
which we shall call a 
\emph{Heyting implication},
for our definition mirrors the 
standard
one for the implication on Heyting algebras.

Tradition, in the present setting, happens to point in the same direction as algebra. In fact, it turns out
that a substantial
part of the theory of 
involutive Stone algebras 
had already been developed, even prior to Cignoli and Sagastume's works, in~\cite{Monteiro1980}, itself a collection
of earlier material. The main subject of Monteiro's monograph is the class of
\emph{symmetric Heyting algebras},  providing an  algebraic counterpart 
to Moisil's \emph{symmetric modal logic}. 
These algebras are
presented precisely in the traditional language of involutive Stone algebras 
(with the $\nabla$ operator which Monteiro dubs `possibility') enriched with a Heyting implication.

The formal relation between the class of PP-algebras enriched with an implication that we shall define  
and Monteiro's symmetric Heyting algebras is discussed in detail in Section~\ref{sec:moi}. For the time being,
let us conclude the present introductory remarks by mentioning an alternative proposal concerning our main question
  --- how to add an implication to logics of PP-algebras, one of them being $\PPlogic$ --- which can be retrieved from a recent paper by~\cite{coniglio2023sixvalued}.
The purpose of the latter is to add a 
`classicality' 
operator $\cons$ 
to the well-known Belnap-Dunn four-valued logic, but the authors actually start
from a logic which is itself a conservative expansion of the Belnap-Dunn logic with a 
\emph{classic-like implication} (in the sense of Definition~\ref{def:classicalimplication}). 
%
The resulting logic --  dubbed  $\mathrm{LET_K^+}$ -- turns out to be determined by a six-valued matrix
that (if we disregard the implication) 
coincides with the matrix
$\langle \PPSix, \Upset\both \rangle$ which  determines the logic $\PPlogic$ 
(see~\cite{Gomes2022}). 
As noted 
in~\citet[Subsec.~5.1]{coniglio2023sixvalued}, it follows that  the implication-free fragment of 
$\mathrm{LET_K^+}$ coincides with  $\PPlogic$. 
The implication connective of $\mathrm{LET_K^+}$ consists thus in a candidate for expanding 
$\PPlogic$ with an implication.
An inspection of its truth-table
(see \citet[Def.~3.9]{coniglio2023sixvalued}) reveals (see also Section~\ref{sec:classical-implication}) that this implication satisfies  (with respect to the designated set $ \Upset\both$) the first
among Avron's requirements, $\CondAOne$,
and  therefore destroys self-extensionality 
(or, one may say, preserves the non-self-extensional character of paraconsistent Nelson logic).
This leads to an interesting observation, namely, that  $\mathrm{LET_K^+}$ is determined by a refinement of the non-deterministic matrix introduced in Section~\ref{sec:classical-implication}, and may be obtained, thus, from the corresponding logic by adding suitable rules.
In this sense, we may say our approach subsumes that of~\cite{coniglio2023sixvalued}, at least insofar as
the task of conservatively expanding $\PPlogic$ 
is concerned.

Our paper is organized as follows.
 After the preliminary sections in which we fix the notation and basic definitions, 
 we delve into the problem of expanding 
 the logics of order of PP-algebras
 with an implication. 
 To that effect, we first consider in Section~\ref{sec:classical-implication} adding an implication
 satisfying at once $\CondAOne$
 and $\CondATwo$.
 After proving the impossibility of this task,
 we focus first on $\CondAOne$, which makes
 the implication qualify as classic-like.
 We explore possibilities
 of doing so using the semantical framework of
 non-deterministic logical matrices, providing analytic axiomatizations
 for the obtained \SetSet{} logics, 
 from which we can readily obtain
 \SetFmla{} axiomatizations
 for the corresponding \SetFmla{} companions.
 
 Next, turning our focus to $\CondATwo$, and having observed that the matrix
 $\langle \PPSix, \Upset\both \rangle$, which determines $\PPlogicMC$ and $\PPlogic$, is based 
 on an algebra (i.e., $\PPSix$) that is a finite distributive lattice, 
 in Section~\ref{sec:intuitionistic-implication}
 we enrich  its language with an operation $\ResImp$
 that corresponds to the relative pseudo-complement implication determined by the lattice order. 
 In fact, the six-element Heyting algebra thus obtained (denoted by $\PPSixResImp$) 
 turns out to be a symmetric Heyting algebra in Monteiro's sense. 
 We then consider the family of all matrices $\langle \PPSixResImp, D \rangle$ such that $D$ is a non-empty lattice filter 
 of the lattice reduct of $\PPSixResImp$. The ($\SetSet$ and $\SetFmla$) logics determined by this family
 are 
 the order-preserving logics induced by the variety $\Variety(\PPSixResImp)$,
guaranteed to be self-extensional.
 
 In Section~\ref{sec:int-hilbert-axiomat}, we proceed to axiomatize both the \SetSet{} and the $\SetFmla$ logics
 thus defined. First focusing on the \SetSet{} logic,
 we present an axiomatization after a detailed
 analysis of the algebraic structure
 of the matrices and of their designated sets.
 The calculus we present is not only sound and complete, but also analytic. 
 From such axiomatization, moreover, we show how to extract an axiomatization for
 the corresponding \SetFmla{} logic of order.

 In Section~\ref{sec:moi},
we look at the algebraic models that correspond to
the \SetFmla{} logics
$\PPResImplogic$ and $\PPResAsslogic$
(the $\top$-assertional
logic associated to
$\Variety(\PPSixResImp)$)
within 
the general theory of algebraization of logics.
We begin by observing that both logics
are closely related to
Moisil's `symmetric modal logic', whose algebraic
counterpart is the class of {symmetric Heyting algebras}
(Definition~\ref{d:symhey}). Indeed,
$\PPSixResImp$ is (term-equivalent to) a symmetric Heyting algebra,
and
  the class of algebras
 providing an algebraic semantics for both $\PPResImplogic$ and $\PPResAsslogic$
is precisely the subvariety $\Variety(\PPSixResImp)$ of symmetric Heyting algebras generated by $\PPSixResImp$ (Theorem~\ref{thm:algppass}); an equational presentation for $\Variety(\PPSixResImp)$ is introduced in Definition~\ref{d:PPImpResVar}, and
shown to be sound and complete in Theorem~\ref{thm:vargen}.
The algebraizability of Moisil's logic (Proposition~\ref{prop:moisalg})
entails that 
$\PPResAsslogic$ is also algebraizable; its  equivalent algebraic semantics 
is precisely
$\Variety(\PPSixResImp)$, and its reduced matrix models are the matrices of the form $\la \A, \{ \top \} \ra $ with $\A \in \Variety(\PPSixResImp)$ (Theorem~\ref{thm:algppass}). On the other hand,
$\PPResImplogic$ is not algebraizable (Proposition~\ref{prop:equi}) but
its algebraic counterpart is also $\Variety(\PPSixResImp)$; the shape of the reduced matrix models of $\PPResImplogic$
is described in Proposition~\ref{prop:redmodPPResImplogic}.
We conclude the section by looking at the subvarieties of $\Variety(\PPSixResImp)$.  
There are only three of them (corresponding precisely to the self-extensional 
extensions of $\PPResImplogic$) which can be axiomatized by adding the axioms that translate the equations shown in Corollary~\ref{cor:subvar}.

 In Section~\ref{sec:interp-amalg}, we draw some conclusions
 regarding interpolation for the implicative extensions
 of interest, as well
 as amalgamation for the
 class of algebraic models of
 $\PPResAsslogic$. 
 Further directions of investigation are
highlighted in Section~\ref{sec:conclusions}.


\section{Preliminaries}
\label{sec:preliminaries}
In this section, we introduce
the main concepts related to
algebras, logics and axiomatizations, 
with particular attention to
perfect paradefinite algebras and their logics.
We also fix 
what we shall understand as an
\emph{implication} and the selection criteria we will employ
in investigating expansions of
a logic by the addition of an implication.

\subsection{Algebras, languages and logics}

A \emph{propositional signature}
is a family
$\Sigma \SymbDef \{\Sigma_k\}_{k \in \omega}$,
where each $\Sigma_k$ is
a collection of $k$-ary \emph{connectives}.
A \emph{$\Sigma$-multialgebra} is a structure
$\mathbf{A} \SymbDef \langle A, \cdot^\mathbf{A} \rangle$, where $A$ is a non-empty
set called the \emph{carrier} of $\mathbf{A}$ and, for each $\DefCon \in \Sigma_k$, the \emph{multioperation} $\DefCon^\mathbf{A} : A^k \to \powerset(A)$ is the \emph{interpretation} of $\DefCon$
in $\mathbf{A}$.
We say that $\AlgA$ is a
\emph{$\Sigma$-algebra}
(or \emph{deterministic $\Sigma$-multialgebra})
when $\DefCon^\mathbf{A}(a_1,\ldots,a_k)$
is a singleton for every $a_1,\ldots,a_k \in A$, $\DefCon \in \Sigma_k$ and $k \in \omega$.
If $\DefCon^\mathbf{A}(a_1,\ldots,a_k) \neq \varnothing$ for every $a_1,\ldots,a_k \in A$, $\DefCon \in \Sigma_k$ and $k \in \omega$,
we say that $\AlgA$ is \emph{total}.
Note that the notion of $\Sigma$-algebra 
matches the usual notion of
algebra in Universal Algebra~(cf. \cite{Sankappanavar1987}).
Given $\Sigma^\prime \subseteq \Sigma$
(that is, $\Sigma^\prime_k \subseteq \Sigma_k$ for all $k \in \omega$),
the \emph{$\Sigma'$-reduct} of a $\Sigma$-multialgebra $\mathbf{A}$ is the $\Sigma'$-multialgebra over
the same carrier of $\mathbf{A}$ that agrees with $\mathbf{A}$ on the interpretation of the connectives in $\Sigma'$.


A \emph{$\Sigma$-homomorphism}
between $\Sigma$-multialgebras
$\AlgA$ and $\AlgB$ is a mapping
$h : A \to B$ such that
$h(\DefCon^{\AlgA}(a_1,\ldots,a_k)) \in \DefCon^\AlgB(h(a_1),\ldots,h(a_k))$ for all
$a_1,\ldots,a_k \in A$, $\DefCon \in \Sigma_k$ and $k \in \omega$.
Note that for $\Sigma$-algebras this definition coincides with the usual notion of
homomorphism.
The collection of all $\Sigma$-homomorphisms between
two $\Sigma$-multialgebras $\mathbf{A}$
and $\mathbf{B}$ is denoted by $\Hom(\mathbf{A},\mathbf{B})$.

Given a denumerable set $P \supseteq \{p, q, r, s, x, y\}$
of \emph{propositional variables}, 
the absolutely free $\Sigma$-algebra freely generated by $P$,
or simply the
\emph{language} over $\Sigma$ (generated by~$P$),
is denoted by $\LangAlg{\Sigma}$,
and its members are called \emph{$\Sigma$-formulas}.
{
The collection of
all subformulas of a formula
$\FmA$ is denoted by
$\mathsf{sub}(\FmA)$
and we let $\mathsf{sub}(\FmSetA) \SymbDef \bigcup_{\Fm\in\FmSetA}\mathsf{sub}(\FmA)$, 
for all $\FmSetA\subseteq\LangSet{\Sigma}$.
Similarly, the collection of all propositional variables
occurring in $\Fm \in \LangSet{\Sigma}$
is denoted by $\Props(\Fm)$, and
we let $\Props(\FmSetA) \SymbDef \bigcup_{\Fm\in\FmSetA}\Props(\Fm)$, 
for all $\FmSetA\subseteq\LangSet{\Sigma}$.}
The elements of $\Hom(\LangAlg{\Sigma}, \mathbf{A})$ will sometimes be referred to as \emph{valuations on~$\mathbf{A}$}.
When $\AlgA$ is $\LangAlg{\Sigma}$, valuations are 
endomorphisms on $\LangAlg{\Sigma}$
and are usually dubbed
\emph{substitutions}.
The set of all substitutions is denoted
by $\End(\LangAlg{\Sigma})$.
{Given $h,h' \in \Hom(\LangAlg{\Sigma}, \mathbf{A})$,
we shall say that \emph{$h'$ agrees with~$h$ on $\FmSetA \subseteq \LangSet{\Sigma}$} provided that $h'(\FmA)=h(\FmA)$
for all $\FmA \in \FmSetA$.}
In case $p_1,\ldots,p_k$
are the only propositional variables ocurring in $\FmA \in \LangSet{\Sigma}$,
we say that $\FmA$ is \emph{$k$-ary} and denote by $\FmA^\mathbf{A}$
the $k$-ary multioperation on $A$
such that
$\FmA^\mathbf{A}(a_1,\ldots,a_k) \SymbDef \{ h(\FmA) : 
h \in \Hom(\LangAlg{\Sigma}, \mathbf{A}) \text{ with }
h(p_i) = a_i \text{ for each }
1 \leq i \leq k
\}$,
for all $a_1,\ldots,a_k \in A$.
Also, if $\FmB_1,\ldots,\FmB_k \in \LangSet{\Sigma}$,
{we let $\FmA(\FmB_1,\ldots,\FmB_k)$ denote the formula
resulting from substituting
each $p_i$
in $\FmA$ by $\FmB_i$,
for all $1 \leq i \leq k$.}
For a set $\Theta$ of formulas $\FmA(p_1,\ldots,p_k)$,
we let $\Theta(\FmB_1,\ldots,\FmB_k) \SymbDef \{ \FmA(\FmB_1,\ldots,\FmB_k) : \FmA \in \Theta \}$.

A~\emph{$\Sigma$-equation} is
a pair $(\FmA,\FmB)$ of $\Sigma$-formulas that we denote by
$\FmA \approx \FmB$,
and a $\Sigma$-multialgebra $\mathbf{A}$ is said to \emph{satisfy}
$\FmA \approx \FmB$ if
$h(\FmA)=h(\FmB)$
for every $h \in \Hom(\LangAlg{\Sigma}, \mathbf{A})$. 
For any given collection of $\Sigma$-equations, the class of all $\Sigma$-algebras that satisfy such equations is called a \emph{$\Sigma$-variety}.
An equation is said to be \emph{valid} in a given variety if it is satisfied by each algebra in the variety. 
The variety generated by
a class $\mathsf{K}$ of $\Sigma$-algebras,
denoted by $\mathbb{V}(\mathsf{K})$,
is the closure of $\mathsf{K}$ under
homomorphic images, subalgebras 
and direct products.
We denote the latter operations, respectively,
by $\HOp$, $\SOp$
and $\POp$.
We write $\ConSet\mathbf{A}$ 
to refer to the collection of all congruence relations
on~$\mathbf{A}$, which is known to form
a complete lattice under inclusion.

{
In what follows, we assume the reader is familiar with basic notations and terminology of lattice theory~(\cite{priestley2002}). We denote by $\LSig$ the signature
containing but two binary connectives,
$\land$ and $\lor$, and
two nullary connectives $\top$ and $\bot$, and by $\DMSig$ the extension of the latter signature by the addition of a unary connective~$\DMNeg$.
Moreover, by adding the unary connective $\nabla$ to $\DMSig$ we obtain a signature we shall call $\DMnSig$, and by adding instead the unary connective $\cons$ we obtain a signature we shall call $\DMoSig$.
We provide below
the definitions and some examples of
De Morgan algebras and of involutive
Stone algebras, in order to illustrate some of the 
notions introduced above.

\begin{definition}
Given a $\DMSig$-algebra whose
$\LSig$-reduct is a bounded distributive lattice, we say that it
constitutes a \emph{De Morgan algebra} if it
satisfies the equations:\\
\begin{tabular}{@{}ll@{}}
    {\rm\textbf{(DM1)}} $\DMNeg \DMNeg x \approx x$
    \hspace{1cm}
    &
    {\rm\textbf{(DM2)}} $\DMNeg (x \land y) \approx \DMNeg x \lor \DMNeg y$\\
\end{tabular}
\end{definition}

\begin{example}
Let $\FourSet \SymbDef \{\tvalue, \both, \neither, \fvalue\}$
and
let $\FDEAlg \SymbDef \langle \FourSet, \cdot^{\FDEAlg} \rangle$
be the $\DMSig$-algebra known as
the Dunn-Belnap lattice, whose interpretations
for the lattice connectives
are those induced by the Hasse diagram in
Figure {\rm \ref{fig:de_morgan}},
and the interpretation for 
$\DMNeg$ is such that
$\DMNeg^{\FDEAlg}\fvalue \SymbDef \tvalue$,
$\DMNeg^{\FDEAlg}\tvalue \SymbDef \fvalue$
and
$\DMNeg^{\FDEAlg} a \SymbDef a$,
for $a \in \{\neither,\both\}$;
as expected, for the nullary connectives, we have
$\top^{\FDEAlg} \;\SymbDef\; \tvalue$
and $\bot^{\FDEAlg} \;\SymbDef\; \fvalue$. In Figure~{\rm\ref{fig:de_morgan}},
besides depicting the lattice structure of $\FDEAlg$, we also show
its subalgebras $\KleeneAlg$ and $\BoolAlg$, which 
coincide with the three-element
Kleene algebra and the two-element Boolean algebra.
These three algebras are the only subdirectly irreducible De Morgan algebras~{\rm\cite[Sec. XI.2, Thm. 6]{dwinger1975}}.
\end{example}
\begin{figure}[ht]
    \centering
    \begin{subfigure}[b]{.45\textwidth}
    \centering
    \begin{tikzpicture}[node distance=1.25cm]
    \node (0DM4)                      {$\fvalue$};
    \node (BDM4)  [above right of=0DM4]  {$\both$};
    \node (ADM4)  [above left of=0DM4]   {$\neither$};
    \node (1DM4)  [above right of=ADM4]  {$\tvalue$};
    \node (AK3)  [right of=BDM4] {$\neither$};
    \node (0K3)  [below=0.4\distance of AK3]  {$\fvalue$};
    \node (1K3)  [above=0.4\distance of AK3] {$\tvalue$};
    \node (0B2)  [right of=0K3]  {$\fvalue$};
    \node (1B2)  [right of=AK3] {$\tvalue$};
    \node (DM4)  [below=1.5\distance of 0DM4] {$\mathbf{DM}_4$};
    \node (K3)  [below=1.5\distance of 0K3] {$\mathbf{DM}_3$};
    \node (B2)  [below=1.5\distance of 0B2] {$\mathbf{DM}_2$};
    \draw (0DM4)   -- (ADM4);
    \draw (0DM4)   -- (BDM4);
    \draw (BDM4)   -- (1DM4);
    \draw (ADM4)  -- (1DM4);
    \draw (AK3)  -- (0K3);
    \draw (AK3)  -- (1K3);
    \draw (0B2)   -- (1B2);
    \end{tikzpicture}
    \caption{The subdirectly irreducible De Morgan algebras.}
    \label{fig:de_morgan}
    \end{subfigure}
    \hfill
    \begin{subfigure}[b]{.5\textwidth}
    \centering
    \begin{tikzpicture}[node distance=1.25cm]
    \node (0PP6)                      {$\fvalue$};
    \node (BPP6)  [above right of=0PP6]  {$\both$};
    \node (APP6)  [above left of=0PP6]   {$\neither$};
    \node (1PP6)  [above right of=APP6]  {$\tvalue$};
    \node (00PP6)  [below=0.4\distance of 0PP6] {$\efvalue$};
    \node (11PP6)  [above=0.4\distance of 1PP6] {$\etvalue$};
    \node (APP5)  [right of=BPP6] {$\neither$};
    \node (0PP5)  [below=0.4\distance of APP5]  {$\fvalue$};
    \node (1PP5)  [above=0.4\distance of APP5] {$\tvalue$};
    \node (00PP5)  [below=0.4\distance of 0PP5] {$\efvalue$};
    \node (11PP5)  [above=0.4\distance of 1PP5] {$\etvalue$};
    \node (0PP4)  [right of=0PP5]  {$\fvalue$};
    \node (1PP4)  [right of=APP5] {$\tvalue$};
    \node (00PP4)  [below=0.4\distance of 0PP4] {$\efvalue$};
    \node (11PP4)  [above=0.4\distance of 1PP4] {$\etvalue$};
    \node (00PP3)  [right of=00PP4] {$\efvalue$};
    \node (0PP3) [above=0.4\distance of 00PP3] {$\neither$};
    \node (1PP3) [above=0.4\distance of 0PP3] {$\etvalue$};
    \node (PP6)  [below=0.4\distance of 00PP6] {$\mathbf{IS}_6$};
    \node (PP5)  [below=0.4\distance of 00PP5] {$\mathbf{IS}_5$};
    \node (PP4)  [below=0.4\distance of 00PP4] {$\mathbf{IS}_4$};
    \node (PP3)  [below=0.4\distance of 00PP3] {$\mathbf{IS}_3$};
    \node (00PP2)  [right of=00PP3]  {$\efvalue$};
    \node (11PP2)  [above=0.4\distance of 00PP2] {$\etvalue$};
    \node (PP2)  [below=0.4\distance of 00PP2] {$\mathbf{IS}_2$};
    \draw (0PP6)   -- (APP6);
    \draw (0PP6)   -- (BPP6);
    \draw (BPP6)   -- (1PP6);
    \draw (APP6)  -- (1PP6);
    \draw (0PP6)   -- (00PP6);
    \draw (1PP6)   -- (11PP6);
    \draw (APP5)  -- (0PP5);
    \draw (APP5)  -- (1PP5);
    \draw (0PP5)   -- (00PP5);
    \draw (1PP5)   -- (11PP5);
    \draw (0PP4)   -- (1PP4);
    \draw (0PP4)   -- (00PP4);
    \draw (1PP4)   -- (11PP4);
    \draw (0PP3) -- (00PP3);
    \draw (1PP3) -- (0PP3);
    \draw (00PP2) -- (11PP2);
    \end{tikzpicture}
    \caption{The subdirectly irreducible IS-algebras.}
    \label{fig:pp_algebra}
    \end{subfigure}
    \vspace{-.75em}
    \caption{}
\end{figure}
\vspace{-.75em}

\begin{definition}
Given a $\DMnSig$-algebra whose
$\DMSig$-reduct is a De Morgan
algebra, we say that it
constitutes an \emph{involutive Stone algebra} (\emph{IS-algebra}) if it
satisfies the equations:
\setlength{\tabcolsep}{3.5pt}
\begin{tabular}{@{}llll@{}}
    {\rm\textbf{(IS1)}} $\nabla \bot \approx \bot$
    \hspace{5mm}
    &
    {\rm\textbf{(IS2)}} $x \land \nabla x \approx x$
    \hspace{5mm}
    & 
    {\rm\textbf{(IS3)}} $\nabla(x \land y) \approx \nabla x \land \nabla y$
    \hspace{5mm}
    &
    {\rm\textbf{(IS4)}} $\DMNeg\nabla x \land \nabla x \approx \bot$
\end{tabular}
\end{definition}

\begin{example}
\label{ex:issix}
Let $\SixSet \SymbDef \FourSet \cup \{\efvalue,\etvalue\}$
and let $\ISAlg_6 \SymbDef \langle \SixSet, \cdot^{\ISAlg_6} \rangle$
be the $\DMnSig$-algebra
whose lattice structure is depicted in Figure~{\rm \ref{fig:pp_algebra}}
and interprets $\DMNeg$ and $\nabla$ as per the following:
\[
\DMNeg^{\ISAlg_6} a \SymbDef
\begin{cases}
    \DMNeg^{\FDEAlg} a & a \in \FourSet\\
    \efvalue & a = \etvalue\\
    \etvalue & a = \efvalue\\
\end{cases}\qquad
\nabla^{\ISAlg_6} a \SymbDef
\begin{cases}
    \etvalue & a \in \SixSet\setminus\{\efvalue\}\\
    \efvalue & a = \efvalue\\
\end{cases}
\]
\noindent 
The subalgebras of $\ISAlg_6$ exhibited in Figure~{\rm\ref{fig:pp_algebra}} constitute the only subdirectly irreducible IS-algebras~({\rm\cite{cignoli1983}}).
\end{example}
}

The notion of \emph{residuation} in $\Sigma$-algebras
will be essential for us here,
as it is tightly connected to
the intuitionistic implication.
Let $\mathsf{K}$ be a class of 
$\Sigma$-algebras, with $\Sigma \supseteq \LSig$,
whose $\{ \land, \top \}$-reducts are
meet-semilattices with a greatest
element assigned to~$\top$.
Given $\AlgA \in \mathsf{K}$,
we say that $\DefCon$ is the 
\emph{residuum of $\land$ in $\AlgA$}
provided that
$a \,\land^\AlgA\, b \leq c$
if, and only if,
$a \leq b \,\DefCon^\AlgA\, c$
for all $a,b,c \in A$.
It is well-known that
in classical logic and in intuitionistic logic
the implication plays the role of residuum of conjunction.
{We will also consider the notion
of \emph{pseudocomplement of $a \in A$}: this will be defined as
the greatest element $\neg a \in A$
such that $a \land^\AlgA \neg a = \bot^\AlgA$.
{When every element of~$\AlgA$ has
a pseudocomplement (unique by definition), we may think of $\neg^\AlgA$ as a unary operation, known as the \emph{pseudocomplement operation of~$\AlgA$}.}}

{
We move now to the logical preliminaries. 
A
\emph{\SetFmla{} logic (over $\Sigma$)}
is a consequence relation
$\vdash$ on $\LangSet{\Sigma}$
and a \emph{\SetSet{} logic (over $\Sigma$)}
is a generalized consequence relation 
$\sequent$ on $\LangSet{\Sigma}$ (\cite{humberstone:connectives}).
The \emph{\SetFmla{} companion}
of a given \SetSet{} logic~$\sequent$
is the \SetFmla{} logic $\vdash_{\sequent}$ 
such that $\FmSetA \vdash_{\sequent} \FmB$
if, and only if, $\FmSetA \sequent \{\FmB\}$.
We will write
$\FmSetA \; \SymLogEquiv \; \FmSetB$
when $\FmSetA \; \sequent \; \FmSetB$
and $\FmSetB \; \sequent \; \FmSetA$;
that being the case, we say these
sets are \emph{logically equivalent}.
{We will adopt the convention of
omitting curly braces when writing sets of formulas
in statements
involving (generalized) consequence relations.}
The complement of a given \SetSet{} logic
$\sequent$,
{i.e. 
$(\powerset\LangSet{\Sigma} \times\powerset\LangSet{\Sigma}){\setminus}\sequent$},
will be denoted
by~$\notsequentx{}$.

Let $\sequent$ and $\sequent'$ be \SetSet{} logics
over signatures $\Sigma$ and $\Sigma'$, respectively,
with $\Sigma \subseteq \Sigma'$.
We say that
$\sequent'$ \emph{expands}
$\sequent$ when
$\sequent' \supseteq \sequent$.
When $\Sigma = \Sigma'$,
we say more simply that $\sequent'$
\emph{extends} $\sequent$.
It is worth recalling that the collection of all extensions of
a given logic forms a complete lattice
under inclusion.
We say that $\sequent'$ is a
\emph{conservative
expansion} of~$\sequent$
when~$\sequent'$ expands~$\sequent$ and, for all $\FmSetA,\FmSetB \subseteq \LangSet{\Sigma}$, we have
$\FmSetA \sequent' \FmSetB$ if, and only if, $\FmSetA \sequent \FmSetB$.
The \emph{finitary companion}
of a \SetSet{} logic $\sequent$
is the \SetSet{} logic $\sequent_{\mathsf{fin}}$ such that
$\FmSetA \sequent_{\mathsf{fin}} \FmSetB$
if, and only if, there are finite
$\FmSetA' \subseteq \FmSetA$ and
$\FmSetB' \subseteq \FmSetB$
such that $\FmSetA' \sequent \FmSetB'$.
The latter concepts may be straightforwardly adapted to $\SetFmla$ logics.
}

We call special attention to the notion of self-extensionality, as it will play an important role in the
implicative extensions we consider 
in this paper:

\begin{definition}
    \label{def:selfextensional}
{
    A logic
    $\sequent$ over the signature~$\Sigma$ is called \emph{self-extensional} (or \emph{congruential})
    if logical equivalence is compatible with each connective in the signature, that is,  
    for every $\DefCon \in \Sigma_k$, it is the case that 
    $\DefCon(\FmA_1,\ldots,\FmA_k) \, \SymLogEquiv \, \DefCon(\FmB_1,\ldots,\FmB_k)$ whenever $\FmA_i \, \SymLogEquiv \, \FmB_i$ for every $1\leq i \leq k$.
}
    For \SetFmla{} logics,
    just replace `$\sequent$'
    by `$\vdash$'.
\end{definition}

{
A \emph{partial non-deterministic $\Sigma$-matrix} (or, more simply, \emph{$\Sigma$-PNmatrix}) $\DefMat$
is a structure $\langle \mathbf{A}, D \rangle$
where $\mathbf{A}$
is a $\Sigma$-multialgebra
and the members of $D \subseteq A$ are called \emph{designated values}.
We will write $\overline{D}$ to refer to $A{\setminus{}}D$.
When $\AlgA$ is a $\Sigma$-algebra,
$\DefMat$ is called a \emph{$\Sigma$-matrix},
and coincides with the usual definition
of logical matrix in the literature.
A \emph{refinement of $\DefMat$}
is a $\Sigma$-PNmatrix obtained
 from $\AlgA$ by deleting values from some
entries of the interpretations of the
connectives in $\AlgA$
(the resulting interpretations are also
said to be \emph{refinements} of the ones in $\AlgA$).
The mappings in
$\Hom(\LangAlg{\Sigma}, \mathbf{A})$ are called \emph{$\DefMat$-valuations}. 
Every $\Sigma$-PNmatrix
determines a \SetSet{} logic
$\sequent_\DefMat$
such that $\FmSetA \sequent{}_{\DefMat}\; \FmSetB$
if{f} $h(\FmSetA) \cap \overline{D} \neq \varnothing$
or $h(\FmSetB) \cap D \neq \varnothing$
for all $\DefMat$-valuations~$h$,
as well as a \SetFmla{} logic
$\vdash_\DefMat$
with $\FmSetA \vdash_\DefMat \FmB$ iff
$h(\FmSetA) \cap \overline{D} \neq \varnothing$
or $h(\FmB) \in D$
for all $\DefMat$-valuations~$h$
(notice that $\vdash_{\DefMat}$ 
is the \SetFmla{} companion of $\sequent_\DefMat$).
Given a \SetSet{} logic~$\sequent$ (resp.\ a \SetFmla{} logic~$\vdash$), 
if $\sequent\;\subseteq\;\sequent_{\DefMat}$ (resp.\ $\vdash \;\subseteq\; \vdash_{\DefMat}$), we shall say that~$\DefMat$ \emph{is a matrix model of}
$\sequent$ (resp.\ $\vdash$),
and if the converse also holds
we shall say that~$\DefMat$ \emph{determines}
$\sequent$ (resp.\ $\vdash$).
The $\SetSet{}$ (resp.\ $\SetFmla$) logic \emph{determined} by a class $\mathcal{M}$
of $\Sigma$-matrices is given
by $\bigcap \{\sequent_\DefMat : \DefMat \in \mathcal{M}\}$
(resp.\ $\bigcap \{\vdash_\DefMat : \DefMat \in \mathcal{M}\}$).

\begin{example}
\label{ex:fde}
The $\DMSig$-matrix $\langle \FDEAlg, \{ \both,\tvalue \} \rangle$ determines the logic known as
the 4-valued Dunn-Belnap logic ({\rm\cite{belnap1977}}), or First-Degree Entailment (FDE),
which we hereby denote by~$\FDELogic$.
Extensions of~$\FDELogic$ 
are known as \emph{super-Belnap logics}  ({\rm\cite{rivieccio2012}}).
\end{example}

\begin{example}
Classical Logic, henceforth denoted by $\CL$, is determined by the
$\DMSig$-matrix $\langle \BoolAlg, \{\tvalue\} \rangle$
(see Figure~\ref{fig:de_morgan}).
\end{example}

Given a $\Sigma$-matrix
$\DefMat = \langle \mathbf{A}, D \rangle$,
a congruence $\theta \in \ConSet \mathbf{A}$
is said to be \emph{compatible with} $\DefMat$
when $b \in D$ whenever both $a \in D$
and $a \theta b$, for all $a,b \in A$.
We denote by $\LeibCong{\DefMat}$
the \emph{Leibniz congruence}
associated to $\DefMat$,
namely
the greatest congruence of 
$\mathbf{A}$ compatible with $\DefMat$.
The matrix $\Reduce{\DefMat} = \langle \mathbf{A}/\LeibCong{\DefMat}, D/\LeibCong{\DefMat} \rangle$
is the \emph{reduced version} of $\DefMat$.
{It is well known that
$\sequent_\DefMat = \sequent_{\Reduce{\DefMat}}$
(and thus $\vdash_\DefMat \;=\; \vdash_{\Reduce{\DefMat}}$)
and, since every logic is determined by a class of matrix models, we have that every logic coincides
with the logic determined by
its reduced matrix models.
As a shortcut, we call a matrix \emph{reduced} when it coincides with its own reduced version
(or, equivalently, when its Leibniz
congruence is the identity relation
on~$A$).

When a $\Sigma$-algebra $\mathbf{A}$ has a lattice structure with underlying order~$\leq$, for any $a\in A$ we write $\Upset{a}$ to refer to the set $\{b \in A : a \leq b\}$.
{For instance, 
over $\ISSix{}$ we may consider the
set
$\Upset{\both} = \{\both, \tvalue, \etvalue\}$ (see~Figure~\ref{fig:pp_algebra}).}
We may employ the same notation to
write the designated set of
a $\Sigma$-matrix having a lattice
as underlying algebra.
A \emph{lattice filter}
of a meet-semilattice~$\mathbf{A}$ with a top element~$\top^{\AlgA}$ is a subset $D \subseteq A$ with $\top^\mathbf{A} \in D$
and closed under~$\land^\mathbf{A}$;
moreover,
$D$ is a \emph{proper}
lattice filter of $\mathbf{A}$ when
$D \neq A$.
A~\emph{principal filter} of $\AlgA$
is a lattice filter of the form
$\up a$, for some $a \in A$.
Note that, if $\AlgA$ is finite,
every lattice filter is principal.
If $\mathbf{A}$ also has a join-semilattice structure,
a \emph{prime filter} of $\mathbf{A}$
is a proper lattice filter $D$ of $\mathbf{A}$
such that $a \lor b \in D$ iff $a \in D$
or $b \in D$, for all $a,b \in A$.
}

\subsection{Logics associated to classes of (ordered) algebras}
\label{sec:ordered-algebras-logics}

Throughout this section,
consider a propositional signature $\Sigma \supseteq \LSig$.
In addition, for
a set of $\Sigma$-formulas
$\FmSetA \SymbDef \{ \FmA_1,\ldots,\FmA_n \}$,
let $\bigwedge \FmSetA \SymbDef \FmA_1 \land \ldots \land \FmA_n$
and
$\bigvee \FmSetA \SymbDef \FmA_1 \lor \ldots \lor \FmA_n$, while, by convention, let
$\bigwedge \varnothing \SymbDef \top$ 
and
$\bigvee \varnothing \SymbDef \bot$.
Moreover, let the \emph{inequality} $\FmA \leq \FmB$ abbreviate
the equation $\FmA \approx \FmA \land \FmB$.

\begin{definition}
\label{def:order-pres-set-set-logic}
Let $\mathsf{K}$ be a 
class of\;$\Sigma$-algebras
such that each $\mathbf{A} \in \mathsf{K}$ has a bounded distributive lattice reduct with greatest element $\top^\mathbf{A}$ 
and a least element $\bot^\mathbf{A}$. 
The \emph{order-preserving logic associated to $\mathsf{K}$}, denoted by $\IneqPresMC{\mathsf{K}}$,
is such that
$\FmSetA \, \IneqPresMC{\mathsf{K}} \, \FmSetB$
 if, and only if,
there are finite 
$\FmSetA'\subseteq\FmSetA$ and
$\FmSetB'\subseteq \FmSetB$
such that
$\bigwedge \FmSetA' \leq 
\bigvee \FmSetB'$
is valid in $\mathsf{K}$.
\end{definition}


For any order-preserving logic $\IneqPresMC{\mathsf{K}}$ associated to an appropriate class $\mathsf{K}$ of algebras, the following
characterization in terms of
prime filters applies:

\begin{proposition}
    \label{prop:mat-char-prime}
    $\IneqPresMC{\mathsf{K}}$ is
    the finitary companion of the logic
    determined
    by
    $
    \mathcal{M^{\Upset\lor}} \SymbDef
    \{ \langle \mathbf{A}, D \rangle :  
        \AlgA \in \mathsf{K}\text{ and }
        D
        \text{ is a prime filter of $\AlgA$}\}$.
\end{proposition}
\begin{proof}
    To start with, assume that 
    $\FmSetA \IneqPresMC{\mathsf{K}}\FmSetB$.
    So, (i) there are finite 
$\FmSetA'\subseteq\FmSetA$ and
$\FmSetB'\subseteq \FmSetB$ such that 
    $\FmSetA'\IneqPresMC{\mathsf{K}}\FmSetB'$.
    Let 
    $\langle \mathbf{A}, D \rangle
    \in \mathcal{M^{\Upset\lor}}$
    and $v \in \Hom(\LangAlg{\Sigma}, \mathbf{A})$ be a valuation.
    Assume that
    $v[\FmSetA'] \subseteq D$.
    Thus
    $v(\bigwedge \FmSetA') \in D$,
    as $D$ is a lattice filter.
    By (i), we know that
    $v(\bigwedge \FmSetA') \leq v(\bigvee \FmSetB')$,
    so $v(\bigvee \FmSetB') \in D$.
    Since $D$ is prime, we must have
    $v(\FmB) \in D$
    for some $\FmB \in \FmSetB'$
    and we are done.

    Conversely,
    assume 
    that
    (ii)
    $\FmSetA \sequent_{\mathcal{M}^{\Upset\lor}}
    \FmSetB$.
    This 
    consequence relation is
    finitary, so 
    $\FmSetA' \sequent_{\mathcal{M^{\Upset\lor}}}
    \FmSetB'$
    for finite 
$\FmSetA'\subseteq\FmSetA$ and
$\FmSetB'\subseteq \FmSetB$.
Let $\mathbf A \in \mathsf{K}$
and $v \in \Hom(\LangAlg{\Sigma}, \mathbf{A})$.
Let $a \SymbDef v(\bigwedge \FmSetA')$
(note that $a = \top^\AlgA$ if $\FmSetA' = \varnothing$).
If $a = \bot^\AlgA$, the result
is straightforward.
Otherwise, $\Upset a$ is a proper filter.
Suppose, by reductio, that
$a \not\leq v(\bigvee \FmSetB')$,
that is,
$v(\bigvee \FmSetB') \not\in \Upset a$.
Since $\AlgA$ is distributive,
consider, by the Prime Filter Theorem, the extension of $\Upset a$
to a prime filter 
$D \supseteq \Upset a$
such that $v(\bigvee \FmSetB') \not\in D$.
Since $a \leq v(\FmA)$
for each $\FmA \in \FmSetA'$,
we have $v(\FmA) \in \Upset a \subseteq D$,
for all $\FmA \in \FmSetA'$
and, by (ii),
we have 
$v(\FmB) \in D$
for some $\FmB \in \FmSetB'$.
Since $v(\FmB) \leq v(\bigvee \FmSetB')$,
we must have $v(\bigvee \FmSetB') \in D$.
As this leads to a contradiction, we conclude that 
$a \leq v(\bigvee \FmSetB')$,
as desired.\qedhere
\end{proof}

For a variety generated by a single finite distributive lattice, we have this simpler
characterization in terms of a finite family of finite matrices:

\begin{proposition}
\label{prop:single-gen-variety-matrices}
Let $\mathsf{K} = \Variety(\{ \AlgB\})$, for $\AlgB$ a finite distributive lattice. Then
$\IneqPresMC{\mathsf{K}}$ is
    determined
    by
    $
    \mathcal{M}^{\Upset\lor}_{\mathsf{fin}} \SymbDef
    \{ \langle \AlgB, D \rangle :  
        D
        \text{ is a prime filter of $\AlgB$}\}$.
\end{proposition}
\begin{proof}
%
%

As
$\mathcal{M}^{\Upset{}\lor}_{\mathsf{fin}} \subseteq
\mathcal{M}^{\Upset{}\lor}$,
we have $\IneqPresMC{\mathsf{K}}
\subseteq \sequent_{\mathcal{M}^{\Upset{}\lor}_{\mathsf{fin}}}$.
Conversely,
suppose that
$\FmSetA \; \notsequent_{\mathsf{K}}^{\leq} \; \FmSetB$.
Consider finite $\FmSetA' \subseteq \FmSetA$
and $\FmSetB' \subseteq \FmSetB$.
We want to show that
$\FmSetA' \notsequent_{\mathcal{M}^{\Upset{}\lor}_{\mathsf{fin}}} \FmSetB'$.
By the hypothesis, we know that
$\bigwedge \FmSetA' \leq \bigvee \FmSetB'$
fails in $\mathsf{K}$,
thus it fails in $\AlgB$,
say under a valuation $v$.
Let $b \SymbDef v(\bigwedge \FmSetA')$.
So, $b \not\leq v(\bigvee \FmSetB')$.
Note that $\Upset b$ must be proper.
As $v(\bigvee \FmSetB') \not\in \Upset b$,
consider the extension of $\Upset b$
to a prime filter $D$,
by the Prime Filter Theorem, such that
$v(\bigvee \FmSetB') \not\in D$.
Since $\AlgB$ is finite, we have
that $D$ is principal.
Clearly, we must have $v(\FmB) \not\in D$
for every $\FmB \in \FmSetB'$, as
$v(\FmB) \leq v(\bigvee \FmSetB')$
and $D$ is upwards closed,
and 
{
$v(\FmA) \in D$ for every $\FmA\in\FmSetA'$.
}
Therefore, 
$\FmSetA' \notsequent_{\mathcal{M}^{\Upset{}\lor}_{\mathsf{fin}}} \FmSetB'$,
as desired.
\end{proof}


\begin{definition}
\label{def:order-pres-set-fmla-logic}
Let $\mathsf{K}$ be a class of $\Sigma$-algebras as in Definition~\ref{def:order-pres-set-set-logic}.
The \SetFmla{} \emph{order-preserving logic induced by~$\mathsf{K}$}, which we denote by $\vdash_\mathsf{K}^{\leq}$,
is such that
$\FmSetA \vdash_\mathsf{K}^{\leq} \FmB$ if, and only if,
there
are $\FmA_1,\ldots,\FmA_n \subseteq \FmSetA$ ($n \geq 1$)
for which
$\bigwedge_i \FmA_i \leq \FmB$
is valid in~$\mathsf{K}$.
\end{definition}

\noindent
Notice that
$\vdash_\mathsf{K}^{\leq}$
is the \SetFmla{} companion of
$\IneqPresMC{\mathsf{K}}$.

The above logics are particularly interesting for
us in view of the following:

\begin{proposition}
\label{prop:order-are-selfext}
    $\IneqPresMC{\mathsf{K}}$
    and
    $\vdash_\mathsf{K}^{\leq}$
    are self-extensional.
\end{proposition}
\begin{proof}
    Directly from \citet*[Sec. 3]{Jansana2006}
    and
    from the fact that 
    $\vdash_\mathsf{K}^{\leq}$
    is the 
     $\SetFmla{}$
    companion
    of
    $\IneqPresMC{\mathsf{K}}$.
\end{proof}

We close by defining
the $\top$-assertional
logics (also known as 1-assertional logics) associated to a class
of bounded lattices.

\begin{definition}
    Let
$\mathsf{K}$ be a class of 
algebras with a bounded lattice reduct.
The \emph{$\top$-assertional logics}
$\sequent^\top_\mathsf{K}$ and $\OneAssert_{\mathsf{K}}$ correspond respectively to the \SetSet{} and \SetFmla{} logics determined by the class of matrices $\{\langle \mathbf{A}, \{\top^\mathbf{A}\} \rangle : \mathbf{A} \in \mathsf{K}\}$. 
\end{definition}

\noindent
Notice that $\OneAssert_{\mathsf{K}}$ is the \SetFmla{} companion of $\sequent^\top_\mathsf{K}$.


\subsection{Hilbert-style axiomatizations}

Based on \cite{shoesmith_smiley:1978} and \cite{marcelinowollic},
we define a \emph{symme\-trical} (\emph{Hilbert-style}) \emph{calculus} $\CalcVar$ (or \emph{\SetSet{} calculus}, for short)
as a collection
of pairs $(\FmSetA,\FmSetB) \in \powerset{\LangSet{\Sigma}}\times\powerset{\LangSet{\Sigma}}$, denoted by $\inferx[]{\FmSetA}{\FmSetB}$ and called (\emph{symmetrical} or \SetSet{}) \emph{inference rules}, where
$\FmSetA$ is the
\emph{antecedent} and $\FmSetB$
is the \emph{succedent} of the said rule.
We will adopt the convention of
omitting curly braces when writing sets of formulas
and leaving a blank space instead of writing $\varnothing$
when presenting inference rules and statements
involving (generalized) consequence relations.
We proceed to define what constitutes a proof
in such calculi.

A \emph{{finite} bounded rooted tree} $\DefTree$ is a finite poset
$\langle \NodeSet(\DefTree), \TreeRel^\DefTree \rangle$
with a single minimal element $\Root{\DefTree}$, 
the \emph{root} of~$\DefTree$,
such that, for each \emph{node} $\DefNode \in \NodeSet(\DefTree)$, the set $\{\DefNode' \in \NodeSet(\DefTree) : \DefNode' \TreeRel^\DefTree \DefNode \}$ of \emph{ancestors} of $\DefNode$ (or the \emph{branch} up to $\DefNode$) is
{linearly ordered} under $\TreeRel^\DefTree$,
and every branch of $\DefTree$ has
a maximal element (a \emph{leaf} of $\DefTree$).
We may assign a label $\Label{\DefTree}(\DefNode) \in \powerset{\LangSet{\Sigma}} \cup \{\Disc\}$ to each node $\DefNode$ of $\DefTree$, in which case $\DefTree$ is said to
be \emph{labelled}.
Given $\FmSetB \subseteq \LangSet{\Sigma}$, a leaf $\DefNode$ is \emph{$\FmSetB$-closed} in $\DefTree$ when
$\Label{\DefTree}(\DefNode) \;=\; \Disc$ or $\Label{\DefTree}(\DefNode) \cap \FmSetB \neq \varnothing$. The tree $\DefTree$ itself is
\emph{$\FmSetB$-closed} when all of its leaves are
$\FmSetB$-closed. The immediate successors of
a node $\DefNode$ with respect to $\TreeRel^\DefTree$ are called the \emph{children}
of $\DefNode$ in $\DefTree$.

Let $\mathsf{R}$ be a symmetrical calculus. 
An \emph{$\mathsf{R}$-derivation} is a 
{finite}
labelled bounded rooted tree 
such that for every non-leaf node $\DefNode$ of $\DefTree$
there exists a rule of inference $\mathsf{r} = \inferx[]{\FmSetC}{\FmSetD} \in \mathsf{R}$ and a substitution $\sigma$
such that $\sigma(\FmSetC) \subseteq \Label{\DefTree}(\DefNode)$, and the set of children of $\DefNode$ is 
either (i) 
$\{\DefNode^\Fm : \Fm \in \sigma(\FmSetD)\}$, in case $\FmSetD \neq \varnothing$, where $\DefNode^\Fm$ is a node labelled with $\Label{\DefTree}(\DefNode) \cup \{\Fm\}$, 
or
{
    (ii) a singleton $\{\DefNode^\Disc\}$
    with $\Label{\DefTree}(\DefNode^\Disc)\;=\;\Disc$,
    in case $\FmSetD\;=\;\varnothing$.
}
{Such node $\DefNode$ is said to be
\emph{expanded} by an application of $\mathsf r$.
Figure~\ref{fig:derivationscheme} illustrates
the general shape of $\CalcVar$-derivations.}
We say that
$\FmSetA\; \sequent_{\CalcVar}\; \FmSetB$
whenever there is a $\FmSetB$-closed derivation $\DefTree$ such that $\FmSetA\;\supseteq\;\Root{\DefTree}$; 
such a tree consists in a \emph{proof}
that $\FmSetB$ follows from~$\FmSetA$ in~$\mathsf{R}$.
As a matter of simplification when drawing such trees, we usually avoid copying the formulas inherited from the parent nodes (see Example~\ref{ex:fdeax} below, where we display some proof trees).
The relation $\sequent_{\CalcVar}$ so defined is
a \SetSet{} logic
and, when $\sequent_{\CalcVar} = \sequent$,
we say that $\CalcVar$
\emph{axiomatizes} $\sequent$.
{It is worth emphasizing that we restrict ourselves to finite derivations
because we are interested in finitary
\SetSet{} logics, but in general
they could be infinite and cover
also non-finitary logics (cf.~\citet[Fig. 1]{marcelino:2019} for an example)}.
A rule $\inferx[]{\FmSetA}{\FmSetB}$ is \emph{sound} in a PNmatrix $\DefMat$ when $\FmSetA \;\sequent_\DefMat\;\FmSetB$.
It should be pointed out
that this proof formalism
generalises the conventional ($\SetFmla{}$) Hilbert-style calculi:
the latter corresponds to symmetrical calculi
whose rules have, each, a finite antecedent
and a singleton as succedent.

{
\begin{remark}\label{rem:no-loops}
    Observe that any application
    of a rule of inference
    that creates a node labelled with the
    same set of formulas as the label of the node being expanded
    is superfluous (cf.~\cite[Prop. 35]{greati2022}); in this way,
    loops in these derivations
    may be readily averted.
\end{remark}
}

{The process of building an $\CalcVar$-derivation, in principle, does not necessarily terminate, since there may be
infinitely-many rule instances that
are applicable.
We now introduce a notion of
analyticity that allows a limit to be set on
the search space in such a way that this process is guaranteed to finish.}
Given $\Lambda \subseteq \LangSet{\Sigma}$,
we write $\FmSetA \sequent_{\mathsf{R}}^\Lambda \FmSetB$
whenever there is a proof of $\FmSetB$ from $\FmSetA$
using only formulas in $\Lambda$.
Given $\FmSetD,\FmSetAnalytic\subseteq \LangSet{\Sigma}$, we define the set $\Upsilon^\FmSetAnalytic$ of \emph{$\FmSetAnalytic$-generalized subformulas of~$\FmSetD$} as the set $\mathsf{sub}(\FmSetD) \cup \{ \sigma(\FmA) : \FmA \in \FmSetAnalytic\mbox{ and } \sigma : P \to \mathsf{sub}(\FmSetD)\}$.
We say that
$\mathsf{R}$ is \emph{$\FmSetAnalytic$-analytic}
when, for all $\FmSetA,\FmSetB \subseteq \LangSet{\Sigma}$,
we always have
$\FmSetA \sequent_\mathsf{R}^{\Upsilon^\FmSetAnalytic} \FmSetB$ whenever we have $\FmSetA\;\sequent_\mathsf{R}\;\FmSetB$; 
intuitively, it means that
proofs in $\mathsf{R}$
that $\FmSetB$ follows from $\FmSetA$, whenever they exist,
do not need to use more material than subformulas of $\FmSetA\cup\FmSetB$
or substitution instances of the
formulas in~$\FmSetAnalytic$ built from the subformulas of $\FmSetA\cup\FmSetB$.
{Searching for proofs in a finite \emph{$\FmSetAnalytic$-analytic} calculus
is then guaranteed to terminate because
(i) due to analyticity, there are finitely-many
rule instances to be considered
when expanding a node
and (ii) in view of Remark~\ref{rem:no-loops}, loops may always be averted.}

A general method is introduced in \cite{marcelinowollic, marcelino:2019} for obtaining analytic calculi (in the sense of analyticity introduced in the previous paragraph)
for logics given by 
a $\Sigma$-PNmatrix $\langle \mathbf{A}, D\rangle$ whenever a certain expressiveness requirement (called `monadicity' in \cite{shoesmith_smiley:1978}) is met: for every $a,b \in A$, there is a single-variable formula
$\mathrm{S}$ such that $\mathrm{S}^\mathbf{A}(a) \in D$ and $\mathrm{S}^\mathbf{A}(b) \not\in D$
or vice-versa. 
We call such formula a \emph{separator (for $a$ and $b$)}.

The following example illustrates a symmetrical
calculus for $\FDELogic$ generated by this method, as well
as some proofs in this calculus.

\begin{figure}[t]
    \centering
    \begin{tikzpicture}[every tree node/.style={},
       level distance=1.2cm,sibling distance=1cm,
       edge from parent path={(\tikzparentnode) -- (\tikzchildnode)}, baseline]
        \Tree[.\node[style={draw,circle}] {};
            \edge[dashed];
            [.\node[style={}] {$\FmSetA$};
            \edge node[auto=right] {$\frac{\FmSetC}{\varnothing}{}$};
            [.{$\ast$}
            ]
            ]
        ]
    \end{tikzpicture}
    \qquad
    \begin{tikzpicture}[every tree node/.style={},
       level distance=1cm,sibling distance=.5cm,
       edge from parent path={(\tikzparentnode) -- (\tikzchildnode)}, baseline]
        \Tree[.\node[style={draw,circle}] {};
            \edge[dashed];
            [.\node[style={}] {$\FmSetA$};
            \edge node[auto=right] {$\frac{\FmSetC}{\FmD_1,\FmD_2,\ldots,\FmD_n}{}$};
            [.${\color{gray}\FmSetA},\FmD_1$
                \edge[dashed];
                [.\node[style={draw,circle}]{};
                ]
                \edge[dashed];
                [.\node[style={}]{$\cdots$};
                ]
                \edge[dashed];
                [.\node[style={draw,circle}]{};
                ]
            ]
            [.${\color{gray}\FmSetA},\FmD_2$
                \edge[dashed];
                [.\node[style={draw,circle}]{};
                ]
                \edge[dashed];
                [.\node[style={}]{$\cdots$};
                ]
                \edge[dashed];
                [.\node[style={draw,circle}]{};
                ]
            ]
            [.$\ldots$
            ]
            [.${\color{gray}\FmSetA},\FmD_n$
                \edge[dashed];
                [.\node[style={draw,circle}]{};
                ]
                \edge[dashed];
                [.\node[style={}]{$\cdots$};
                ]
                \edge[dashed];
                [.\node[style={draw,circle}]{};
                ]
            ]
            ]
        ]
    \end{tikzpicture}
    \caption{Graphical representation of $\CalcVar$-derivations,
    where $\CalcVar$ is a
    $\SetSet{}$ system.
    The dashed edges and blank circles represent other branches that may exist in the derivation.
    We usually omit the formulas inherited from the parent
    node, 
    exhibiting only the ones introduced by the
    applied rule of inference.
    Recall that, in both cases,
    we must have $\FmSetC \subseteq \FmSetA$.}
    \label{fig:derivationscheme}
\end{figure}
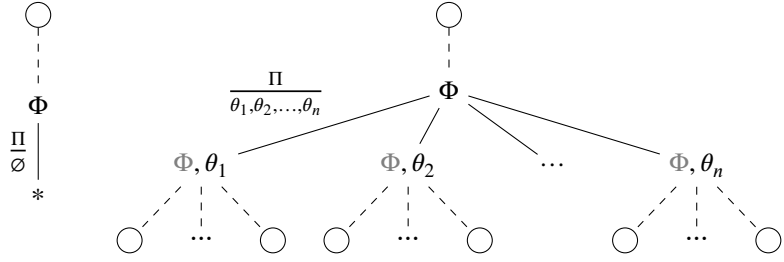

\begin{example}
\label{ex:fdeax}
The matrix $\langle \mathbf{DM}_4, {\uparrow}\both \rangle$ fulfills the above mentiond expressiveness requirement, with the following set of separators: $\mathcal{S} \SymbDef \{p, \dmneg p\}$. 
We may therefore apply the method introduced in {\rm\cite{marcelino:2019}} to obtain for $\mathcal{B}$ the following $\mathcal{S}$-analytic
axiomatization we call~$\mathsf{R}_\mathcal{B}$:
\vspace{-.5em}
$$
\inferx[\mathsf{r}_1]
{}
{\top}
\quad
\inferx[\mathsf{r}_2]
{\dmneg \top}
{}
\quad
\inferx[\mathsf{r}_3]
{}
{\dmneg \bot}
\quad
\inferx[\mathsf{r}_4]
{\bot}
{}
\quad
\inferx[\mathsf{r}_5]
{p}
{\dmneg \dmneg p}
\quad
\inferx[\mathsf{r}_6]
{\dmneg \dmneg p}
{p}
$$
$$
\inferx[\mathsf{r}_7]
{p \land q}
{p}
\quad
\inferx[\mathsf{r}_8]
{p \land q}
{q}
\quad
\inferx[\mathsf{r}_9]
{p, q}
{p \land q}
\quad
\inferx[\mathsf{r}_{10}]
{\dmneg p}
{\dmneg(p \land q)}
\quad
\inferx[\mathsf{r}_{11}]
{\dmneg q}
{\dmneg(p \land q)}
\quad
\inferx[\mathsf{r}_{12}]
{\dmneg (p \land q)}
{\dmneg p, \dmneg q}
$$
$$
\inferx[\mathsf{r}_{13}]
{p}
{p \lor q}
\quad
\inferx[\mathsf{r}_{14}]
{q}
{p \lor q}
\quad
\inferx[\mathsf{r}_{15}]
{p \lor q}
{p, q}
\quad
\inferx[\mathsf{r}_{16}]
{\dmneg p, \dmneg q}
{\dmneg(p \lor q)}
\quad
\inferx[\mathsf{r}_{17}]
{\dmneg(p \lor q)}
{\dmneg p}
\quad
\inferx[\mathsf{r}_{18}]
{\dmneg (p \lor q)}
{\dmneg q}
$$

\noindent Figure {\rm \ref{fig:derivations}} illustrates proofs in $\mathsf{R}_\mathcal{B}$.

\vspace{-.5em}
\begin{figure}[H]
    \centering
    \scalebox{0.9}{
    \begin{tikzpicture}[node distance=1.5cm]
    \node (0)                     {$\dmneg (p \land q)$};
    \node (11) [below left of=0]  {$\dmneg p$};
    \node (12) [below right of=0] {$\dmneg q$};
    \node (21) [below of=11]      {$\dmneg p \lor \dmneg q$};
    \node (22) [below of=12]      {$\dmneg p \lor \dmneg q$};
    \draw (0)  -- (11) node[midway,above left] {$\mathsf{r}_{12}$};
    \draw (0)  -- (12);
    \draw (11) -- (21) node[midway,left] {$\mathsf{r}_{13}$};
    \draw (12) -- (22) node[midway,right] {$\mathsf{r}_{14}$};
    \end{tikzpicture}
    \begin{tikzpicture}[node distance=1.5cm]
    \node (0)                     {$\dmneg p \lor \dmneg q$};
    \node (11) [below left of=0]  {$\dmneg p$};
    \node (12) [below right of=0] {$\dmneg q$};
    \node (21) [below of=11]      {$\dmneg (p \land q)$};
    \node (22) [below of=12]      {$\dmneg (p \land q)$};
    \draw (0)  -- (11) node[midway,above left] {$\mathsf{r}_{15}$};
    \draw (0)  -- (12);
    \draw (11) -- (21) node[midway,left] {$\mathsf{r}_{10}$};
    \draw (12) -- (22) node[midway,right] {$\mathsf{r}_{11}$};
    \end{tikzpicture}
    \qquad \qquad
    \begin{tikzpicture}[node distance=1.5cm]
    \node (0)                     {$p \lor \bot$};
    \node (11) [below left of=0]  {$p$};
    \node (12) [below right of=0] {$\bot$};
    \node (22) [below of=12]      {*};
    \draw (0)  -- (11) node[midway,above left] {$\mathsf{r}_{15}$};
    \draw (0)  -- (12);
    \draw (12) -- (22) node[midway,right] {$\mathsf{r}_{4}$};
    \node (11) [right of=12]               {$p \lor \bot$};
    \node (00) [above=0.5\distance of 11] {$p, r$};
    \node (22) [below of=11]               {};
    \draw (00) -- (11) node[midway,left] {$\mathsf{r}_{13}$};
    \end{tikzpicture}}
    \caption{Proofs in $\mathsf{R}_\FDELogic$ witnessing that $\dmneg(p \land q) \;\SymLogEquiv_{\FDELogic}\; \dmneg p \lor \dmneg q$ and $p \lor \bot \; \SymLogEquiv_{\FDELogic} \; p, r$.}
    \label{fig:derivations}
\end{figure}
\end{example}

\begin{remark}
\label{rem:why-not-sequent-calculi}
\cite{Avron2007}
employed essentially the same 
sufficient expressiveness condition 
{
---further discussed in \cite{cal:mar:JMVLSC}---
}
as the one described
above to provide cut-free sequent calculi
for logics determined by finite non-deterministic
matrices.
In this work we focus, instead, on Hilbert-style systems
for their proximity with the 
logics being axiomatized (rules are simply consequence statements involving an antecedent and a succedent), what facilitates 
the algebraic study thereof. 
On what concerns the underlying decision procedures, we emphasize
that analytic \SetSet{} calculi
admit straightforward (in general exponential-time) proof-search procedures.
\end{remark}

\subsection{Implication connectives
and criteria for implicative expansions}
\label{ss:impcon}

In the present study, we will follow some
very specific research
paths with the goal of adding
an implication to
given logics; we will refer to the 
resulting logics as
\emph{implicative expansions}.
First of all, it is important to define precisely what
we mean by an implication connective on \SetSet{}
and \SetFmla{}
logics.
Throughout this subsection, let $\Sigma$
be an arbitrary signature
with a binary connective $\lor$
and denote by $\Sigma_{\GenImp}$ 
the signature resulting from adding to
$\Sigma$ the binary connective $\GenImp$.
In what follows, let
$\bigvee \{ \FmB_1,\ldots,\FmB_m \} \SymbDef \FmB_1 \ArbDisj (\FmB_2 \ArbDisj \ldots (\ldots \ArbDisj \FmB_n)\ldots)$.


\begin{definition}
\label{def:implication}
The connective $\GenImp$ is
an \emph{implication in a \SetSet{} logic $\sequent{}$ over
$\Sigma_\GenImp$
} 
provided, 
for all
$\FmSetA, \{ \FmB_1,\ldots,\FmB_m \}, \{ \FmA,\FmB \} \subseteq \LangSet{\Sigma}$,
\[
\FmSetA, \FmA \sequent 
\FmB, \FmB_1,\ldots,\FmB_m
\text{ if, and only if, }
\FmSetA \sequent \FmA \GenImp \left(\bigvee \{\FmB, \FmB_1,\ldots,\FmB_m \}\right).
\]
\end{definition}

\noindent 
Note that the above definition reduces in \SetFmla{} to
the \emph{deduction-detachment theorem} ({DDT}) that holds,
for example, in intuitionistic and classical logics,
and this is indeed what we will take to be an implication
in \SetFmla{}, in this paper. 
{That is, $\vdash$
has the 
DDT with respect to $\Imp$
in case
$\FmSetA, \FmA \vdash \FmB$
if, and only if,
$\FmSetA \vdash \FmA \GenImp \FmB$
for all $\FmSetA, \{ \FmA,\FmB \} \subseteq \LangSet{\Sigma}$.}
What we did above was a 
convenient 
generalization of the DDT 
for \SetSet{} logics, using it to
abstractly characterize what we expect of an implication
connective (namely, an internalization of the
consequence relation 
{that allows formulas to be discharged from the antecedent}%
). 

Based on the behaviour of
implication in classical logic,
we also consider the following
stronger notion of implication:

\begin{definition}
\label{def:classicalimplication}
The connective $\GenImp$ is
a \emph{classic-like implication in a \SetSet{} logic $\sequent{}$ over
$\Sigma_\GenImp$} 
provided, for all
$\FmSetA, \FmSetB, \{ \FmA,\FmB \} \subseteq \LangSet{\Sigma}$,
\[
\FmSetA, \FmA \sequent \FmB,
\FmSetB
\text{ if, and only if, }
\FmSetA \sequent 
 \FmA \GenImp \FmB, \FmSetB.
\]
\end{definition}
\noindent 
{
It should be clear enough that, for logics $\sequent$ with $\lor$
satisfying $p \sequent p \lor q$; $q \sequent p \lor q$;
and $p \lor q \sequent p,q$, classic-like implications are implications in the sense of Definition~\ref{def:implication}.
The converse, however, might not
hold (see Remark~\ref{rem:just-imp-def}).
}

In this paper, we shall require two minimal criteria 
on implicative expansions:

\begin{description}
    \item[$\CondIOne$] The connective being added
must qualify as an implication (Definition~\ref{def:implication}); 
    \item[$\CondITwo$] The expansion must be conservative.
\end{description}

We will prove below some facts regarding the
connections among the
notions of implication,
residuation and characterizability via
a single PNmatrix.

\begin{proposition}
\label{fact:residuated-is-imp}
    Let $\GenImp^{\mathbf{A}}$ be 
    the residuum of $\land^\mathbf{A}$
    in each $\mathbf{A} \in \mathsf{K}$,
    where $\mathsf{K}$ is a class of $\Sigma_\GenImp$-algebras, with $\Sigma_\GenImp \supseteq \LSig$,
whose $\{ \land, \top \}$-reducts are
$\land$-semilattices with a greatest
element assigned to $\top$.
    Then $\GenImp$ is an
    implication in
    $\IneqPresMC{\mathsf{K}}$
    (recall Definition~\ref{def:order-pres-set-set-logic}).
\end{proposition}
\begin{proof}
    Let
$\FmSetA, \FmSetB, \{ \FmA \} \subseteq \LangSet{\Sigma}$ with $\FmSetB \neq \varnothing$
finite.
From the left to the right,
    suppose that
    $\FmSetA, \FmA \IneqPresMC{\mathsf{K}}
\FmSetB$.
Then $\FmSetA' \IneqPresMC{\mathsf{K}}
\FmSetB'$, for finite 
$\FmSetA' \subseteq \FmSetA \cup \{ \FmA \}$ and
$\FmSetB' \subseteq \FmSetB$,
and thus ($\star$)
$\FmSetA'',\FmA \IneqPresMC{\mathsf{K}}
\FmSetB'$, for $\FmSetA' \cup \{ \FmA \} = \FmSetA'' \cup \{ \FmA \}$ and $\FmA \not\in \FmSetA''$.
Let $\AlgA \in \mathsf{K}$.
By ($\star$), we have that $\bigwedge \FmSetA'' \land \FmA \leq \bigvee \FmSetB'$ is valid in
$\AlgA$, thus
$\bigwedge \FmSetA'' \land \FmA \leq \bigvee \FmSetB$ is valid in $\AlgA$; hence, by residuation,
$\bigwedge \FmSetA'' \leq \FmA \GenImp \bigvee \FmSetB$ is valid in $\AlgA$,
thus $\FmSetA \IneqPresMC{\mathsf{K}} \FmA \GenImp \bigvee \FmSetB$.
From the right to the left,
suppose that
$\FmSetA \IneqPresMC{\mathsf{K}} \FmA \GenImp \bigvee \FmSetB$.
Then, by definition, we have that
$\FmSetA' \IneqPresMC{\mathsf{K}} \FmA \GenImp \bigvee \FmSetB$,
for some finite $\FmSetA' \subseteq \FmSetA$.
Let $\AlgA \in \mathsf{K}$.
We have that
$\bigwedge \FmSetA' \leq \FmA \GenImp \bigvee \FmSetB$
is valid in $\AlgA$. By residuation,
we have $\bigwedge \FmSetA' \land \FmA \leq \bigvee \FmSetB$ also valid in $\AlgA$,
from which we easily obtain that
$\FmSetA, \FmA \IneqPresMC{\mathsf{K}} \FmSetB$.
\end{proof}

\noindent 
Observe that the above result extends to $\vdash_\mathsf{K}^{\leq}$ (recall Definition~\ref{def:order-pres-set-fmla-logic}) since it is the
\SetFmla{} companion of
    $\IneqPresMC{\mathsf{K}}$.
This also looks like the right moment to introduce intuitionistic-like implications:

{
\begin{definition}
    A \emph{Heyting implication}
    in an algebra $\AlgA \in \mathsf{K}$,
    with $\mathsf{K}$ as described
    in Proposition~{\ref{fact:residuated-is-imp}},
    is an implication that corresponds to the
    residuum of $\land^\AlgA$.
\end{definition}
}

Inspired by~\cite{avron2020},
we also consider the following additional criteria to guide our investigations: 

\begin{description}
    \item[$\CondAOne$] 
    \label{condA1}
    The expanded logic
    is determined by
    the single PNmatrix
    $\langle \mathbf A, D \rangle$,
    and,
    for all $a,b \in A$,
    we have
    $a \GenImp^{\mathbf A} b \subseteq D$ if, and only if,
    either
    $a \not\in D$ or
    $b \in D$.
    \item[$\CondATwo$] 
    \label{condA2}
    The
    expanded logic
    is self-extensional (Definition~\ref{def:selfextensional}).
\end{description}

Note that, for \SetSet{} logics, the criterion $\CondAOne$ is strong
to the point of forcing the referred implication
to be classic-like:

\begin{proposition}
\label{A1-is-classiclike}
    Let $\DefMat \SymbDef \langle \mathbf A, D \rangle$ be a $\Sigma_{\GenImp}$-PNmatrix.
    Then
    $\GenImp$ 
    is a classic-like implication in $\sequent_{\DefMat}$
    if, and only if,
    $\DefMat$ satisfies $\CondAOne$.
\end{proposition}
\begin{proof}
We have that
    $\FmSetA,\FmA \, \notsequent_{\DefMat} \, \FmB, \FmSetB$
    if, and only if,
    $v[\FmSetA\cup\{\FmA\}\}] \subseteq D$
    and
    $v[\{\FmB\}\cup\FmSetB\}] \subseteq \overline{D}$
    for some valuation $v$
    if, and only if,
    $v[\FmSetA] \subseteq D$
    and
    $v(\FmB) \in \overline{D}$
    and
    $v(\FmA \GenImp \FmB) \in \overline{D}$
    for some valuation $v$
    if, and only if,
    $\FmSetA \,\notsequent_{\DefMat}\, \FmA \GenImp \FmB, \FmSetB$.
\end{proof}

\section{Perfect Paradefinite Algebras and their logics}
In this section, we present the main definitions
concerning perfect paradefinite algebras
and the logics associated to them.
Most of the material come from~\cite{Gomes2022},
where these objects were first introduced and investigated.

\begin{definition}\label{def:ppalgebra}
Given a $\DMoSig$-algebra whose
$\DMSig$-reduct is a De Morgan
algebra, we say that it
constitutes a \emph{perfect paradefinite algebra} (\emph{PP-algebra}) if it
satisfies the equations:
\begin{table}[H]
\begin{tabular}{@{}llll@{}}
    \textbf{(PP1)} $\cons\cons x \approx \top$
    \hspace{5mm}
    &
    \textbf{(PP2)} $\circ x \approx \circ\DMNeg x$
    \hspace{5mm}
    & 
    \textbf{(PP3)} $\cons \top \approx \top$
    \hspace{5mm}
    &
    \textbf{(PP4)} $ x \land \DMNeg x \land \cons x \approx \bot$ 
\end{tabular}
\begin{tabular}{@{}llll@{}}
    \textbf{(PP5)} $\circ(x \land y) \approx (\circ x \lor \circ y) \land (\circ x \lor \DMNeg y) \land (\circ y \lor \DMNeg x)$
\end{tabular}
\end{table}
\end{definition}

\begin{example}
\label{ex:ppidefs}
An example of PP-algebra is $\PPAlg_6 \SymbDef \langle \SixSet, \cdot^{\PPAlg_6} \rangle$,
the $\DMoSig$-algebra
defined as $\ISAlg_6$
in Example {\rm\ref{ex:issix}},
differing only in that,
instead of containing an interpretation for
$\nabla$, it contains the following interpretation for~$\cons$:
\[
\cons^{\PPAlg_6} a \SymbDef
\begin{cases}
    \efvalue & a \in \SixSet{\setminus}\{\efvalue,\etvalue\}\\
    \etvalue & a \in \{\efvalue,\etvalue\}\\
\end{cases}
\]
\noindent 
Other examples are the algebras $\PPAlg_i$, the 
subalgebras
of $\PPAlg_6$
having, respectively, the same lattice structures of the algebras $\ISAlg_i$, for $2 \leq i \leq 5$, exhibited in Figure~{\rm\ref{fig:pp_algebra}}.
\end{example}

We denote by $\PPVar$
the variety of
PP-algebras.
This variety is term-equivalent
to the variety of involutive Stone algebras
\cite[Thm. 3.6]{Gomes2022} --- in particular, $\nabla x \SymbDef x \lor \DMNeg\cons x $.
Also, it holds that
$\PPVar = \Variety(\{\PPSix\})$
\cite[Prop. 3.8]{Gomes2022}.
As it occurs with IS-algebras, 
we may define, 
in the language of PP-algebras, 
a pseudo-complement satisfying the Stone equation; to that effect, it suffices to set $\PPComp x \;\SymbDef\; \DMNeg x \land \cons x$ (alternatively, one might simply set $\PPComp x \;\SymbDef\; \DMNeg\nabla x$, as usual).

We shall denote by 
$\PPlogicPrimeMC$ 
and $\PPlogic$, respectively, the \SetSet{} and \SetFmla{} order-preserving
logics induced by $\PPVar$
(cf. Subsection~\ref{sec:ordered-algebras-logics}).
In addition, we shall denote by $\PPOnelogicMC{}$ and $\PPOnelogic{}$, respectively,
the \SetSet{} and \SetFmla{} $\top$-assertional logics induced by $\PPVar$.
We know that
$\PPlogicPrimeMC{} \;=\;\sequent_{\langle \PPSix, \Upset{\both} \rangle}$
and thus $\PPlogic{} \;=\;\vdash_{\langle \PPSix, \Upset{\both} \rangle}$~\cite[Theorem 3.11]{Gomes2022}.

Taking a proof-theoretical perspective,
from~\cite[Cor. 4.3]{Gomes2022}
we know that
$\PPlogicPrimeMC{}$ is axiomatized by
a 
$\{p, \DMNeg p, \cons p\}$-analytic
\SetSet{} calculus, which we now recall:

\begin{definition}
\label{def:ppsixmccalc}
Let $\CalcVar_{\PPlogicPrimeMC}$
be the \SetSet{} calculus
containing the rules
in Example~\ref{ex:fdeax}
together with the following rules:
\begin{align*}
\inferx[\RuleA_{19}]
{}
{\cons \bot}
&&
\inferx[\RuleA_{20}]
{}
{\cons \top}
&&
\inferx[\RuleA_{21}]
{}
{\cons \cons p}
&&
\inferx[\RuleA_{22}]
{\cons p }
{\cons \DMNeg p}
&&
\inferx[\RuleA_{23}]
{\cons \DMNeg p }
{\cons p}
&&
\inferx[\RuleA_{24}]
{\cons p}
{p, \DMNeg p}
&&
\inferx[\RuleA_{25}]
{\cons p, p, \DMNeg p}
{}
\end{align*}
\begin{align*}
\inferx[\RuleA_{26}]
{\cons p}
{\cons(p \land q), p}
&&
\inferx[\RuleA_{27}]
{\cons q}
{\cons(p \land q), q}
&&
\inferx[\RuleA_{28}]
{\cons(p \land q), q}
{\cons p}
&&
\inferx[\RuleA_{29}]
{\cons(p \land q), p}
{\cons q}
&&
\inferx[\RuleA_{30}]
{\cons p, \cons q}
{\cons(p \land q)}
&&
\inferx[\RuleA_{31}]
{\cons(p \land q)}
{\cons p, \cons q}
\end{align*}
\begin{align*}
\inferx[\RuleA_{32}]
{\cons p, \cons q}
{\cons(p \lor q)}
&&
\inferx[\RuleA_{33}]
{\cons(p \lor q)}
{\cons p, \cons q}
&&
\inferx[\RuleA_{34}]
{\cons p, p}
{\cons(p \lor q)}
&&
\inferx[\RuleA_{35}]
{\cons q, q}
{\cons(p \lor q)}
&&
\inferx[\RuleA_{36}]
{\cons(p \lor q)}
{\cons p, q}
&&
\inferx[\RuleA_{37}]
{\cons(p \lor q)}
{\cons q, p}
\end{align*}
\end{definition}

\noindent 
In the mentioned paper,
the above calculus was  transformed
into a $\SetFmla{}$ axiomatization for
$\PPlogic$,
using a technique that we shall detail and
employ in Section~\ref{sec:int-hilbert-axiomat}.

It is perhaps worth  calling attention to the contribution played by rules $\RuleA_{24}$ and $\RuleA_{25}$ in making the perfection operator, $\cons$, restore `classicality', as described in the so-called Derivability Adjustment Theorems (for a semantical perspective, see \citet[Sec. 2]{jmarcos2005}, and more specifically \citet[Thm.~3.29]{Gomes2022}).

Finally, we have that the $\top$-assertional logics $\PPOnelogicMC{}$ and $\PPOnelogic{}$
are determined by a single three-valued matrix.
In fact, it can be shown that such logics
are term-equivalent to the three-valued \L ukasiewicz logic (in \SetSet{} and \SetFmla{}, respectively) ---
{
in a language containing a `possibility' operator definable by setting $\nabla x \SymbDef x \lor \DMNeg\cons x $, as above.
}

\begin{proposition}[{\citet*[Prop. 3.12]{Gomes2022}}]
$\PPOnelogicMC{} \;=\; \sequent^{\top}_{\Variety(\PPAlg_3)} \;=\; \sequent_{\langle \PPAlg_3, \{\etvalue\} \rangle}$,
and thus $\PPOnelogic{} \;=\; \OneAssert_{\Variety(\PPAlg_3)} \;=\; \vdash_{\langle \PPAlg_3, \{\etvalue\} \rangle}$
(recall the definition of $\PPAlg_3$ in Example~\ref{ex:ppidefs}).
\end{proposition}

We have at this point all the relevant facts about
the logics of PP-algebras we are interested in.
Let us move to the main goal of the paper:
adding an implication to them.

\section{Conservatively 
expanding
$\PPlogicPrimeMC$
and
$\PPlogic$
by adding a classic-like
implication
to their
matrix 
}
\label{sec:classical-implication}

The first path we shall entertain amounts to 
modifying the 
logical matrix of 
$\PPlogicPrimeMC$%
, namely $\langle \PPSix, \Upset{\both} \rangle$,
by enriching its
algebra with a new multioperation $\DesImp$,
thus obtaining a multialgebra $\PPSixDesImp$
for which 
both criteria $\CondAOne$ and $\CondATwo$ mentioned
in Section~\ref{sec:introduction} hold,
that is:
\begin{description}
    \item[$\CondAOne$] \label{condaone}
    $a \DesImp b \subseteq \Upset{\both}$ 
    if, and only if,
    either
    $a \not\in \Upset{\both}$ or
    $b \in \Upset{\both}$.

    \item[$\CondATwo$] \label{condatwo} 
    The resulting logic
    $\sequent_{\langle \PPSixDesImp, \Upset{\both} \rangle}$
    is self-extensional.
\end{description}

This path soon leads to a dead end, since:
\begin{therm}
\label{thm:noexist}
There is no multialgebra $\PPSixDesImp$ simultaneously satisfying conditions $\CondAOne$
and $\CondATwo$.
\end{therm}
\begin{proof}
Let $\DesImp$ be an implication defined in $\PPSix$ that satisfies $\CondAOne$
and $\DefMat^{\DesImp} 
\SymbDef \langle \PPSixDesImp, \Upset{\both} \rangle$.
Then, ($\star$) we have\; $\sequent_{\DefMat^{\DesImp}} \, p  \lor (p  \GenImp \bot )$, 
because, for every valuation $v$, either $v (p ) \in \Upset{\both}$ ---
in which case we have 
$v(p  \lor (p  \GenImp \bot )) \in v(p ) \lor v(p  \GenImp \bot )  \subseteq \Upset{\both}$ as well --- or
$v (p ) \notin \Upset{\both}$, which gives us
$ v(p  \GenImp \bot )  \in v(p ) \GenImp v(\bot ) \subseteq \Upset{\both}$ by $\CondAOne$.

Note that, for a valuation $v$  such that $v(p) = \both $, using  $\CondAOne$ we have
$v(p  \GenImp \bot ) \in \both \GenImp \efvalue \not\subseteq \Upset{\both}$, which means that
$v(p  \GenImp \bot )$ may take a value in $\{ \neither, \fvalue, \efvalue \}$.
Hence, this being the case,
$v(p  \lor (p  \GenImp \bot )) 
\in \both \lor v(p  \GenImp \bot ) \in \{ \both, \tvalue \}$,
which entails
$v ( \circ (p  \lor (p  \GenImp \bot )) ) = \efvalue$.
Thus $\notsequent_{\DefMat^{\DesImp}} \,
\circ (p  \lor (p  \GenImp \bot )) $. 

This prevents the logic from being self-extensional. Indeed, assuming $\CondATwo$,
from ($\star$) we would have that $p  \lor (p  \GenImp \bot )$ and $\top$ are logically equivalent,
thus
we would be able to conclude that $\circ (p  \lor (p  \GenImp \bot )) $ and $\circ \top$ are logically equivalent too. Since $\sequent_{\DefMat^{\DesImp}} \, \circ \top$,
we would conclude that $\sequent_{\DefMat^{\DesImp}} \, \circ (p  \lor (p  \GenImp \bot )) $, against
what we have shown.
\end{proof}

In view of the preceding result,
in this section we shall proceed by
pursuing $\CondAOne$, thus necessarily admitting
non-self-extensional logics.
In the subsequent sections we will explore instead some avenues that arise when we opt for abandoning $\CondAOne$.

{
The space of binary multioperations over
$\SixSet$ satisfying
$\CondAOne$ is finite but very large, 
consisting of total
refinements~\footnote{
We disregard partial refinements here (i.e. refinements with empty sets as outputs) since $\CondAOne$ was originally formulated for
the total case and at this point we do not envisage meaningful gains in making allowance for the partial case.}
of 
the multioperation
$\FullImp$ defined as:}
\[
\FullImp(a, b) \SymbDef
\begin{cases}
    \Upset\both & \text{if $a \not\in \Upset\both$ or $b \in \Upset\both$ }\\
    \SixSet{\setminus}\Upset\both
    & \text{otherwise.}
\end{cases}
\]
Denote by $\PPSix^{\!\!\FullImp}$ the multialgebra
obtained from $\PPSix$ by expanding its signature
with $\GenImp$ and interpreting this connective
as $\FullImp$. Let $\DefMat^{\FullImp} \SymbDef
\langle \PPSix^{\FullImp}, \Upset\both \rangle$.
We will soon see  
that the logic induced
by this matrix
plays an important role
regarding the conservative expansions
of $\PPlogicPrimeMC$ by a classic-like implication; but let us first
 present an analytic
axiomatization for it.

\begin{definition}
    \label{def:full-calc}
    Let
    $\CalcVar_{\DefMat^{\FullImp}}$
    be the calculus given
    by all inference rules in
    $\CalcVar_{\PPlogicPrimeMC}$
    plus the following three inference rules:
$$
\inferx[\RuleA_1^{\CL}]
{\DeffProp}
{\DefProp\GenImp\DeffProp}
\qquad
\inferx[\RuleA_2^{\CL}]
{}
{\DefProp, \DefProp\GenImp\DeffProp}
\qquad
\inferx[\RuleA_3^{\CL}]
{\DefProp, \DefProp\GenImp\DeffProp}
{\DeffProp}
$$
\end{definition}

\begin{remark}
    The rule
    $\RuleA_2^{\CL}$ is 
    responsible for 
    the classicality of $\GenImp$
    in $\CalcVar_{\DefMat^{\FullImp}}$. 
    One may observe, in particular, that 
    in the presence of this rule,
    the connective defined by setting
    $\neg p := p \GenImp \bot$
    satisfies 
    the `law of excluded middle'
    ($p \lor \neg p$ is provable).
\end{remark}

In a PNmatrix $\DefMat \SymbDef
\langle \AlgA, D \rangle$,
a set $\mathcal{S}_a$ of unary formulas is said to 
\emph{isolate $a \in A$}
if for every $b \neq a \in A$ there is in $\mathcal{S}_a$ a separator
for $a$ and $b$.
A \emph{discriminator for $\DefMat$}
is a family $\{ (\DiscSet{a}, \DiscNSet{a}) \}_{a \in A}$ such that $\DiscSet{a} \cup \DiscNSet{a}$
isolates $a$,
where $\mathrm{S}^\AlgA(a) \subseteq D$
whenever $\mathrm{S} \in \DiscSet{a}$,
and $\mathrm{S}^\AlgA(a) \subseteq \overline D$
whenever $\mathrm{S} \in \DiscNSet{a}$.
We say that a PN-matrix is \emph{monadic} in case there is a discriminator for it.
Discriminators play an essential role
in the axiomatization of monadic PN-matrices,
as we shall see in the following results.


\begin{therm}
\label{fact:classical-imp-axiomat}
$\CalcVar_{\DefMat^{\FullImp}}$
is $\{ \DefProp, \DMNeg\DefProp, \cons\DefProp \}$-analytic
and axiomatizes
    $\sequent_{\DefMat^{\FullImp}}$.
\end{therm}
\begin{proof}
    We consider the axiomatization
    method presented by
    \citet*[Theorem 3.5 and Theorem 3.12]{marcelino:2019},
    which can be applied because
    $\DefMat^{\FullImp}$ is
    monadic, with the same discriminator
    as that of 
    $\PPSixMat$ (see Table~\ref{tab:discrim-full}).
    The method can be seen
    essentially
    as a process of refining
    fully indeterministic
    six-valued 
    interpretations
    of the connectives
    (that is, interpretations
    where $\SixSet$ appears at every entry)
    by imposing soundness
    of some collections of
    inference rules, until obtaining
    the desired interpretations
    (in our case, the ones
    of $\DefMat^{\FullImp}$).
    For the connectives 
    $\land,\lor,\DMNeg$
    and $\cons$,
    the method produces
    the calculus 
    $\CalcVar_{\PPlogicPrimeMC}$
    after some simplifications (\cite{Gomes2022}).
    Thus, since the method
    is modular on the connectives,
    we only need to
    check what happens when
    we run it on the new 
    connective $\FullImp$
    and then add the resulting rules
    to $\CalcVar_{\PPlogicPrimeMC}$.
    The rules that are 
    imposed will have the following
    shape:
        \[
    \inferx
    {\DiscSet{a}(\DefProp),
    \DiscSet{b}(\DeffProp),
    \DiscSet{c}(\DefProp \GenImp \DeffProp)
    }
    {\DiscNSet{a}(\DefProp),
    \DiscNSet{b}(\DeffProp),
    \DiscNSet{c}(\DefProp \GenImp \DeffProp)
    }
    \]
    for each
    $c \in \SixSet{\setminus}
    (a \FullImp b)$
    and each $a,b \in \SixSet$,
    where
    the sets $\DiscSet{d}$
    and $\DiscNSet{d}$ for each
    $d \in \SixSet$ are given by the following table:
    \begin{table}[H]
        \centering
        \begin{tabular}{c|c|c}
             $d$ & $\DiscSet{d}$ & $\DiscNSet{d}$\\
             \midrule 
             $\efvalue$ & $\cons p$ & $p$\\
             $\fvalue$ & $\DMNeg p$ & $\cons p, p$\\
             $\neither$ & $\varnothing$ & $p, \cons p, \DMNeg p$\\
             $\both$ & $p, \DMNeg p$ & $\cons p$\\
             $\tvalue$ & $p$ & $\cons p, \DMNeg p$\\
             $\etvalue$ & $p, \cons p$ & $\varnothing$\\
        \end{tabular}
        \caption{Discriminator for 
        $\DefMat^{\FullImp}$.
        They give, for example,
        $\DiscSet{\both} = \{ p, \DMNeg p \}$ and
        $\DiscNSet{\neither} = \{ p, \cons p, \DMNeg p \}$.}
        \label{tab:discrim-full}
    \end{table}
    
    These rules
    at first do not
    seem to relate to 
    the three rules for $\GenImp$ in $\CalcVar_{\DefMat^{\FullImp}}$.
    However, we will see below that each of them is
    a `dilution' of one of
    the latter (that is,
    they are obtained from the
    latter by adding formulas
    on the antecedents and 
    succedents).
    Then, because these three rules
    for implication
    are sound,
    it follows trivially that
    we can use only them
    and discard its dilutions
     without 
    harm for completeness
    and for analyticity.
    
    First, consider the entries
    of
    $a \FullImp b$
    for $b \in \Upset\both$.
    The values
    $c \in \SixSet{\setminus}
    (a \FullImp b)$
    are precisely
    $\neither$,
    $\fvalue$ and
    $\efvalue$.
    Then $\DeffProp \in \DiscSet{b}(\DeffProp)$ and
    $\DefProp\GenImp\DeffProp 
    \in \DiscNSet{c}(\DefProp\GenImp\DeffProp)$,
    thus the above rules are
    all dilutions of 
    $\mathsf{r}_1^{\CL}$.

    Second, consider the entries
    of
    $a \FullImp b$
    for $a \not\in \Upset\both$.
    The values
    $c \in \SixSet{\setminus}
    (a \FullImp b)$
    are precisely
    $\neither$,
    $\fvalue$ and
    $\efvalue$.
    Then $\DefProp \in \DiscNSet{a}(\DefProp)$ and
    $\DefProp\GenImp\DeffProp 
    \in \DiscNSet{c}(\DefProp\GenImp\DeffProp)$,
    thus the above rules are
    all dilutions of 
    $\mathsf{r}_2^{\CL}$.

    Finally, consider the entries
    of
    $a \FullImp b$
    for $a \in \Upset\both$
    and $b \not\in \Upset\both$.
    The values
    $c \in \SixSet{\setminus}
    (a \FullImp b)$
    are precisely
    $\both$,
    $\tvalue$ and
    $\etvalue$.
    Then $\DefProp \in \DiscSet{a}(\DefProp)$,
    $\DefProp\GenImp\DeffProp 
    \in \DiscSet{c}(\DefProp\GenImp\DeffProp)$
    and
    $\DeffProp \in \DiscNSet{b}(\DeffProp)$,
    thus the above rules are
    all dilutions of 
    $\mathsf{r}_3^{\CL}$.
\end{proof}

\begin{remark}
    For a discriminator in
    the language of
    involutive Stone
    algebras, check~\cite[Prop. 4.3]{cantu2022}.
\end{remark}

\begin{remark}
    The above result could
    have been obtained by another 
    strategy,  observing that
    $\sequent_{\DefMat^{\FullImp}}$
    is induced by
    the \emph{strict product} \cite[Definition 10]{Caleiro2023} of
    $\PPSixMat$
    and the matrix $\DefMat_2$ over the
    signature containing
    only $\GenImp$, in which
    the algebra has carrier
     $\{ \efvalue, \etvalue \}$ and interprets
    $\GenImp$ as in classical logic
    ($a \GenImp^{\DefMat_2} b = \efvalue$ if, and only if, $a = \etvalue$ and $b = \efvalue$),
    and the designated set is $\{ \etvalue \}$.
    By~\citet*[Theorem 12]{Caleiro2023},
    this implies that 
    $\sequent_{\DefMat^{\FullImp}}$
    is the \emph{disjoint fibring} of
    $\PPlogicPrimeMC$
    and $\sequent_{\DefMat_2}$
    (that is, the smallest \SetSet{} logic
    in the signature $\Sigma_\GenImp$
    extending both logics).
    One can then show
    that, because both
    logics have
    analytic \SetSet{} axiomatizations,
    it is enough 
    to merge both calculi in order to
    axiomatize their disjoint fibring.
    As it is well-known
    that $\sequent_{\DefMat_2}$
    is axiomatized by
    the $\{ \DefProp \}$-analytic
    calculus 
    given by the three
    rules for implication in
    Definition~\ref{def:full-calc},
    the desired result follows.
\end{remark}

From the above, we have that $\sequent_{\DefMat^{\FullImp}}$
is special in the sense of
being the smallest 
conservative expansion of
$\PPlogicPrimeMC$ in which the added
implication is classic-like
(cf.\ Definition~\ref{def:classicalimplication}):

\begin{proposition}
    Let $\sequent$ be a
    conservative expansion
    of $\PPlogicPrimeMC$
    over $\DMoSig_\GenImp$
    in which $\GenImp$ is
    a classic-like implication.
    Then $\sequent_{\DefMat^{\FullImp}} \subseteq \sequent$.
\end{proposition}
\begin{proof}
    It is easy to see that
    any classic-like implication
    must satisfy the three
    rules 
    for implication
    in the calculus
    for $\sequent_{\DefMat^{\FullImp}} $,
    and this is all we need
    for the present result.
\end{proof}

We observe that each proper refinement of $\FullImp$
produces a proper extension
of $\sequent_{\DefMat^{\FullImp}}$.
In fact, following
again the axiomatization method
in \cite{marcelinowollic}, 
for each such extension
we can use the monadicity of
$\DefMat^{\FullImp}$ to obtain
$\SetSet{}$ rules that
hold in it
but not in $\CalcVar_{\DefMat^{\FullImp}}$,
and they will be precisely
the rules that need to be added
to the latter to obtain
a $\{ \DefProp,\DMNeg\DefProp, \cons\DefProp \}$-analytic
calculus for these extensions.

\begin{proposition}
    \label{prop:axiomat-ref-mfull}
    Let $\DefMat$ be 
    obtained from $\DefMat^{\FullImp}$
    by refining $\FullImp$.
    Then 
    $\CalcVar_{\DefMat^{\FullImp}}$
    with the following rules
    provide a $\{ \DefProp,\DMNeg\DefProp, \cons\DefProp \}$-analytic
    calculus for $\DefMat$:
    \[
    \inferx
    {\DiscSet{a}(\DefProp),
    \DiscSet{b}(\DeffProp),
    \DiscSet{c}(\DefProp \GenImp \DeffProp)
    }
    {\DiscNSet{a}(\DefProp),
    \DiscNSet{b}(\DeffProp),
    \DiscNSet{c}(\DefProp \GenImp \DeffProp)
    }
    \]
    for each
    $c \in (a \FullImp b) {\setminus} (a \GenImp^{\DefMat} b)$
    and each $a,b \in \SixSet$.
\end{proposition}
\begin{proof}
    Directly from
    the method in~\citet*[Theorem 3.5 and Theorem 3.12]{marcelino:2019}.
\end{proof}

For an example of the latter axiomatization technique in action, 
if we remove the value $\neither$
from the entry
$\both \FullImp \efvalue$,
we axiomatize the resulting logic
by adding the following rule:
\[
\inferx
{p, \DMNeg p, \cons q }
{\cons p, q, p \GenImp q, \DMNeg(p \GenImp q), \cons (p \GenImp q)},
\]
which clearly does not hold in 
$\DefMat^{\FullImp}$
under a valuation $v$
with $v(p) = \both$,
$v(q) = \efvalue$
and $v(p \GenImp q) = \neither$,
but holds in the new logic
precisely because this valuation
was forbidden when we deleted
$\neither$ from that entry.

As another application of this 
technique, we show how to axiomatize the logic
$\mathrm{LET_K^+}$ of~\cite{coniglio2023sixvalued}
with a $\{ \DefProp,\DMNeg\DefProp, \cons\DefProp \}$-analytic
    calculus,
as it fits precisely in this setting, that is, 
it is determined by a refinement
of $\DefMat^{\FullImp}$.
In fact, its implication is given by the following
truth table:

\begin{center}
\begin{tabular}[c]{c|cccccc}
$\GenImp^{\mathrm{LET_K^+}}$ & $\efvalue$ & $\fvalue$       & $\neither$       & $\both$ & $\tvalue$ & $\etvalue$\\
\toprule
$\efvalue$  & $\etvalue$ & $\etvalue$ & $\etvalue$ & $\etvalue$ & $\etvalue$ & $\etvalue$ \\
$\fvalue$        & $\tvalue$       & $\tvalue$ & $\tvalue$       & $\tvalue$  & $\tvalue$   & $\etvalue$ \\

$\neither$        & $\tvalue$       & $\tvalue$ & $\tvalue$       & $\tvalue$  & $\tvalue$   & $\etvalue$ \\
$\both$  & $\efvalue$ & $\fvalue$ & $\neither$    & $\both$  & $\tvalue$ & $\etvalue$\\
$\tvalue$  & $\efvalue$ & $\fvalue$ & $\neither$    & $\both$   & $\tvalue$ & $\etvalue$\\
$\etvalue$  & $\efvalue$ & $\fvalue$ & $\neither$    & $\both$   & $\tvalue$ & $\etvalue$
\end{tabular}
\end{center}

\noindent Then, to obtain the desired axiomatization,
it is enough to add to 
$\CalcVar_{\DefMat^{\FullImp}}$
the following inference rules:
$$
\inferx[]{\DMNeg (p \GenImp q)}{p}
\quad
\inferx[]{\DMNeg (p \GenImp q)}{\DMNeg q}
\quad
\inferx[]{p, \DMNeg q}{\DMNeg(p \GenImp q)}
\quad
\inferx[]{\cons(p \GenImp q)}{\cons p, \cons q}
\quad
\inferx[]{\cons(p \GenImp q)}{\cons p, p, q}
$$
$$
\inferx[]{\cons(p \GenImp q), p}{\cons q}
\quad 
\inferx[]{\cons p}{\cons(p \GenImp q), p}
\quad 
\inferx[]{p, \cons q}{\cons(p \GenImp q)}
\quad 
\inferx[]{\cons q, q}{\cons(p \GenImp q)}
$$
\noindent 
The attentive reader will notice that they are not
quite the same rules as the ones
produced by the recipe in Proposition~\ref{prop:axiomat-ref-mfull}.
In fact, they are simplifications thereof,
following the streamlining procedures described in~\cite{marcelino:2019}.
Still, it is not hard to confirm that
these rules produce the desired refinements of $\FullImp$
in a similar way as we did in the previous example.

\section{Logics of PP-algebras expanded
with a Heyting implication}
\label{sec:intuitionistic-implication}
In the light of Theorem~\ref{thm:noexist}, as we proceed to expand
our logic
with an implication
we shall necessarily have to drop either $\CondAOne$ or $\CondATwo$.
In this section we explore the first option, that is,
we stick to 
$\CondATwo$ while
dropping $\CondAOne$.
Having fixed a (deterministic) implication operator (say, on $\PPSix$),
a straightforward  way to ensure that the resulting logic will be self-extensional
(cf.~Proposition~\ref{prop:order-are-selfext})
is to consider, as in the
implication-less case, the \SetFmla{} consequence
relation that preserves the lattice order of $\PPSix$ (or, to be more precise,
of the resulting class of algebras augmented with an implication).
Indeed, as shown by~\citet[Sec. 3]{Jansana2006}, when a conjunction is present, every self-extensional logic turns out to  be the consequence associated to a suitably defined partial order.
%
We shall thus follow this route, which still leaves us free to choose among the implication operators for $\PPSix$. 
Since on $\PPSix$ we cannot define a classic-like implication suitable for our purposes
---i.e., one satisfying both $\CondAOne$ and $\CondATwo$--- we suggest
introducing a Heyting
implication.
From an algebraic point of view, such an operator is readily available.
Indeed, since $\PPSix$ has a (finite) distributive lattice reduct, 
the meet operation has a residuum (we will  denote it by $\ResImp$), which is
precisely the relative pseudo-complement operation. 

\begin{definition}
    Let $\PPSixResImp$ be the algebra obtained
    by expanding $\PPSix$ with the operation
    $\ResImp$ defined as follows:
    \[
    a \ResImp b \SymbDef
    \max\{ c \in \SixSet : a \land^{\PPSix} c \leq b \},
    \text{ for all $a,b \in \SixSet$.}
    \]
\end{definition}

\noindent 
In accordance with the above definition, the truth-table of the implication in $\PPSixResImp$ looks as follows: 
\begin{center}
\vspace{-3mm}
\begin{tabular}[c]{c|cccccc}
$\ResImp$ & $\efvalue$ & $\fvalue$       & $\neither$       & $\both$ & $\tvalue$ & $\etvalue$\\
\toprule
$\efvalue$  & $\etvalue$ & $\etvalue$ & $\etvalue$ & $\etvalue$ & $\etvalue$ & $\etvalue$ \\
$\fvalue$        & $\efvalue$       & $\etvalue$ & $\etvalue$       & $\etvalue$  & $\etvalue$   & $\etvalue$ \\

$\neither$        & $\efvalue$ & $\both$   & $\etvalue$        & $\both$ & $\etvalue$ & $\etvalue$\\
$\both$  & $\efvalue$ & $\neither$ & $\neither$    & $\etvalue$  & $\etvalue$ & $\etvalue$\\
$\tvalue$  & $\efvalue$ & $\fvalue$ & $\neither$    & $\both$   & $\etvalue$ & $\etvalue$\\
$\etvalue$  & $\efvalue$ & $\fvalue$ & $\neither$    & $\both$ & $\tvalue$ & $\etvalue$
\end{tabular}
\end{center}
\vspace{-3mm}

\noindent
$\PPSixResImp$ is  obviously (i.e., has a reduct which is) a Heyting algebra. 
Furthermore, since it also carries a De Morgan negation,
it may be called a \emph{De Morgan-Heyting algebra} according to~\cite{Sankappanavar1987},
or a \emph{symmetric Heyting algebra} according to~\cite{Monteiro1980}\footnote{Monteiro's work suggests
another natural candidate for an implication (this one, already term-definable) in $\PPSix$.
This is the `implication faible' $\GenImp_{\mathsf{W}}$ introduced in~\cite[Ch.~IV,~Def.~4.1]{Monteiro1980}, which can be given 
by the following term:
$
x \GenImp_{\mathsf{W}} y :=   \DMNeg x \lor \DMNeg \cons  x  \lor y
$.}.
Indeed, since $\PPSixResImp$ (as $\PPSix$) also has a Stone lattice reduct (i.e. a pseudo-complemented distributive lattice satisfying $\neg x \lor \neg\neg x \approx \top$),
we may be a bit more specific, observing that $\PPSixResImp$ is also,
in Monteiro's terminology, 
a \emph{Stonean symmetric Heyting algebra}~\cite[Ch.~IV,~Def.~1.1]{Monteiro1980}.
These observations will be exploited in Section~\ref{sec:moi}.

\begin{remark}
Before we proceed any further, one may wonder whether we are really 
adding something new to $\PPSix$. In other words,
was the implication $\ResImp$ already term-definable in this algebra?
The answer is negative.
To see why, observe that
$\PPSix$ is a subdirectly irreducible algebra having a single non-trivial congruence $\theta$,
which is the one 
that identifies (only) the elements in the set
$\{ \tvalue, \fvalue, \both, \neither \}$.
By adding the Heyting implication, we obtain a simple algebra, in which $\theta$ is no longer a congruence 
(indeed,
$\both \theta \neither$,
but it is not the case that
$(\both \ResImp \both) \theta (\both \ResImp \neither)$). 
\end{remark}

Before moving to the logics of order
associated to the new algebra,
we could first consider 
the \SetSet{} logic determined by
$\langle \PPSixResImp, \Upset{\both} \rangle$,
which is guaranteed to conservatively expand
$\PPlogicPrimeMC$ (with the implication $\ResImp$  being the residuum of  $\land$).
Since
$\PPSixMat$ is monadic, we can
 axiomatize $\ResImp$
straight away (similarly to what we did in the
proof of Theorem~\ref{fact:classical-imp-axiomat}) following
the method of~\cite{marcelino:2019}:

\begin{therm}
The \SetSet{} $\Sigma_\GenImp$-logic
induced by
$\PPSixMat$ expanded with
$\ResImp$ is axiomatized
by the $\{ \DefProp, \DMNeg\DefProp, \cons\DefProp \}$-analytic
\SetSet{} calculus given
by $\CalcVar_{\PPlogicPrimeMC}$ plus
the following rules:

$$
\inferx[]
{\DeffProp}
{\DefProp\GenImp\DeffProp}
\qquad
\inferx[]
{\DefProp, \DefProp\GenImp\DeffProp}
{\DeffProp}
\qquad 
\inferx[]
{\DMNeg(\DefProp \GenImp \DeffProp)}
{\DMNeg\DeffProp}
\qquad 
\inferx[]
{{\DMNeg}\DeffProp}
{\DMNeg(\DefProp \GenImp \DeffProp), 
{\DMNeg}\DefProp}
$$
$$
\inferx[]
{}
{\DefProp\GenImp\DeffProp,\cons\DeffProp,\DefProp}
\qquad
\inferx[]
{\DefProp\GenImp\DeffProp}
{\cons(\DefProp\GenImp\DeffProp), \DMNeg\DeffProp,\DeffProp}
\qquad
\inferx[]
{\DefProp\GenImp\DeffProp, \cons\DeffProp}
{\cons\DefProp, \DeffProp}
\qquad
\inferx[]
{\DMNeg(\DefProp \GenImp \DeffProp), \DMNeg\DefProp}
{\cons(\DefProp \GenImp \DeffProp)}
$$
$$
\inferx[]
{\DMNeg\DefProp}
{\cons(\DefProp \GenImp \DeffProp), \DefProp}
\qquad
\inferx[]
{\cons(\DefProp \GenImp \DeffProp), \cons\DefProp, \DefProp}
{\cons\DeffProp}
\qquad
\inferx[]
{\cons(\DefProp \GenImp \DeffProp), \DefProp}
{\cons\DeffProp, \DeffProp}
\qquad
\inferx[]
{\cons\DefProp}
{\DefProp \GenImp \DeffProp, \DefProp}
$$
$$
\inferx[]
{\cons\DeffProp}
{\cons(\DefProp \GenImp \DeffProp)}
\qquad
\inferx[]
{\DeffProp}
{\DMNeg(\DefProp \GenImp \DeffProp),
\cons(\DefProp \GenImp \DeffProp),
\cons\DefProp}
$$

\end{therm}
    
Note however that
$\langle \PPSixResImp, \Upset{\both} \rangle$
does not satisfy $\CondAOne$:
pick
for instance $a \SymbDef \fvalue$ and $b \SymbDef \efvalue$.
{
As for $\CondATwo$, we note that self-extensionality also fails. 
To see why, consider the formula 
$\varphi:=\DMNeg (p \GenImp q) \land \DMNeg (q \GenImp p )$, 
and recall that $\nabla x \SymbDef x \lor \DMNeg\cons x $.
While it is clear that both $\varphi$ and $\nabla\bot$ are logically equivalent to $\bot$, the same does not hold for $\nabla\varphi$.
}
To see this,
consider a valuation such that $v(p) = \both$ and $v(q) = \neither$. Then 
$v (\nabla (\DMNeg (p \GenImp q) \land \DMNeg (q \GenImp p ))) 
= \nabla^{\PPSixResImp}(\DMNeg^{\PPSixResImp} (\both \GenImp^{\PPSixResImp} \neither) \land^{\PPSixResImp} \DMNeg^{\PPSixResImp} (\neither \GenImp^{\PPSixResImp} \both )) =
\etvalue$. 
{
It follows that $\varphi\; \SymLogEquiv \;\bot$ holds good, while $\nabla\varphi\; \SymLogEquiv \;\nabla\bot$ does not hold.
}
In fact, as we are going to show
in Proposition~\ref{prop:notselfnew},
there is only one self-extensional \SetFmla{} logic determined by a class of matrices based on the  algebra  
$\PPSixResImp$ and principal filters.
    

Let us now resume our discussion about the logics of order.
Having obtained a new
algebra $\PPSixResImp$, we can consider
the variety it generates, denoted  $\PPImpResVar$,
and 
the order-preserving logics associated to it.
We denote by 
$\PPResImplogicPrimeMC$
and
$\PPResImplogicSingle$,
respectively,
the
\SetSet{} and
\SetFmla{}
order-preserving
logics associated to $\PPImpResVar$
(cf. Subsection~\ref{sec:ordered-algebras-logics}
for the precise definitions).
By the residuation property, we clearly have that
$\ResImp$ is an implication in
    these logics
    (cf. Definition~\ref{def:implication}
    and Proposition~\ref{fact:residuated-is-imp}).
    %

%
%

\begin{proposition}
\label{prop:consexp}
$\PPResImplogicPrimeMC$
and
    $\PPResImplogicSingle$
    are conservative expansions of
    $\PPlogicPrimeMC$
    and
    $\PPlogic$,
    respectively.
    Moreover,
    they are self-extensional.
\end{proposition}
\begin{proof}
    That both logics are conservative expansions of
    $\PPlogic$ follows
    directly from their matrix characterizations
    (Proposition \ref{prop:single-gen-variety-matrices}), while
    self-extensionality 
    follows from Proposition~\ref{prop:order-are-selfext}.
\end{proof}

{
\begin{remark}
\label{rem:just-imp-def}
Observe that 
$p \lor q \sequent_{\PPResImplogicPrimeMC} p, q$ (by the above proposition), while
$\notsequent_{\PPResImplogicPrimeMC} (p \lor q) \GenImp p, q$
(take a valuation $v$ such that $v(p) := \both$ and $v(q) := \neither$),
showing that $\ResImp$ is not
a classic-like implication in 
$\PPResImplogicPrimeMC$
(recall Definition~\ref{def:classicalimplication}).
This is one of the main reasons
for which we chose
Definition~\ref{def:implication}
as a more general notion of
implication in \SetSet{}
instead of the one in
Definition~\ref{def:classicalimplication}.
\end{remark}
}

{
We may actually improve Proposition \ref{prop:consexp} by removing some redundancies from the classes of matrices that characterize the logics thereby considered:
}

\begin{proposition}
\label{prop:redmatrdeflog}
The logics
$\PPResImplogicPrimeMC$ and $\PPResImplogic$
are determined by the class 
of matrices
$\{ 
            \langle \PPSixResImp, \Upset{\fvalue} \rangle,
            \langle \PPSixResImp, \Upset{\both} \rangle,
            \langle \PPSixResImp, \Upset{\etvalue} \rangle
        \}$.
    Moreover, all the latter matrices are
    reduced.
\end{proposition}
\begin{proof}
Note that
any matrix $\langle \PPSixResImp, D \rangle$ with $D \neq \SixSet$ is reduced, because $\PPSixResImp$ is a simple algebra. 
Regarding $\PPResImplogicPrimeMC$,
we observe that two of the matrices appearing in
Proposition~\ref{prop:single-gen-variety-matrices} may be safely omitted.
Obviously this holds for
$\langle \PPSixResImp, \Upset{\efvalue} \rangle$, which
defines a trivial logic.
Note further that $\langle \PPSixResImp, \Upset{\both} \rangle$
and $\langle \PPSixResImp, \Upset{\neither} \rangle$
are isomorphic, and thus determine the same
logic; so one of them may also be omitted.
Finally, note that the matrix with set of designated values 
$\Upset \tvalue$ is not considered as
we only need to take prime filters into account.
Thus only the matrices listed in the statement  remain.
Regarding $\PPResImplogic$, the result follows from the observation that $\PPResImplogic$ 
is the \SetFmla{} companion of the \SetSet{}
order-preserving logics.
\end{proof}

There is another \SetSet{} logic that has
$\PPResImplogic$ as \SetFmla{} companion, namely the one determined by
the class of matrices
based on the  algebra $\PPSixResImp$
having as designated sets
the principal filters of $\PPSixResImp$. 
We shall denote this logic by
$\PPResImplogicOrderMC$.
By a similar argument to the one above,
we see that
$\PPResImplogicOrderMC$ 
        is determined by the 
        class of matrices
        $\{ \langle \PPSixResImp, \Upset{\fvalue} \rangle,
            \langle \PPSixResImp, \Upset{\both} \rangle,
            \langle \PPSixResImp, \Upset{\tvalue} \rangle,
            \langle \PPSixResImp, \Upset{\etvalue} \rangle \}$
            --- that is, the matrices from
            Proposition~\ref{prop:redmatrdeflog}
            together with the matrix
            $\langle \PPSixResImp, \Upset{\tvalue} \rangle$.
In this logic, however, $\GenImp$
is not an implication in our sense,
since $\sequent_\PPResImplogicOrderMC (p \lor q) \GenImp (p \lor q)$, but
$p \lor q \, \notsequent_{\PPResImplogicOrderMC}\, p,q$.
Still, we will include $\PPResImplogicOrderMC$ in our next considerations in view of its close
relationship with~$\PPResImplogic$
and because our techniques will also apply very
naturally to it, as we shall see.

The above results do not clarify whether
one could find a single matrix to characterize
those logics (as it happens with $\PPlogic$):
the next result rules out this possibility.

\begin{therm}
\label{fact:non-single-matrix-pp6imp}
    The logic $\PPResImplogic$
    is not determined by a single logical matrix.
    The same holds also for the logics 
    $\PPResImplogicPrimeMC$ and
    $\PPResImplogicOrderMC$.
\end{therm}
\begin{proof}
    Since the two \SetSet{} logics share the
    same $\SetFmla$ companion
    $\PPResImplogic$,
    it is enough to prove that the latter
    is not determined by a single logical matrix.
{
    To that effect, we use the fact that any 
    \SetFmla{} logic~$\vdash$ that is determined by a single logical matrix respects the so-called \emph{uniformity} property (cf.\ \cite{Woj:69}), according to which $\FmSetA, \FmSetB
    \vdash \FmA$ implies
    $\FmSetA \vdash \FmA$
    for all $\FmSetA\cup\FmSetB\cup\{ \FmA \} \subseteq 
    \LangSet{\Sigma}$
    such that (1) $\Props(\FmSetA \cup \{ \FmA \})
        \cap \Props(\FmSetB) = \varnothing$ 
        and (2) $\FmSetB \not\vdash \FmB$
        for some $\FmB \in \LangSet{\Sigma}$.
    
Consider now $\FmSetA := \varnothing$, $\FmSetB := \{ p \land \DMNeg p \land q \land \DMNeg q \land \DMNeg\cons(p \GenImp q) \}$ and $\FmA := r \lor \DMNeg r$.  Condition (1) is obviously satisfied.  To see that condition (2) is also satisfied, one may consider the matrix $\langle \PPSixResImp, \Upset \fvalue \rangle$, 
    and to see that
    $\not\vdash_\PPResImplogicSingle r \lor \DMNeg r$
    it suffices to consider the matrix $\langle \PPSixResImp, \Upset \both \rangle$.
    Yet, one may now use the semantics of the order-preserving logic associated to $\PPImpResVar$ to confirm that
    $p \land \DMNeg p \land q \land \DMNeg q \land \DMNeg\cons(p \GenImp q) \vdash_{\PPResImplogicSingle} r \lor \DMNeg r$.
}
\end{proof}

As earlier anticipated, we now proceed to show
that $\PPResImplogicSingle$ is actually the only 
self-extensional
\SetFmla{} logic 
determined by
a class of matrices
based on the algebra $\PPSixResImp$
and principal filters.

\begin{proposition}
\label{prop:notselfnew}

   For $ \mathcal{V} \subseteq  \SixSet $,
let 
%
%
  $\vdash_{\mathcal{V}}$ be the \SetFmla{} logic determined by $\left\{ \langle \PPSixResImp, \up a \rangle : a \in  \mathcal{V} \right\}$. 
If the logic 
   $\vdash_{\mathcal{V}}$
   is  self-extensional,
    then
    $\vdash_{\mathcal{V}} \, = \, \PPResImplogicSingle$.
%
   
\end{proposition}
\begin{proof}
    By~\citet[Thm.~3.7]{Jansana2006}, the finitary self-extensional extensions of 
    $\PPResImplogicSingle$ are in a one-to-one correspondence with the subvarieties of $\PPImpResVar$.
    The matrices in $\{ \langle \PPSixResImp, \up a \rangle : a \in  \mathcal{V} \}$
    are reduced models
    of 
    $\vdash_{\mathcal{V}}$,
    implying that
    $\PPSixResImp$ belongs to 
    the subvariety $\mathsf{K} \subseteq \Variety (\PPSixResImp)$ associated with 
    $\vdash_{\mathcal{V}}$.
    Then $\Variety (\PPSixResImp) \subseteq \mathsf{K}$, so 
    $\Variety (\PPSixResImp) = \mathsf{K}$,
    and thus
$\PPResImplogicSingle \,=\, \vdash_{\mathcal{V}} 
$.
\end{proof}

Adopting the terminology of~\cite{Jansana2006},
we may say that $\PPResImplogicSingle$ is \emph{semilattice-based} relative to $\land$ and to $\PPImpResVar$. This observation allows us
to obtain further information:  for instance,
we know by~\citet[Thm.~3.12]{Jansana2006}
that the class of algebra reducts of reduced
matrices for $\PPResImplogicSingle$ is precisely $\PPImpResVar$. 
Semilattice-based logics are often non-algebraizable but have an algebraizable
companion, i.e.~an extension that shares the same algebraic models. We establish this
for $\PPResImplogicSingle$ below.

\begin{proposition}
\label{prop:equi}
    $\PPResImplogicSingle$ is equivalential but not algebraizable. The following is a set of equivalence formulas for it: 
    $$\Xi(x,y) \SymbDef \{ x \GenImp y, y \GenImp x, 
    \circ (x \GenImp y), \circ (y \GenImp x) \}$$
    or, equivalently (setting $\De x : = x \land \cons x $
    ),
    $$\Xi(x,y) \SymbDef \{ 
    \De (x \GenImp y), \De (y \GenImp x) \}. $$
\end{proposition}
\begin{proof}
    To prove that $\PPResImplogicSingle$ is equivalential, it suffices to verify that the following conditions are met~(see \citet[Thm.~6.60]{font:2016}):
    \begin{enumerate}
        \item $\vdash_{\PPResImplogicSingle} \Xi(x,x)$
        \item $x, \Xi(x,y) \vdash_{\PPResImplogicSingle} y$
        \item $\Xi(x,y) \vdash_{\PPResImplogicSingle} \Xi( \circ x, \circ y) $ and $\Xi(x,y) \vdash_{\PPResImplogicSingle} \Xi( \DMNeg x, \DMNeg y) $
        \item $\Xi(x_1,y_1), \Xi(x_2,y_2) \vdash_{\PPResImplogicSingle} \Xi( x_1 \DefCon x_2,  y_1 \DefCon y_2) $ for every binary connective $\DefCon$ in the signature.
    \end{enumerate}
    The verification of the above conditions
    is straightforward from the matrix characterization
    of $\PPResImplogicSingle$.
    Now, if $\PPResImplogicSingle$ were algebraizable, then each set of designated elements would be equationally definable
    on the corresponding reduced matrix~(see \citet[Def.~6.90]{font:2016}).
    But this is not the case, because the matrices 
     $\{ \langle \PPSixResImp, \Upset{a} \rangle : a \in \SixSet \setminus \{ \efvalue \}\}$, all based on the same algebra, are all reduced, and have
     distinct sets of designated elements. 
\end{proof}

The algebraizable companion of $\PPResImplogicSingle$ is the $\top$-assertional logic 
associated to the
variety $\PPImpResVar$, 
 which we denote by $\PPResAsslogic$.
The reduced models of this logic are all matrices
of the form 
$\langle \mathbf A, \{ \top^{\mathbf A} \} \rangle$, where
$\mathbf A$ is an algebra in  $\PPImpResVar$. 
One easily verifies that $\PPResAsslogic$ satisfies the following rules --- any of them being in fact sufficient
to distinguish $\PPResAsslogic$ from $\PPResImplogicSingle$
(given that they are sound in the former but not in the latter):
%
$$
\frac{p}
{\De p} 
\qquad
\frac{p}
{\circ p} 
\qquad
\frac{p \GenImp q  }
{\DMNeg q \GenImp \DMNeg p}
\qquad
\frac{p, \ p \GenImp_{\mathsf{W}} q  }
{q}
$$
where
$p \GenImp_{\mathsf{W}} q \SymbDef  \DMNeg  p \lor \DMNeg \cons  p  \lor q$
and
$\De p \SymbDef p \land \cons p $.

\bigskip

\subsection{Characterizability by single finite PNmatrices  }

We saw in Theorem~\ref{fact:non-single-matrix-pp6imp}
that the logics we introduced in this section
are not characterized by any single logical matrix.
In this subsection, we will demonstrate the
power of partial non-deterministic matrices
by showing that 
both $\PPResImplogicOrderMC$ and
 $\PPResImplogicPrimeMC$
are characterized by single finite PNmatrices.
In consequence, $\PPResImplogicSingle$ will
 be characterized by either of these matrices.
The essential idea is that the collections of matrices
that characterize these logics can be packaged
into a single structure using partiality.

The construction we will present makes use of
the notion of total components of a $\Sigma$-PNmatrix,
which we now  proceed to introduce. 
Let $\PartialMat \SymbDef \langle \AlgA, D \rangle$
be a $\Sigma$-PNmatrix.
For $X \subseteq A$, denote by 
$\PartialMat_X$ the $\Sigma$-PNmatrix $\langle \AlgA_X, D \cap X \rangle$, where $\AlgA_X \SymbDef \langle A \cap X, \cdot^{\AlgA_X} \rangle$ is a $\Sigma$-multialgebra such that 
$\DefCon^{\AlgA_X}(a_1,\ldots,a_k) \SymbDef \DefCon^{\AlgA}(a_1,\ldots,a_k) \cap X$ for all $a_1,\ldots,a_k \in X$, $k \in \omega$
and $\DefCon \in \Sigma_k$.
This PNmatrix is called \emph{the restriction of $\DefMat$ to $X$}.
We say that $X \neq \varnothing$ is a \emph{total component}
of $\PartialMat$ whenever $\PartialMat_X$ is total.
A total component $X$ is \emph{maximal} if adding any other
value to $X$ leads to a component that is not total.
Denote by $\TotComp{\PartialMat}$ the collection of
maximal total components of $\PartialMat$.
Then we have that
$\sequent_{\PartialMat} = {\sequent_{\{ \PartialMat_X : X \in \TotComp{\PartialMat}\}}}$~(\cite{Caleiro2023}).
The latter observation is key to us: the matrices induced by the maximal total components of the PNmatrices we will construct are precisely
the ones in the classes that determine the logics 
$\PPResImplogicOrderMC$ and
 $\PPResImplogicPrimeMC$.

We display below diagrams of the four matrices (see discussion after Proposition~\ref{prop:redmatrdeflog}) that determine
the logic $\PPResImplogicOrderMC$ (in each case, the 
sets of designated elements are highlighted
with an ellipse).
\begin{figure}[H]
    \centering
    \begin{tikzpicture}[node distance=1.25cm]
    \node (0PP6)                      {\HighlDes{$\fvalue$}};
    \node (BPP6)  [above right of=0PP6]  {\HighlDes{$\both$}};
    \node (APP6)  [above left of=0PP6]   {\HighlDes{$\neither$}};
    \node (1PP6)  [above right of=APP6]  {\HighlDes{$\tvalue$}};
    \node (00PP6)  [below=0.4\distance of 0PP6] {$\efvalue$};
    \node (11PP6)  [above=0.4\distance of 1PP6] {\HighlDes{$\etvalue$}};
    \draw (0PP6)   -- (APP6);
    \draw (0PP6)   -- (BPP6);
    \draw (BPP6)   -- (1PP6);
    \draw (APP6)  -- (1PP6);
    \draw (0PP6)   -- (00PP6);
    \draw (1PP6)   -- (11PP6);

    \node[ellipse,
    draw = blue,
    text = orange,
    minimum width = 3.3cm, 
    minimum height = 2.2cm,
    rotate = 90] (e) at (0,1.3) {};
    \end{tikzpicture}
    \qquad 
    \begin{tikzpicture}[node distance=1.25cm]
    \node (0PP6)                      {$\fvalue$};
    \node (BPP6)  [above right of=0PP6]  {\HighlDes{$\both$}};
    \node (APP6)  [above left of=0PP6]   {$\neither$};
    \node (1PP6)  [above right of=APP6]  {\HighlDes{$\tvalue$}};
    \node (00PP6)  [below=0.4\distance of 0PP6] {$\efvalue$};
    \node (11PP6)  [above=0.4\distance of 1PP6] {\HighlDes{$\etvalue$}};
    \draw (0PP6)   -- (APP6);
    \draw (0PP6)   -- (BPP6);
    \draw (BPP6)   -- (1PP6);
    \draw (APP6)  -- (1PP6);
    \draw (0PP6)   -- (00PP6);
    \draw (1PP6)   -- (11PP6);
    \node[ellipse,
    draw = blue,
    text = orange,
    minimum width = 2.7cm, 
    minimum height = 1.3cm,
    rotate = -60] (e) at (0.4,1.8) {};
    \end{tikzpicture}
    \qquad 
    \begin{tikzpicture}[node distance=1.25cm]
    \node (0PP6)                      {$\fvalue$};
    \node (BPP6)  [above right of=0PP6]  {$\both$};
    \node (APP6)  [above left of=0PP6]   {$\neither$};
    \node (1PP6)  [above right of=APP6]  {\HighlDes{$\tvalue$}};
    \node (00PP6)  [below=0.4\distance of 0PP6] {$\efvalue$};
    \node (11PP6)  [above=0.4\distance of 1PP6] {\HighlDes{$\etvalue$}};
    \draw (0PP6)   -- (APP6);
    \draw (0PP6)   -- (BPP6);
    \draw (BPP6)   -- (1PP6);
    \draw (APP6)  -- (1PP6);
    \draw (0PP6)   -- (00PP6);
    \draw (1PP6)   -- (11PP6);
    \node[ellipse,
    draw = blue,
    text = orange,
    minimum width = 1.5cm, 
    minimum height = 1cm,
    rotate = 90] (e) at (0,2.2) {};
    \end{tikzpicture}
    \qquad 
    \begin{tikzpicture}[node distance=1.25cm]
    \node (0PP6)                      {$\fvalue$};
    \node (BPP6)  [above right of=0PP6]  {$\both$};
    \node (APP6)  [above left of=0PP6]   {$\neither$};
    \node (1PP6)  [above right of=APP6]  {$\tvalue$};
    \node (00PP6)  [below=0.4\distance of 0PP6] {$\efvalue$};
    \node (11PP6)  [above=0.4\distance of 1PP6] {\HighlDes{$\etvalue$}};
    \draw (0PP6)   -- (APP6);
    \draw (0PP6)   -- (BPP6);
    \draw (BPP6)   -- (1PP6);
    \draw (APP6)  -- (1PP6);
    \draw (0PP6)   -- (00PP6);
    \draw (1PP6)   -- (11PP6);
    \node[ellipse,
    draw = blue,
    text = orange,
    minimum width = .6cm, 
    minimum height = .6cm,
    rotate = 90] (e) at (0,2.65) {};
    \end{tikzpicture}
\end{figure}

Below we depict the structure of the PNmatrix
$\PartialMat_{\mathsf{up}}$ we propose for this logic, whose principle of construction
is the combination of the above matrices, such that
each of them consists of 
a total component 
of the PNmatrix.
Note: There is nothing special about the dashed edge, it is pictured as dashed only because it crosses
other edges.

\begin{figure}[H]
\[\begin{tikzcd}[ampersand replacement=\&,
execute at end picture={
     \node[ellipse,
    draw = blue,
    text = orange,
    minimum width = 7cm, 
    minimum height = 3cm,
    rotate = -30] (e) at (2.2,.55) {};
}
]
	\&\PartialMat_{\mathsf{up}}\&\& {\HighlDes{\etvalue}} \\
	\& {\tvalue^-} \&\&\&\& {\HighlDes{\tvalue^+}} \\
	{\both^-} \&\& {\neither^-} \&\& {\HighlDes{\both^+}} \&\& {\HighlDes{\neither^+}} \\
	\& {\fvalue^-} \&\&\&\& {\HighlDes{\fvalue^+}} \\
	\&\&\& \efvalue
	\arrow[no head, from=1-4, to=2-2]
	\arrow[no head, from=2-2, to=3-1]
	\arrow[no head, from=2-2, to=3-3]
	\arrow[no head, from=1-4, to=2-6]
	\arrow[no head, from=2-6, to=3-5]
	\arrow[no head, from=2-6, to=3-7]
	\arrow[no head, from=3-1, to=4-2]
	\arrow[no head, from=3-3, to=4-2]
	\arrow[no head, from=3-5, to=4-6]
	\arrow[no head, from=3-7, to=4-6]
	\arrow[no head, from=4-2, to=5-4]
	\arrow[no head, from=4-6, to=5-4]
	\arrow[no head, from=3-3, to=2-6]
	\arrow[no head, from=4-2, to=3-5]
	\arrow[dashed, no head, from=3-1, to=2-6];
\end{tikzcd}\]
\end{figure}

We try to make the above idea clearer in the following pictures,
which show how to identify each of the four matrices inside
$\PartialMat_{\mathsf{up}}$.

\begin{figure}[H]
\[
\resizebox{.5\textwidth}{!}{%
\begin{tikzcd}[ampersand replacement=\&,
execute at end picture={
     \node[ellipse,
    draw = blue,
    text = orange,
    minimum width = 7cm, 
    minimum height = 3cm,
    rotate = -30] (e) at (2.2,.5) {};
}]
	\&\&\& {\HighlDes{\etvalue}} \\
	\& {\NotHighlPic{\tvalue^-}} \&\&\&\& {\HighlDes{\tvalue^+}} \\
	{\NotHighlPic{\both^-}} \&\& {\NotHighlPic{\neither^-}} \&\& {\HighlDes{\both^+}} \&\& {\HighlDes{\neither^+}} \\
	\& {\NotHighlPic{\fvalue^-}} \&\&\&\& {\HighlDes{\fvalue^+}} \\
	\&\&\& \efvalue
	\arrow[no head, from=1-4, to=2-2, lightgray]
	\arrow[no head, from=2-2, to=3-1, lightgray]
	\arrow[no head, from=2-2, to=3-3, lightgray]
	\arrow[no head, from=1-4, to=2-6]
	\arrow[no head, from=2-6, to=3-5]
	\arrow[no head, from=2-6, to=3-7]
	\arrow[no head, from=3-1, to=4-2, lightgray]
	\arrow[no head, from=3-3, to=4-2, lightgray]
	\arrow[no head, from=3-5, to=4-6]
	\arrow[no head, from=3-7, to=4-6]
	\arrow[no head, from=4-2, to=5-4, lightgray]
	\arrow[no head, from=4-6, to=5-4]
	\arrow[no head, from=3-3, to=2-6, lightgray]
	\arrow[no head, from=4-2, to=3-5, lightgray]
	\arrow[dashed, no head, from=3-1, to=2-6, lightgray]
\end{tikzcd}
}
\resizebox{.5\textwidth}{!}{
\begin{tikzcd}[ampersand replacement=\&,
execute at end picture={
     \node[ellipse,
            draw = blue,
            text = orange,
            minimum width = 5cm, 
            minimum height = 2.5cm,
            rotate = -30] (e) at (1.5,1.1) {};
}]
	\&\&\& {\HighlDes{\etvalue}} \\
	\& {\NotHighlPic{\tvalue^-}} \&\&\&\& {\HighlDes{\tvalue^+}} \\
	{\NotHighlPic{\both^-}} \&\& {\neither^-} \&\& {\HighlDes{\both^+}} \&\& {\NotHighlPic{\HighlDes{\neither^+}}} \\
	\& {\fvalue^-} \&\&\&\& {\NotHighlPic{\HighlDes{\fvalue^+}}} \\
	\&\&\& \efvalue
	\arrow[no head, from=1-4, to=2-2, lightgray]
	\arrow[no head, from=2-2, to=3-1, lightgray]
	\arrow[no head, from=2-2, to=3-3, lightgray]
	\arrow[no head, from=1-4, to=2-6]
	\arrow[no head, from=2-6, to=3-5]
	\arrow[no head, from=2-6, to=3-7, lightgray]
	\arrow[no head, from=3-1, to=4-2, lightgray]
	\arrow[no head, from=3-3, to=4-2]
	\arrow[no head, from=3-5, to=4-6, lightgray]
	\arrow[no head, from=3-7, to=4-6, lightgray]
	\arrow[no head, from=4-2, to=5-4]
	\arrow[no head, from=4-6, to=5-4, lightgray]
	\arrow[no head, from=3-3, to=2-6]
	\arrow[no head, from=4-2, to=3-5]
	\arrow[dashed, no head, from=3-1, to=2-6, lightgray]
\end{tikzcd}
}
\]

\vspace{1em}

\[
\resizebox{.5\textwidth}{!}{
\begin{tikzcd}[ampersand replacement=\&,
execute at end picture={
     \node[ellipse,
            draw = blue,
            text = orange,
            minimum width = 4cm, 
            minimum height = 1.5cm,
            rotate = -20] (e) at (1.5,1.8) {};
}]
	\&\&\& {\HighlDes{\etvalue}} \\
	\& {\NotHighlPic{\tvalue^-}} \&\&\&\& {\HighlDes{\tvalue^+}} \\
	{\both^-} \&\& {\neither^-} \&\& {\NotHighlPic{\HighlDes{\both^+}}} \&\& {\NotHighlPic{\HighlDes{\neither^+}}} \\
	\& {\fvalue^-} \&\&\&\& {\NotHighlPic{\HighlDes{\fvalue^+}}} \\
	\&\&\& \efvalue
	\arrow[no head, from=1-4, to=2-2, lightgray]
	\arrow[no head, from=2-2, to=3-1, lightgray]
	\arrow[no head, from=2-2, to=3-3, lightgray]
	\arrow[no head, from=1-4, to=2-6]
	\arrow[no head, from=2-6, to=3-5,lightgray]
	\arrow[no head, from=2-6, to=3-7,lightgray]
	\arrow[no head, from=3-1, to=4-2]
	\arrow[no head, from=3-3, to=4-2]
	\arrow[no head, from=3-5, to=4-6,lightgray]
	\arrow[no head, from=3-7, to=4-6,lightgray]
	\arrow[no head, from=4-2, to=5-4]
	\arrow[no head, from=4-6, to=5-4,lightgray]
	\arrow[no head, from=3-3, to=2-6]
	\arrow[no head, from=4-2, to=3-5,lightgray]
	\arrow[dashed, no head, from=3-1, to=2-6]
\end{tikzcd}
}
\resizebox{.5\textwidth}{!}{
\begin{tikzcd}[ampersand replacement=\&,
execute at end picture={
     \node[ellipse,
            draw = blue,
            text = orange,
            minimum width = 1cm, 
            minimum height = 1cm,
            rotate = -20] (e) at (0,2.4) {};
}]
	\&\&\& {\HighlDes{\etvalue}} \\
	\& {\tvalue^-} \&\&\&\& {\NotHighlPic{\HighlDes{\tvalue^+}}} \\
	{\both^-} \&\& {\neither^-} \&\& {\NotHighlPic{\HighlDes{\both^+}}} \&\& {\NotHighlPic{\HighlDes{\neither^+}}} \\
	\& {\fvalue^-} \&\&\&\& {\NotHighlPic{\HighlDes{\fvalue^+}}} \\
	\&\&\& \efvalue
	\arrow[no head, from=1-4, to=2-2]
	\arrow[no head, from=2-2, to=3-1]
	\arrow[no head, from=2-2, to=3-3]
	\arrow[no head, from=1-4, to=2-6, lightgray]
	\arrow[no head, from=2-6, to=3-5, lightgray]
	\arrow[no head, from=2-6, to=3-7, lightgray]
	\arrow[no head, from=3-1, to=4-2]
	\arrow[no head, from=3-3, to=4-2]
	\arrow[no head, from=3-5, to=4-6, lightgray]
	\arrow[no head, from=3-7, to=4-6, lightgray]
	\arrow[no head, from=4-2, to=5-4]
	\arrow[no head, from=4-6, to=5-4, lightgray]
	\arrow[no head, from=3-3, to=2-6, lightgray]
	\arrow[no head, from=4-2, to=3-5, lightgray]
	\arrow[dashed, no head, from=3-1, to=2-6, lightgray]
\end{tikzcd}
}
\]
\end{figure}

The PNmatrix for $\PPResImplogicPrimeMC$,
which we dub $\PartialMat_{\leq}$, is 
essentially the same, the only difference being the absence of the
dashed line.

After this informal presentation of the general approach,
we proceed to a precise definition of the PNmatrices
$\PartialMat_{\mathsf{up}}$ and $\PartialMat_{\leq}$,
and prove
that they indeed determine, respectively, the logics 
$\PPResImplogicOrderMC$ and
 $\PPResImplogicPrimeMC$.

\begin{definition}
    Let
    $\PartialMatValues \SymbDef \{\efvalue,\fvalue^-,\neither^-,\both^-,\tvalue^-,\fvalue^+,\neither^+,\both^+,\tvalue^+,\etvalue\}$
    and
    $\PartialMatDes \SymbDef \{\fvalue^+,\neither^+,\both^+,\tvalue^+,\etvalue\} \subseteq \PartialMatValues$.
{
Consider the predicates $\mathsf{inc}_{\mathsf{up}}$
and $\mathsf{inc}_{\leq}$
over $\wp(\PartialMatValues)$, where
$\mathsf{inc}_{\mathsf{up}}(X)$ (resp. $\mathsf{inc}_{\leq}(X)$)
means that $X$ is not contained in the image of valuations over
$\PartialMat_{\mathsf{up}}$
(resp. $\PartialMat_{\leq}$),
and can be defined as follows:

}

    \begin{align*}
   \mathsf{inc}_{\mathsf{up}}(X) & \text{ iff }\\
    &X\supseteq Y\text{, for some }Y\in \{\{d^-,d^+\}:d\in \{\fvalue,\neither,\both,\tvalue\}\}\cup\\
   &\qquad\{\{\both^-,\fvalue^+\},
   \{\neither^-,\fvalue^+\},\{\both^+,\tvalue^-\},
   \{\neither^+,\tvalue^-\},
   \{\neither^+,\both^+,\fvalue^-\}\}\\
         {\mathsf{inc}_{\leq}(X)}& \text{ iff }\\
   &X\supseteq Y\text{, for some }Y\in \{\{d^-,d^+\}:d\in \{\fvalue,\neither,\both,\tvalue\}\}\cup\\
   &\qquad\{
   \{\both^-,\fvalue^+\},
   \{\neither^-,\fvalue^+\},\{\both^+,\tvalue^-\},
   \{\neither^+,\tvalue^-\},
   \{\neither^+,\both^+,\fvalue^-\},
   \{\neither^-,\both^-,\tvalue^+\}\}
\end{align*}
Consider also the function $g:\PartialMatValues \to \{\efvalue,\fvalue,\neither,\both,\tvalue,\etvalue\}$
given by $g(\efvalue) \SymbDef \efvalue$, $g(\etvalue) \SymbDef \etvalue$,
$g(a^i) \SymbDef a$ for $a\in \{\fvalue,\neither,\both,\tvalue\}$ and $i \in \{ +,- \}$.
We define the PNmatrices 
    $\PartialMat_{\mathsf{up}} \SymbDef \tuple{\mathbf{A}_{\mathsf{up}},\PartialMatDes}$ and
    $\PartialMat_{\leq} \SymbDef \tuple{\mathbf{A}_{\leq},\PartialMatDes}$,
    with $\mathbf{A}_{\mathsf{up}} \SymbDef \langle \PartialMatValues,\cdot_{\mathsf{up}} \rangle$
    and
    $\mathbf{A}_{\leq} \SymbDef \langle \PartialMatValues,\cdot_{\leq} \rangle$
    such that
\begin{itemize}
    \item 
    for $\DefCon \in \{\DMNeg,\cons\}$,
\begin{align*}
    \DefCon_{\mathsf{up}}(a)& \SymbDef \{b\in \PartialMatValues:\copyright_{\PPSixResImp}(g(a))=g(b) \text{ and not } 
    {\mathsf{inc}_{\mathsf{up}}(\{ a,b \})}\}\\
      \DefCon_\leq(a)&\SymbDef\{b\in \PartialMatValues: \DefCon_{\PPSixResImp}(g(a))=g(b) \text{ and not } 
      {\mathsf{inc}_{\leq}(\{a,b\})}
      \}
\end{align*}
    \item
    for 
 $\DefCon\in \{\land,\lor,\GenImp\}$,
\begin{align*}
    \copyright_{\mathsf{up}}(a,b)& \SymbDef \{c\in \PartialMatValues: \copyright_{\PPSixResImp}(g(a),g(b))=g(c) \text{ and not } 
    {\mathsf{inc}_{\mathsf{up}}(\{a,b,c\})}\}\\
      \copyright_{\leq}(a,b)&\SymbDef\{c\in \PartialMatValues: \copyright_{\PPSixResImp}(g(a),g(b))=g(c)\text{ and not }{\mathsf{inc}_{\leq}(\{a,b,c\})}
      \}
\end{align*}
\end{itemize}
\end{definition}

\begin{proposition}
    $\PPResImplogicOrderMC$ is determined by ${\PartialMat_{\mathsf{up}}}$ and $\PPResImplogicPrimeMC$ 
    is determined by ${\PartialMat_{\leq}}$.

\end{proposition}

\begin{proof}
We have that if $\{a,b,c\} \subseteq X$ and ${\mathsf{inc}_{\mathsf{up}}(\{a,b,c\})}$ then 
$X$ is not contained in a total component of $\PartialMat_{\mathsf{up}}$.
In fact,
 \begin{itemize}
     \item if $b_1,b_2\in X$ 
with $$\{b_1,b_2\}\in \{\{a^-,a^+\}:a\in \{\fvalue,\neither,\both,\tvalue\}\}\cup\{\{\both^-,\fvalue^+\},
   \{\neither^-,\fvalue^+\},\{\both^+,\tvalue^-\},
   \{\neither^+,\tvalue^-\}\}$$
then 
$b_1\land_{\mathsf{up}} b_2=\varnothing$.




    
 \item  if $\{\neither^+,\both^+,\fvalue^-\}\subseteq X$, then 
$\neither^+\land_{\mathsf{up}} \both^+=\{\fvalue^+\}$ and $\fvalue^-\land_{\mathsf{up}} \fvalue^+=\varnothing$.
 \end{itemize}

    The maximal total components of $\PartialMat_{\mathsf{up}}$ are  
    $$\TotComp{{\PartialMat_{\mathsf{up}}}}=\{
    \{\efvalue,\fvalue^-,\neither^-,\both^-,\tvalue^-,\etvalue\},
     \{\efvalue,\fvalue^-,\neither^-,\both^-,\tvalue^+,\etvalue\},
    \{\efvalue,\fvalue^-,\neither^-,\both^+,\tvalue^+,\etvalue\},
    \{\efvalue,\fvalue^+,\neither^+,\both^+,\tvalue^+,\etvalue\}\}$$
since for
every $X\in \TotComp{{\PartialMat_{\mathsf{up}}}}$ and
every $a,b,c\in X$ we have that 
${\mathsf{inc}_{\mathsf{up}}(\{a,b,c\})}$ is not the case.
Thus the restriction of $\PartialMat_{\mathsf{up}}$ to 
$X$ is isomorphic to (that is, it is the same up to renaming of truth values) some matrix with set of designated values 
$D_X = \Upset a$ for $a\in \SixSet$.
The isomorphism is given by the restriction of $g$ to $X$.

Similarly, the maximal total components of $\PartialMat_{\leq}$ are  
    $$\TotComp{{\PartialMat_{\leq}}}=\{
    \{\efvalue,\fvalue^-,\neither^-,\both^-,\tvalue^-,\etvalue\},
    \{\efvalue,\fvalue^-,\neither^-,\both^+,\tvalue^+,\etvalue\},
    \{\efvalue,\fvalue^+,\neither^+,\both^+,\tvalue^+,\etvalue\}\}$$
This is so because 
${\mathsf{inc}_{\leq}(\{a,b,c\})}$ iff 
${\mathsf{inc}_{\mathsf{up}}(\{a,b,c\})}$ or $\{\neither^-,\both^-,\tvalue^+\}\subseteq X$,
and
    the fact that if $\{\neither^-,\both^-,\tvalue^+\}\subseteq X$ then $X$
is not in a total component of $\PartialMat_{\leq}$, since 
 $\neither^-\lor_{\mathsf{up}} \both^-=\{\tvalue^-\}$ and $\tvalue^-\land_{\mathsf{up}} \tvalue^+=\varnothing$.\qedhere
\end{proof}

\section{Hilbert-style axiomatizations 
}
\label{sec:int-hilbert-axiomat}

Our goal in this section is to present 
analytic \SetSet{}
Hilbert-style axiomatizations
for the implicative expansions
$\PPResImplogicOrderMC$
    and
    $\PPResImplogicPrimeMC$,
    as well
    as a \SetFmla{}
    axiomatization for
    $\PPResImplogicSingle$.
Unfortunately, the PNmatrices from the previous section do not allow
us to extract automatically an analytic \SetSet{}
calculus for the corresponding logics
using the technology of~\cite{marcelinowollic},
employed in~\cite{Gomes2022} and in the previous sections of the present paper,
in view of the following:

\begin{proposition}
Neither $\PartialMat_\leq$ nor $\PartialMat_{\mathsf{up}}$ is monadic.
\end{proposition}
\begin{proof}
    Note that no unary formula can separate $\neither^+$ from $\both^+$, nor $\neither^-$ from $\both^-$.
  Indeed, 
  one can show inductively on the structure of a unary formula $\FmA(p)$
  that, given valuations $v_a$, $v_b$ such that $v_a(p)=a^s$ and
  $v_b(p)=b^s$ for $s\in \{+,-\}$,
  we have that
  $v_x(\FmA)\in \{\efvalue,x^s,\etvalue\}$,
  and 
either  $v_a(\FmA)=v_b(\FmA)$ or
$v_a(\FmA)=a^s$
and
$v_b(\FmA)=b^s$.
\end{proof}
    
    We need, therefore, to delve into specific
    details of the 
    logical matrices introduced above, and extract what
    is important to characterize
    their algebraic and logical structures in terms
    of formulas and \SetSet{} rules of inference.
    We will do that and obtain analytic 
    axiomatizations
    for the \SetSet{} logics and then, taking advantage of
    the fact that in $\PPResImplogicPrimeMC$
    we have a disjunction (a notion we will soon make precise),
    we will convert the calculus for $\PPResImplogicPrimeMC$ into a \SetFmla{}
    axiomatization for $\PPResImplogicSingle$.%

\subsection{Analytic axiomatizations for $\PPResImplogicOrderMC$
    and
    $\PPResImplogicPrimeMC$}


    We begin by the analytic \SetSet{} axiomatizations.
    What follows is a succession
    of definitions introducing
    groups of rules of inference
    that capture particular aspects
    of the collections of logical matrices
    determining 
    $\PPResImplogicOrderMC$
    and
    $\PPResImplogicPrimeMC$.
    We check the soundness of each
    of them and ultimately arrive to the desired completeness results.
    Throughout the proofs, we will
    make use of the following abbreviations:
    
\begin{definition} 
Set $\up p \SymbDef \cons(\DMNeg p\GenImp p)$ and $\down p \SymbDef \cons(p\GenImp \DMNeg p)$.
\end{definition}

By way of an example, the truth tables of the above derived connectives
interpreted in $\PPSixResImp$
are given in Table~\ref{tab:up-down-tables}.
A key feature of these connectives
is that they characterize the values $\fvalue$
and $\tvalue$ in the following sense:
for each valuation $v$ over $\langle\PPSixResImp, D\rangle$,
$v(\up p) \not\in D$ iff
$v(p) = \fvalue$
and
$v(\down p) \not\in D$
iff $v(p) = \tvalue$,
for $D$ containing $\etvalue$
and not containing $\efvalue$.

\begin{table}
    \centering
    \begin{tabular}{c|c}
         & $\up p$  \\
         \hline
         $\efvalue$& $\etvalue$\\
         $\fvalue$& $\efvalue$\\
         $\neither$& $\etvalue$\\
         $\both$& $\etvalue$\\
         $\tvalue$& $\etvalue$\\
         $\etvalue$&$\etvalue$ 
    \end{tabular}
    \qquad
        \begin{tabular}{c|c}
         & $\down p$  \\
         \hline
         $\efvalue$& $\etvalue$\\
         $\fvalue$& $\etvalue$\\
         $\neither$& $\etvalue$\\
         $\both$& $\etvalue$\\
         $\tvalue$& $\efvalue$\\
         $\etvalue$& $\etvalue$ 
    \end{tabular}
    \caption{Truth tables of the connectives $\up$ and $\down$ interpreted in $\PPSixResImp$.}
    \label{tab:up-down-tables}
\end{table}

Note that the above definitions
introduce new connectives
by means of abbreviations,
so we have
$\sub(\up p) = \{p,\DMNeg p,\DMNeg p\GenImp p, \cons(\DMNeg p\GenImp p)\}$
(the case of $\sub(\down p)$ is similar).


\begin{definition}
Let $\CalcVar_{\Diamond}$
be the \SetSet{}
calculus given by
the following inference rules:
$$
\inferx[\RuleA^\Diamond_{\up\mathsf{or}\down}]{}{\up p,\down p}
\qquad
\inferx[{\RuleA_\mathsf{id}^{\Diamond}}]{}{\cons(p\GenImp p)}
\qquad
\inferx[\RuleA_\mathsf{trans}^\Diamond]{\cons(p\GenImp q),\cons(q\GenImp r)}{\cons q,\cons(p\GenImp r)}
$$ 

$$\inferx[\RuleA^\Diamond_{\leq\tvalue}]{}{\down p,\cons q,\cons( q\GenImp p)} \qquad 
\inferx[\RuleA^\Diamond_{\geq\fvalue}]{}{\up p,\cons( p\GenImp q)} 
$$


$$\inferx[\RuleA^\Diamond_{\mathsf{incclass}_1}]{\up p,\cons(p\GenImp q)}{\cons p,\up q}
\qquad 
\inferx[\RuleA^\Diamond_{\mathsf{incclass}_2}]{\down q,\cons(p\GenImp q)}{\cons q,\down p}
\qquad
\inferx[\RuleA^\Diamond_{\mathsf{incclass}_3}]{\up p,\down q,\cons (p\GenImp q)}{\cons q,\cons (q\GenImp p)}
$$

$$
\inferx[\RuleA^\Diamond_{\mathsf{just}_2}]{\down p,\up r}{\cons p,\cons(p\GenImp  q),\cons( p\GenImp r),\cons(q\GenImp  r)}
$$

\end{definition}

\begin{proposition}
    The rules of $\CalcVar_\Diamond$ are sound for any matrix $\langle \PPSixResImp, D \rangle$ with $D$
    containing $\etvalue$ and not containing $\efvalue$.
\end{proposition}
\begin{proof}
We proceed rule by rule.
    \begin{description}
        \item[$\RuleA^\Diamond_{\up\mathsf{or}\down}$:]
 If $v(\up p)\notin D$, then $v(\up p)=\efvalue$, so $v(p)=\fvalue$ and thus $v(\down p)=\etvalue\in D$.
        
        \item[${\RuleA^\Diamond_\mathsf{id}}$:] Clearly,
        $v(\cons(p\GenImp p))=\etvalue\in D$.  
 
        \item[$\RuleA^\Diamond_\mathsf{trans}$:]
If $v(\cons q)\notin D$, we have $v(q)\neq \efvalue$,
hence $v(\cons(p\GenImp q))=\etvalue$  and $v(\cons(q\GenImp r))=\etvalue$,
thus $v(p)\leq v(q)\leq v(r)$,
and therefore $v(\cons(p \GenImp r))=\etvalue$.

\item[$\RuleA^\Diamond_{\leq\tvalue}$:]
If $v(\down p)\notin D$, then $v(p)=\tvalue$.
If $v(\cons q)\notin D$, then $\fvalue\leq v(q)\leq \tvalue$
and therefore $v(\cons(q\GenImp p))=\etvalue$.

\item[$\RuleA^\Diamond_{\geq\fvalue}$:]
If $v(\up p)\notin D$ then $v(p)=\fvalue$
and therefore $v(\cons(p \GenImp q))=\etvalue$.

\item[$\RuleA^\Diamond_{\mathsf{incclass}_1}$:]
If $v(\up p)\in D$ and $v(\cons p)\notin D$ then
$v(p)\in \{\both,\neither,\tvalue\}$
and if 
$v(\cons(p\GenImp q))\in D$ then 
$v(q)\in \{\efvalue,\both,\neither,\tvalue,\etvalue
\}$
and therefore 
$v(\up q)=\etvalue\in D$.

\item[$\RuleA^\Diamond_{\mathsf{incclass}_2}$:]

If $v(\down q)\in D$ and $v(\cons q)\notin D$ then
$v(q)\in \{\fvalue,\both,\neither\}$ and if 
$v(\cons(p \GenImp q))\in D$ then 
$v(p)\in \{\efvalue,\fvalue,\both,\neither\}$
and therefore 
$v(\down p)=\etvalue\in D$.

\item[$\RuleA^\Diamond_{\mathsf{incclass}_3}$:]
If $v(\up p),v(\down q)\in D$ and 
$v(\cons q)\notin D$ then 
$v(p) \neq \fvalue$
and
$v(q)\in \{\fvalue,\both,\neither\}$. Further, if 
$v(\cons(p \GenImp q))\in D$
then $v(p) \leq v(q)$ and so
$v(q \GenImp p)\in D$.


\item[$\RuleA^\Diamond_{\mathsf{just}_2}$:]
If $v(\down p),
v(\up r)\in D$ and 
$v(\cons p)\notin D$
then $v(p)\in \{\fvalue,\both,\neither\}$
and $v(r) \neq \fvalue$.
If $v(\cons(p \GenImp q)),v(\cons(p \GenImp q)) \not\in D$,
then $v(p \GenImp q), v(q \GenImp r) 
\in \{ \fvalue, \both, \neither, \tvalue\}$.
Then $v(p), v(q),v(r) \in \{\both, \neither \}$.
Hence, either $v(p)=v(q)$, in which case
$v(p\GenImp q)=\etvalue\in D$,
or 
$v(q)=v(r)$, in which case
$v(q \GenImp r)=\etvalue\in D$
(both cases contradicting the assumptions)
or 
$v(p)=v(r)$, in which case
$v(p \GenImp r)=\etvalue\in D$, as desired.
\qedhere

 
    \end{description}
\end{proof}

\begin{definition}
    Let $\CalcVar_{\GenImp}$
be the \SetSet{}
calculus given by
the following inference rules:
$$
    \inferx[\RuleA^\GenImp_1]{\cons q}{\cons(p \GenImp q)}
    \qquad
    \inferx[\RuleA^\GenImp_2]{q}{p \GenImp q}
    \qquad
    \inferx[\RuleA^\GenImp_3]{p, p \GenImp q}{q}
    \qquad
    \inferx[\RuleA^\GenImp_4]{\cons p, p,\cons(p \GenImp q)}{\cons q}
    \qquad
    \inferx[\RuleA^\GenImp_5]{\cons p, p,\down(p \GenImp q)}{\down q}
$$
$$
 \inferx[\RuleA^\GenImp_6]{\cons p, p,\up(p \GenImp q)}{\up q}
 \qquad
 \inferx[\RuleA^\GenImp_7]{\up q}{\up(p \GenImp q)}
 \qquad
 \inferx[\RuleA^\GenImp_8]{\down q}{\down(p \GenImp q)}
 \qquad
 \inferx[\RuleA^\GenImp_9]{}{\cons(q \GenImp (p \GenImp q))}
 \qquad
 \inferx[\RuleA^\GenImp_{10}]{\cons p}{p,\cons(p \GenImp q)}
$$
$$
\inferx[\RuleA^\GenImp_{11}]{\cons p}{p, p \GenImp q}
\qquad
\inferx[\RuleA^\GenImp_{12}]{\cons q,p \GenImp q}{q,\cons p}
\qquad
\inferx[\RuleA^\GenImp_{13}]{\cons(p \GenImp q)}{\cons q, p\GenImp q}
\qquad
\inferx[\RuleA^\GenImp_{14}]{}{\down p,
                \cons(p \GenImp q),
                \cons((p \GenImp q) \GenImp q)}
$$
$$
\inferx[\RuleA^\GenImp_{15}]{\up p,\cons(p\GenImp q)}{\cons p,\up q}
\qquad
\inferx[\RuleA^\GenImp_{16}]{\down p}{\cons p,\up(p\GenImp q)}
\qquad 
\inferx[\RuleA^\GenImp_{17}]{}{\cons p,\down(p\GenImp q)}
$$
$$
\inferx[\RuleA^\GenImp_{18}]{\up p,
                \cons(p \GenImp (p \GenImp q))}{\cons p,\up q}
\qquad
\inferx[\RuleA^\GenImp_{19}]{\up q}{\cons(p\GenImp q),\cons((p \GenImp q) \GenImp q)}
$$
\end{definition}

\begin{proposition}
  The rules of $\CalcVar_\GenImp$ are sound for any matrix $\langle \PPSixResImp, D \rangle$ with $D = \Upset a$ for some
  $a > \efvalue$.
\end{proposition}
\begin{proof}
We proceed rule by rule.

    \begin{description}
        \item[{$\RuleA^{\GenImp}_{1}$}:]
            If $v(\cons q)\in D$ then $v(q)\in \{\etvalue,\efvalue\}$,
        hence $v(\cons (p\GenImp q))\in D$.
        \item[{$\RuleA^{\GenImp}_{2}$}:]
                    If $v(q)\in D$, then $v(q)\geq a$
            and by analysing the table of $\ResImp$
            we conclude that $v(p\GenImp q)\geq v(q)\geq a$ and
            thus $v(p \GenImp q)\in D$.
        \item[{$\RuleA^{\GenImp}_{3}$}:]
            If $v(p),v(p\GenImp q)\in D$
                then $v(p)\geq a$ and $v(p\GenImp q)\geq a$,
                and by analysing the table of $\ResImp$
                we conclude that $v(q)\geq a$ and
                thus $v(q)\in D$.
        \item[{$\RuleA^{\GenImp}_{4}$}:]
        If $v(\cons p),v(p) \in D$
        and $v(\cons q) \not\in D$,
        then $v(p) = \etvalue$
        and $v(q) \in \{ \fvalue, \both, \neither, \tvalue \}$
        and $v(p \GenImp q) \in \{ \fvalue, \both, \neither, \tvalue \}$, so
        $v(\cons(p \GenImp q)) = 
        \efvalue \not\in D$.
        \item[{$\RuleA^{\GenImp}_{5}$}:]
        If $v(\cons p),v(p) \in D$
        and $v(\down q) \not\in D$,
        then $v(p) = \etvalue$
        and $v(q) = \tvalue$,
        thus $v(p \GenImp q) = \tvalue$ and
        we are done.
        \item[{$\RuleA^{\GenImp}_{6}$}:]
        If $v(\cons p),v(p) \in D$
        and $v(\up q) \not\in D$,
        then $v(p) = \etvalue$
        and $v(q) = \fvalue$,
        thus $v(p \GenImp q) = \fvalue$ and
        we are done.
        \item[{$\RuleA^{\GenImp}_{7}$}:]
        If $v(\up q) \in D$,
        then $v(q) \neq \fvalue$.
        But then $v(p \GenImp q) \neq \fvalue$,
        and $v(\up(p \GenImp q)) = \etvalue$.
        \item[{$\RuleA^{\GenImp}_{8}$}:]
        Similar to the
        proof for {$\RuleA^{\GenImp}_{7}$}.
        \item[{$\RuleA^{\GenImp}_{9}$}:]
        If $v(\cons(p \GenImp (q \GenImp p))) 
        \not\in D$,
        then $v(p \GenImp (q \GenImp p)) \in \{ \fvalue,\both,\neither,\tvalue \}$.
        Then $v(p) \in \{\both,\neither,\tvalue,\etvalue \}$
        and $v(q \GenImp p) \in \{ \fvalue,\both,\neither,\tvalue \}$.
        So, $v(q) \in \{\both,\neither,\tvalue,\etvalue \}$
        and 
        $v(p) \in \{\both,\neither,\tvalue \}$.
        If $v(p) = \both$ and $v(q \GenImp p) = \both$, then $v(p \GenImp (q \GenImp p)) = \etvalue$.
        If $v(p) = \neither$,
        the proof is similar.
        If $v(p) = \tvalue$,
        then $v(q \GenImp p) = \etvalue$
        or $v(q) = \etvalue$.
        In all cases we reach a contradiction.
        \item[{$\RuleA^{\GenImp}_{10}$, 
        $\RuleA^{\GenImp}_{11}$}:]
        If $v(p),v(\cons p)\notin D$,
            then $v(p)=\efvalue$,
            thus $v(p\GenImp q)=v(\cons(p\GenImp q))=\etvalue \in D$.
        \item[{$\RuleA^{\GenImp}_{12}$}:]
        If $v(\cons q) \in D$
        and $v(q) \not\in D$,
        then $v(q) = \efvalue$.
        If $v(\cons p) \not\in D$,
        then
        $v(p) \in \{ \fvalue,\both,\neither,\tvalue \}$.
        But then $v(p \GenImp q) = \fvalue \not\in D$.
        \item[{$\RuleA^{\GenImp}_{13}$}:]
        If $v(\cons(p \GenImp q)) \in D$,
        then $v(p \GenImp q) \in \{ \efvalue,\etvalue \}$.
        Further, if $v(\cons q) \not\in D$,
        then $v(q) \in \{ \fvalue,\both,\neither,\tvalue \}$,
        thus
         $v(p \GenImp q) = \etvalue$.
        \item[{$\RuleA^{\GenImp}_{14}$}:]
        If $v(\cons(p \GenImp q)) \not\in D$,
        we have $v(p \GenImp q) \in \{ \fvalue,\both,\neither,\tvalue \}$.
        If $v(\down p) \not\in D$,
        then $v(p) = \tvalue$,
        and $v(p \GenImp q) = v(q)$,
        and we are done.
        \item[{$\RuleA^{\GenImp}_{15}$}:]
        If $v(\cons p) \not\in D$,
        then $v(p) \in \{\both, \neither, \tvalue\}$.
        If $v(\up q) \not\in D$,
        then $v(q) = \fvalue$.
        Thus $v(p \GenImp q) \in \{ \both,\neither,\fvalue \}$,
        so $v(\cons(p \GenImp q)) = \efvalue$.
        \item[{$\RuleA^{\GenImp}_{16}$}:]
        If $v(\cons p) \not\in D$
        and $v(\down p) \in D$,
        then $v(p) \in \{ \fvalue,\both,
        \neither \}$
        and $v(p \GenImp q) \neq \fvalue$,
        and thus $v(\up(p \GenImp q)) = \etvalue$.
        \item[{$\RuleA^{\GenImp}_{17}$}:]
        If $v(\cons p) \not\in D$, then $ v(p) \in \{ \fvalue,\both,\neither,\tvalue \}$.
        Thus $v(p \GenImp q) \neq \tvalue$,
        and so $v(\down(p \GenImp q)) = \etvalue$.
        \item[{$\RuleA^{\GenImp}_{18}$}:] 
        If $v(\up p) \in D$
        and $v(\cons p) \not\in D$,
        then $v(p) \in \{ \both,\neither,\tvalue \}$. If $v(\up q) \not\in D$,
        then $v(q) = \fvalue$.
        If $v(p) = \both$,
        then $v(p \GenImp q) = \neither$,
        and then $v(p \GenImp (p \GenImp q)) = \neither$.
        Similarly, we have
        $v(p \GenImp (p \GenImp q)) = \both$
        if $v(p) = \neither$.
        If $v(p) = \tvalue$,
        then $v(p \GenImp q) = \fvalue$,
        and $v(p \GenImp (p \GenImp q)) = \fvalue$.
        In any case, $v(\cons(p \GenImp (p \GenImp q))) = \efvalue$.
        \item[{$\RuleA^{\GenImp}_{19}$}:]
        If $v(\cons(p \GenImp q)) \not\in D$,
        we have $v(p \GenImp q) \in \{ \fvalue,\both,\neither,\tvalue \}$.
        If $v(\up q) \in D$,
        we have $v(q) \neq \fvalue$.
        The case $v(p \GenImp q) = \fvalue$
        is impossible.
        If $v(p \GenImp q) = \both$,
        we have $v(q) = \both$,
        and clearly $v(\cons((p \GenImp q) \GenImp q)) = \etvalue$.
        The case $v(p \GenImp q) = \neither$ 
        is
        similar to the previous case.
        If $v(p \GenImp q) = \tvalue$,
        then $v(p) = \etvalue$, thus $v(\cons((p \GenImp q) \GenImp q)) = \etvalue$.
        \qedhere
    \end{description}
\end{proof}

\begin{definition}
        Let $\CalcVar_{\DMNeg}$
be the \SetSet{}
calculus given by
the following inference rules:
{
$$\inferx[\RuleA^\DMNeg_1]{\cons p}{p, \DMNeg p}\qquad
\inferx[\RuleA^\DMNeg_2]{\cons p,p, \DMNeg p}{}\qquad
\inferx[\RuleA^\DMNeg_3]{\cons p}{\cons \DMNeg p}\qquad
\inferx[\RuleA^\DMNeg_4]{\cons \DMNeg p}{\cons p}
$$
 
$$\inferx[{\RuleA^\DMNeg_5}]{\up \DMNeg p}{\down p}\qquad
\inferx[\RuleA^\DMNeg_6]{\down \DMNeg p}{\up p}\qquad
\inferx[\RuleA^\DMNeg_7]{\down p}{\up \DMNeg p}\qquad
\inferx[\RuleA^\DMNeg_8]{\up p}{\down \DMNeg p}
$$}

\end{definition}

\begin{proposition}
  The rules of $\CalcVar_\DMNeg$ are sound for any matrix $\langle \PPSixResImp, D \rangle$ with $D = \Upset a$ for some
  $a>\efvalue$.
\end{proposition}

\begin{proof}
    We proceed rule by rule.
    
    \begin{description}
        \item[${\RuleA^\DMNeg_1}$:]
 If $v(\cons p)\in D$ and $v(p)\notin D$
 then $v(p)=\efvalue$ and
 $v(\DMNeg p)=\etvalue\in D$.
 
 \item 
 [${\RuleA^\DMNeg_2}$:]
If $v(p),v(\cons p)\in D$
then $v(p)=\etvalue$ and $v(\DMNeg p)=\efvalue\notin D$.

\item 
[${\RuleA^\DMNeg_3}$:]
If $v(\cons p)\in D$ then $v(p),v(\DMNeg p)\in \{\efvalue,\etvalue\}$ and so 
$v(\cons\DMNeg p) = \etvalue\in D$.

\item[$\RuleA^\DMNeg_4$:]
If $v(\cons \DMNeg p)\in D$ then $v(\DMNeg p),v( p)\in \{\efvalue,\etvalue\}$ and so 
$v(\cons p) = \etvalue\in D$.
 
\item 
[${{\RuleA^\DMNeg_5}}$:]
If $v(\down p)\notin D$ then $v(p)=\tvalue$,
so $v(\DMNeg p)=\fvalue$ and hence $v(\up \DMNeg p)=\efvalue\notin D$.

\item 
[${\RuleA^\DMNeg_6}$:]
If $v(\up p)\notin D$ then $v(p)=\fvalue$,
so $v(\DMNeg p)=\tvalue$ and hence $v(\down \DMNeg p)=\efvalue\notin D$.

\item 
[${\RuleA^\DMNeg_7}$:]
If $v(\down p)\in D$ then $v(p)\neq \tvalue$,
hence $v(\DMNeg p)\neq \fvalue$ and so
$v(\up \DMNeg p)=\etvalue\in D$.

\item 
[${\RuleA^\DMNeg_8}$:]
If $v(\up p)\in D$ then $v(p)\neq \fvalue$,
hence $v(\DMNeg p)\neq \tvalue$ and so
$v(\down \DMNeg p)=\etvalue\in D$.\qedhere

    \end{description}
\end{proof}

\begin{definition} 
Let $\CalcVar_{\cons}$
be the \SetSet{}
calculus given by
the following inference rule:
$$\inferx[\RuleA_{\cons}]{}{\cons \cons p}
$$
\end{definition}

\begin{proposition}
          The rules of $\CalcVar_\cons$ are sound for any matrix $\langle \PPSixResImp, D \rangle$ 
          with $\etvalue \in D$. 
\end{proposition}
\begin{proof}
Easily,
    $v(\circ \circ p)=\etvalue\in D$.
\end{proof}

\begin{definition}
Let $\CalcVar_{\land}$
be the \SetSet{}
calculus given by
the following inference rules:

$$
\inferx[\RuleA^{\land}_1]{\cons p,\cons q}{\cons(p \land q)}
\qquad
\inferx[\RuleA^{\land}_2]{p,q}{p\land q}
\qquad
\inferx[\RuleA^{\land}_3]{p \land q}{q}
\qquad
\inferx[\RuleA^{\land}_4]{p, \cons(p \land q)}{\cons q}
\qquad
\inferx[\RuleA^{\land}_5]{\cons p}{\cons(q \GenImp (p \land q))}
$$
$$
\inferx[\RuleA^{\land}_6]{}{\cons((p \land q) \GenImp q)}
\qquad
\inferx[\RuleA^{\land}_7]{p \land q}{p}
\qquad
\inferx[\RuleA^{\land}_8]{q, \cons(p \land q)}{\cons p}
\qquad
\inferx[\RuleA^{\land}_9]{\cons q}{\cons(p \GenImp (p \land q))}
$$
$$
\inferx[\RuleA^{\land}_{10}]{}{\cons((p \land q) \GenImp p)}
\qquad
\inferx[\RuleA^{\land}_{11}]{\cons p}{p,\cons(p \land q)}
\qquad
\inferx[\RuleA^{\land}_{12}]{\cons q}{q,\cons(p \land q)}
\qquad
\inferx[\RuleA^{\land}_{13}]{\cons(p \land q)}{\cons p, \cons q}
$$
$$
\inferx[\RuleA^{\land}_{14}]{\cons(p \GenImp q)}{\cons(p \GenImp (p \land q))}
\qquad
 \inferx[\RuleA^{\land}_{15}]{\cons(q \GenImp p)}{\cons(q \GenImp (p \land q))}
 \qquad
 \inferx[\RuleA^{\land}_{16}]
        {\down p,\up(p \land q)}
        {\cons p,\cons(p \GenImp q)}
$$

\end{definition}

\begin{proposition}
  The rules of $\CalcVar_\land$ are sound for any matrix $\langle \PPSixResImp, D \rangle$ with $D = \Upset a$ for some
  $a>\efvalue$.
\end{proposition}
\begin{proof}
We proceed rule by rule:

    \begin{description}
        \item[$\RuleA^{\land}_{1}$:]
        If $v(\cons p), v(\cons q) \in D$,
        then $v(p),v(q) \in \{ \efvalue, \etvalue \}$, thus
        $v(p \land q) \in \{ \efvalue, \etvalue \}$, and $v(\cons(p \land q)) = \etvalue$.
        \item[$\RuleA^{\land}_{2}$:]
        Soundness follows from the fact that
        principal filters are
        closed under meets.
        \item[$\RuleA^{\land}_{3}$,
        $\RuleA^{\land}_{7}$:]
         If $v(p\land q)\geq a$, then from $v(p) \geq  v(p\land q)$ and $v(q)\geq v(p\land q)$
    we have that $v(p),v(q)\in D$.
        \item[$\RuleA^{\land}_{4}$, $\RuleA^{\land}_{8}$:]
        If  $v(\cons q) \not\in D$,
        then $v(q) \in \{ \fvalue, \both, \neither, \tvalue \}$.
        If $v(\cons(p \land q)) \in D$,
        then $v(p \land q) \in \{ \efvalue, \etvalue \}$.
        Thus if $v(p \land q) = \etvalue$,
        then $v(p) = v(q) = \etvalue$,
        absurd.
        Otherwise, we must have
        $v(p) = \efvalue$.
        The proof is analogous for the other rule.
        \item[$\RuleA^{\land}_{5}$,$\RuleA^{\land}_{9}$:]
        If $v(\cons p) \in D$,
        then $v(p) \in 
        \{ \efvalue, \etvalue \}$.
        Then either $v(q \GenImp (p \land q)) = v(q \GenImp q) = \etvalue$
        or $v(q \GenImp (p \land q)) \in \{ \efvalue,\etvalue \}$.
        In both cases we are done.
        The proof is analogous for the other rule.
        \item[$\RuleA^{\land}_{6}$, $\RuleA^{\land}_{10}$:]
        We have that $v(p\land q)\leq v(p)$ and $v(p\land q)\leq v(q)$
        hence $v((p\land q)\GenImp q)=v(\cons((p\land q)\GenImp q))=v((p\land q)\GenImp p)=v(\cons((p\land q)\GenImp p))=\etvalue\in D$.
        \item[$\RuleA^{\land}_{11}$, $\RuleA^{\land}_{12}$:]
         If $v(\cons p)\in D$ and $v(p)\notin D$, or $v(\cons q)\in D$ and $v(q)\notin D$, 
         then either  
        $v(p)=\efvalue$ or $v(q)=\efvalue$.
        In any case we have that $v(p\land q)=\efvalue$ and $v(\cons(p\land q))=\etvalue\in D$.
        \item[$\RuleA^{\land}_{13}$:]
        If $v(\cons(p \land q)) \in D$,
        then $v(p \land q) \in \{ \efvalue, \etvalue \}$.
        If $v(p) = v(q) = \etvalue$,
        clearly $v(\cons p) = v(\cons q) = \etvalue$.
        If $v(p) = \efvalue$, then
        $v(\cons p) = \etvalue$.
        Analogously if $v(q) = \efvalue$.
        \item[$\RuleA^{\land}_{14}$, $\RuleA^{\land}_{15}$:] 
        Suppose $\cons(p\GenImp q)\in D$.
    Then either
    $v(q)=\efvalue$, in 
    which case 
    $v((p\lor q)\GenImp q),v((q\lor p)\GenImp q)\in \{\efvalue,\etvalue\}$ 
    and thus $v(\cons((p\lor q)\GenImp q))=v(\cons((q\lor p)\GenImp q))=\etvalue\in D$;
    or $v(p)\leq v(q)$, and so $v(p\lor q)=v(q\lor p)=v(q)$, hence $v(\cons((p\lor q)\GenImp q))=v(\cons((q\lor p)\GenImp q))=\etvalue\in D$.
        \item[$\RuleA^{\land}_{16}$:]
        If $v(\cons p) \not\in D$
        and $v(\down p) \in D$,
        then $v(p) \in \{ \fvalue, \both, \neither \}$.
        If $v(\up(p \land q)) \in D$,
        then $v(p \land q) \neq \fvalue$.
        We may safely focus on cases in which
        $v(p) \neq v(q)$.
        If $v(p) = \fvalue$,
        then $v(p \GenImp q) \in \{ \efvalue,\etvalue \}$.
        If $v(p) = \both$,
        we have $v(q) \not\in \{ \fvalue, \both, \neither\}$ from the above assumptions,
        and this gives $v(p \GenImp q) \in \{ \efvalue,\etvalue \}$.
        Similarly if $v(p) = \neither$.
        \qedhere
    \end{description}
\end{proof}

\begin{definition}
Let $\CalcVar_{\lor}$
be the \SetSet{}
calculus given by
the following inference rules:
$$
    \inferx[\RuleA_1^{\lor}]{\cons p,\cons q}{\cons (p \lor q)}
    \qquad
    \inferx[\RuleA_2^{\lor}]{\cons p, p \lor q}{p, q}
    \qquad
    \inferx[\RuleA_3^{\lor}]{q}{p \lor q}
    \qquad
    \inferx[\RuleA_4^{\lor}]{\cons(p \lor q)}{p, \cons q}
    \qquad
    \inferx[\RuleA_5^{\lor}]{}{\cons(q \GenImp (p \lor q))}
$$
$$
    \inferx[\RuleA_6^{\lor}]{\cons p}{p, \cons((p \lor q) \GenImp q)}
    \qquad
    \inferx[\RuleA_7^{\lor}]{p}{p \lor q}
    \qquad
    \inferx[\RuleA_8^{\lor}]{\cons(p \lor q)}{q,\cons p}
    \qquad 
    \inferx[\RuleA_9^{\lor}]{}{\cons(p \GenImp (p \lor q))}
$$
$$
\inferx[\RuleA_{10}^{\lor}]{\cons q}{q, \cons((p \lor q) \GenImp p)}
\qquad
\inferx[\RuleA_{11}^{\lor}]{p,\cons p}{\cons(p\lor q)}
\qquad
\inferx[\RuleA_{12}^{\lor}]{q,\cons q}{\cons(p\lor q)}\qquad
\inferx[\RuleA_{13}^{\lor}]{\cons(p \lor q)}{\cons p, \cons q}
$$
$$
\infer[\RuleA_{14}^{\lor}]{\cons((p \lor q) \GenImp q)}{\cons(p \GenImp q)}\qquad
\inferx[\RuleA_{15}^{\lor}]{\cons(q \GenImp p)}{\cons((p \lor q) \GenImp p)}
\qquad
\inferx[\RuleA_{16}^{\lor}]
        {\up q, \down(p \lor q)}
        {\cons p,\cons(p \GenImp q)}
$$

\end{definition}

\begin{proposition}
  The rules of $\CalcVar_\lor$ are sound for any matrix $\langle \PPSixResImp, D \rangle$ with $D = \Upset a$ for some
  $a>\efvalue$.
\end{proposition}
\begin{proof}
We proceed rule by rule.
    \begin{description}
        \item[$\RuleA_1^\lor$:]
        If $v(\cons p), v(\cons q) \in D$, then $v(p), v(q) \in \{ \efvalue, \etvalue \}$,
        thus
        $v(p \lor q) \in \{ \efvalue, \etvalue \}$, and we are done.
        \item[$\RuleA_2^\lor$:]
        If $v(\cons p) \in D$
        and $v(p) \not\in D$,
        $v(p) = \efvalue$.
        Then $v(p \lor q) = v(q)$
        and we are done.
        \item[$\RuleA_3^\lor$, $\RuleA_7^\lor$:]
        As $v(q) \leq v(p \lor q)$,
        if $v(q) \in D$, then
        $v(p \lor q) \in D$.
        Similarly for the other rule.
        \item[$\RuleA_4^\lor$, $\RuleA_8^\lor$:] 
        If $v(\cons(p \lor q)) \in D$,
        then $v(p \lor q) \in \{ \efvalue,\etvalue \}$.
        If $v(\cons q) \not\in D$,
        $v(q) \in \{ \fvalue,\both,\neither,\tvalue \}$.
        Then $v(p) = \etvalue$.
        Similarly for the other rule.
        \item[$\RuleA_5^\lor$, $\RuleA_9^\lor$:]
        We have that $v(p\lor q)\geq v(p)$ and $v(p\lor q)\geq v(q)$,
    hence $v(q\GenImp (p\lor q))=
    v(\cons(q\GenImp (p\lor q)))=v(p\GenImp (p\lor q))=
    v(\cons(p\GenImp (p\lor q)))=\etvalue \in D$.
     \item[$\RuleA_{6}^\lor$, $\RuleA_{10}^\lor$:]
     If $v(\cons p)\in D$
        and $v(p)\notin D$
        then $v(p)=\efvalue$ and
        $v(p\lor q)=v(q)$, so $v((p\lor q)\GenImp q) = \etvalue$.
        Similarly if $v(\cons q)\in D$
        and $v(q)\notin D$.
        \item[$\RuleA_{11}^\lor$, $\RuleA_{12}^\lor$:]
        If $v(p),v(\cons p)\in D$, or $v(q),v(\cons q)\in D$, 
     then either  
    $v(p)=\etvalue$ or $v(q)=\etvalue$.
    In any case we have that $v(p\lor q)=v(\cons(p\lor q))=\etvalue\in D$.
    \item[$\RuleA_{13}^\lor$:]
    If $v(\cons p),v(\cons q) \not\in D$, we have $v(p),v(q) \in \{ \fvalue,\both,\neither,\tvalue \}$,
    thus $v(p \lor q) \in \{ \fvalue,\both,\neither,\tvalue \}$,
    and $v(\cons(p \lor q)) \not\in D$.
    \item[$\RuleA_{14}^\lor$, $\RuleA_{15}^\lor$:]
    Suppose $\cons(p\GenImp q)\in D$.
Then either $v(q)=\efvalue$, in 
which case 
$v((p\lor q)\GenImp q),v((q\lor p)\GenImp q)\in \{\efvalue,\etvalue\}$ 
and thus $v(\cons((p\lor q)\GenImp q))=v(\cons((q\lor p)\GenImp q))=\etvalue\in D$;
or $v(p)\leq v(q)$, so $v(p\lor q)=v(q\lor p)=v(q)$, hence $v(\cons((p\lor q)\GenImp q))=v(\cons((q\lor p)\GenImp q))=\etvalue\in D$.
    \item[$\RuleA_{16}^\lor$:]
    If $v(\cons p) \not\in D$
    and $v(\up q) \in D$,
    $v(p) \in \{ \fvalue,\both,\neither,\tvalue \}$
    and $v(q) \neq \fvalue$.
    If $v(p) = v(q)$, we 
    are done, so suppose
    $v(p) \neq v(q)$.
    If $v(\down(p \lor q)) \in D$,
    then $v(p \lor q) \neq \tvalue$.
    If $v(p) = \fvalue$,
    then $v(p \GenImp q) \in \{ \efvalue,\etvalue \}$, and we are done.
    If $v(p) = \both$,
    the only non-obvious case
    is $v(q) = \neither$,
    but this is impossible
    as $v(p \lor q) \neq \tvalue$.
    If $v(p) = \neither$,
    the proof is similar to
    the previous case.
    If $v(p) = \tvalue$,
    the only non-obvious
    cases are those in which $v(q) \in \{ \both,\neither \}$; yet 
    they are impossible, again,
    because $v(p \lor q) \neq \tvalue$.
    \qedhere
    \end{description}
\end{proof}

\begin{definition}
    Let $\CalcVar_{\bot\top}$
be the \SetSet{}
calculus given by
the following inference rules:
$$
\inferx[\RuleA_1^\top]{}{\top}
\qquad 
\inferx[\RuleA_2^\top]{}{\cons\top}
\qquad 
\inferx[\RuleA_1^\bot]{\bot}{}
\qquad
\inferx[\RuleA_2^\bot]{}{\cons\bot}
$$
\end{definition}

\begin{proposition}
  The rules of $\CalcVar_{\top\bot}$ are sound for any matrix $\langle \PPSixResImp, D \rangle$ with $D = \Upset a$ for some
  $a>\efvalue$.
\end{proposition}
\begin{proof}
Obvious.
\end{proof}

The final rules we introduce encode
the differences between the logics
$\PPResImplogicOrderMC$ and $\PPResImplogicPrimeMC$.

\begin{definition}
Consider the following inference rules:
$$
\inferx[\RuleA_{D_\land}]{p, \down p, q}{\cons p,\cons(p \GenImp q),\cons r, r}
\qquad 
\inferx[\RuleA_{D_\leq}]{p,\cons(p \GenImp q)}{\cons q, q}
\qquad
\inferx[\RuleA_{D\neq \up \tvalue}]{r,\up q}{\down r, \cons (p\GenImp q),p,q}
$$
\end{definition}

\begin{proposition}\label{lem:Drules}
    In 
    the matrix
    $\langle \PPSixResImp, D \rangle$:
    (i) the rules
    $\RuleA_{D_\land}$ and
    $\RuleA_{D_\leq}$ are sound when $D=\up a$ for $a>\efvalue$; (ii) the rule 
    $\RuleA_{D\neq \up \tvalue}$ is sound
    when $D$ is a prime filter; (iii) the rule 
    $\RuleA_{D\neq \up \tvalue}$ is not sound
    when $D = \up \tvalue$.
\end{proposition}
 \begin{proof}
 We first check items (i) and (ii):
 \begin{description}
     \item[$\RuleA_{D_\land}$:]
     If $v(\cons(p \GenImp q)) \not\in D$,
     then $v(p \GenImp q) \in \{ \fvalue,\both,\neither,\tvalue \}$.
     From that and the
     assumption that $v(\cons p) \not\in D$
     and $v(\down p) \in D$,
     we have $v(p) \in \{\both,\neither \}$.
     Also, we obtain $v(q) \in \{ \fvalue,\neither,\both \}$.
     Suppose $v(p),v(q) \in D$.
     If $v(q) = \fvalue$,
     then $D = \up \fvalue$,
     and supposing
     $v(\cons r) \not\in D$,
     we have $v(r) \in D$.
     If $v(q) = \both$,
     we must have $v(p) = \neither$.
     Then again $D = \up \fvalue$,
     for the same reason as before.
     The case of $v(q) = \neither$
     is analogous.
     \item[$\RuleA_{D_\leq}$:]
     If $v(\cons(p\GenImp q)) \in D, v(p)\geq a$ and $v(\circ q) < a$,
     then $v(\cons(p\GenImp q))=\etvalue$.
     \item[$\RuleA_{D\neq \up \tvalue}$:]
     If $v(\down r) \not\in D$,
     we have $v(r) = \tvalue$.
     If $v(\up q) \in D$,
     we have $v(q) \neq \fvalue$.
     If $v(\cons(p \GenImp q)) \not\in D$,
     we have $v(p \GenImp q) \in \{ \fvalue,\both,\neither,\tvalue \}$.
     This gives $v(p) \in \{ \both,\neither,\tvalue \}$ 
     and $v(q) \in \{ \both,\neither,\tvalue \}$.
     We only consider the cases
     $v(p) \neq v(q)$.
     If $v(p) = \both$,
     then $v(q) = \neither$,
     and supposing $v(p),v(q) \not\in D$, given that $D$
     is prime we must have
     $D = \up \etvalue$,
     and $v(r) \not\in D$.
     If $v(p) = \neither$,
     the proof is similar to the previous case.
     If $v(p) = \tvalue$,
     again we have
     $D = \up \etvalue$, and we are done.
 \end{description}
 For item (iii), a valuation $v$
 such that $v(r) = \tvalue$,
 $v(p) = \neither$ and
 $v(q) = \both$
 shows that $\RuleA_{D\neq \up \tvalue}$
 is not sound when $D = \up \tvalue$.
 \qedhere
 \end{proof}

We are ready to define
our \SetSet{}
axiomatizations
and prove the desired completeness results.

\begin{definition}
Using the previous definitions,
we set the two following collection of rules: 
\begin{enumerate}
    \item $\CalcVar_{\mathsf{up}} \SymbDef\CalcVar_\Diamond\cup \CalcVar_D\cup \CalcVar_\GenImp  \cup \CalcVar_\cons\cup \CalcVar_\DMNeg \cup \CalcVar_\land \cup \CalcVar_\lor \cup \CalcVar_{\top\bot} \cup \{\RuleA_{D_\land},\RuleA_{D_\leq}\}$; 
    \item $\CalcVar_{\leq}\SymbDef\CalcVar_{\mathsf{up}} \cup \{\RuleA_{D\neq \up \tvalue}\}$.
\end{enumerate}
\end{definition}

\begin{therm}
\label{the:axiomatizations-leq-up}
$\CalcVar_{\mathsf{up}}$
axiomatizes
    $\PPResImplogicOrderMC$
    and
    $\CalcVar_{\leq}$
    axiomatizes
    $\PPResImplogicPrimeMC$.
    Moreover,
   both $\CalcVar_\mathsf{up}$ and $\CalcVar_{\leq}$ are $\Theta$-analytic calculi,
   for
$\Theta \SymbDef
\{p,\cons p,\cons(p\GenImp q),\up p,\down p\}$.
\end{therm}
\begin{proof}
    
{%
From the previous propositions it easily follows that  $\sequent_{\CalcVar_{\mathsf{up}}}\subseteq \PPResImplogicOrderMC$
and $\sequent_{\CalcVar_{\leq}}\subseteq \PPResImplogicPrimeMC$, thus guaranteeing soundness of the proposed axiomatizations.
To check completeness for $\CalcVar\in \{\CalcVar_{\mathsf{up}},\CalcVar_{\leq}\}$ we need to show that each consequence statement of the form
   $\FmSetA\not\sequent_\CalcVar \FmSetB$ is witnessed by some valuation and some principal filter over the algebra $\PPSixResImp$.
}%
   Let  $\Lambda \SymbDef \sub(\FmSetA \cup \FmSetB)$.
Recall that
 from $\FmSetA\not\sequent_\CalcVar\FmSetB$, by cut for sets, there is a partition $(\Omega,\overline{\Omega})$ of $\Theta(\Lambda)$
such that $\Omega\not\sequent_\CalcVar \,\overline{\Omega}$.
   {%
   Lemmas \ref{diamondrules} to \ref{specificrules}, below, will give us material to show how a partial valuation $f$ on
    $\Lambda$
    may be extended to
    the whole language 
    to such an effect that we will have
    $f(\Lambda\cap \Omega)\subseteq D$
    and
    $f(\Lambda\cap \overline{\Omega})\subseteq \SixSet{\setminus}D$
    for some suitable~$D$ (either principal or prime, depending on the case).
    }
    Further, analyticity will follow from the fact that
    we will be using only instances of the rules in $\CalcVar$ with formulas in $\Lambda$, thus only formulas in    
 $\Theta(\Lambda) \SymbDef \{\psi^\sigma:\psi\in \Theta\text{ and }\sigma:\{p,q\}\to \Lambda\}$ will appear along this proof.
 
The following abbreviations will be helpful in what follows.  
Let
\begin{align*}
\Lambda_{\efvalue}& \SymbDef\{\varphi\in \Lambda: \varphi\in \overline{\Omega}\text{ and }\cons \varphi\in \Omega\} \\
\Lambda_{\etvalue}&\SymbDef\{\varphi\in \Lambda: \varphi,\cons \varphi\in \Omega\} \\
\Lambda_{\Diamond}&\SymbDef\{\varphi\in \Lambda: \cons \varphi
\in \overline{\Omega} \}\\
\Lambda_\tvalue&\SymbDef\{\varphi\in \Lambda_{\Diamond}:  \down \varphi\in \overline{\Omega} \}\\
\Lambda_\fvalue&\SymbDef \{\varphi\in \Lambda_{\Diamond}: \up \varphi\in \overline{\Omega} \}\\
   \Lambda_{\mathsf{mid}}&\SymbDef\{\varphi\in \Lambda_{\Diamond}: \up\varphi,\down \varphi\in \Omega 
   \} 
\end{align*}

\noindent 
We will also
consider the relation $\equiv_\Diamond \,\subseteq\Lambda_\Diamond \times \Lambda_\Diamond$ given by
$\varphi\equiv_\Diamond \psi$ if, and only if, $\cons(\varphi\GenImp \psi),\cons(\psi\GenImp \varphi)\in \Omega$.


\begin{lemma}\label{diamondrules}
From $\CalcVar^\Diamond\subseteq \CalcVar$, 
we obtain that 
\begin{enumerate}
    \item $\Lambda_\Diamond=\Lambda_\fvalue\cup \Lambda_\tvalue\cup\Lambda_{\mathsf{mid}}$\,;
    \item the relation
 $\equiv_\Diamond$ is an equivalence relation that partitions $\Lambda_\Diamond$ into at most 
four equivalence classes.
In particular, each set $\Lambda_\fvalue$ and $\Lambda_\tvalue$ consists of formulas belonging to a single $\equiv_\Diamond$-class,
    and if $\Lambda_{\mathsf{mid}}\neq\varnothing$ then this set is partitioned into exactly two $\equiv_\Diamond$-classes, which we name
     $\Lambda_{\both}$ and $\Lambda_\neither$,
     such that,
     for 
     $\varphi,\psi \in \Lambda_{\mathsf{mid}}$
     and
     $a \in \{ \both, \neither \}$,
     we have
     $\varphi,\psi \in \Lambda_{a}$
     if, and only if,
     $\cons(\varphi \GenImp \psi) \in \Omega$; 
    \item for $\varphi\in \Lambda_a$ and $\psi\in \Lambda_b$ with $a,b\in \{\fvalue,\neither,\both,\tvalue\}$,
we have
$a \leq b$ 
if, and only if, $\cons(\varphi\GenImp \psi)\in \Omega$.
\end{enumerate}
\end{lemma}
\begin{proof}
%
%
Note that since $\Omega\cap \overline{\Omega}=\varnothing$ 
we have that $\Lambda_a\cap \Lambda_b=\varnothing$ for $a\neq b$ and 
$a,b\in \{\fvalue,\tvalue,\mathsf{mid},\etvalue,\efvalue\}$.
   By $\RuleA^\DMNeg_{\up\mathsf{or}\down}$
    we know that $\Lambda_\Diamond= \Lambda_\fvalue\cup \Lambda_\tvalue\cup \Lambda_{\mathsf{mid}}$, and this takes care of item (1).

For all that follows, recall that $\cons \varphi \in \overline\Omega$ for any $\varphi \in\Lambda_\Diamond$,
by definition of $\Lambda_\Diamond$.
That $\equiv_\Diamond$ is an equivalence relation follows by the fact that the definition is symmetric and by the presence of the rules
$\RuleA_\mathsf{id}$ and $\RuleA_\mathsf{trans}$.

We show now that formulas in
$\Lambda_\tvalue$ and $\Lambda_\fvalue$
correspond to the same
$\equiv_\Diamond$-class.
If
$\varphi,\psi\in \Lambda_\tvalue$, then
$\down \varphi,\down \psi\in \overline{\Omega}$ and  
by
$\RuleA^\Diamond_{\leq\tvalue}$
we obtain that $\varphi \equiv_\Diamond \psi$.
Similarly, if
$\varphi,\psi\in \Lambda_\fvalue$, then
$\up \varphi,\up \psi\in \overline{\Omega}$ and by $\RuleA^\Diamond_{\geq\fvalue}$
we obtain that $\varphi \equiv_\Diamond \psi$.

We will show that there are
subsets of $\Lambda_{\mathsf{mid}}$
corresponding to a partition of this set.
Our candidates are precisely the classes
$[\varphi]_{\equiv_{\Diamond}}$.
They are clearly disjoint
and their union yields $\Lambda_{\mathsf{mid}}$,
so it is enough to show that
they are all subsets of $\Lambda_{\mathsf{mid}}$.
In fact, if $\psi \in [\varphi]_{\equiv_{\mathsf{\Diamond}}}$,
we have 
$\cons\psi \in \overline{\Omega}$
and
$\cons(\varphi \GenImp \psi),\cons(\psi \GenImp \varphi) \in \Omega$.
%
%
Then we obtain $\cons\psi 
\in \overline{\Omega}$
and
$\up\psi, \down\psi \in \Omega$ 
by the rules $\RuleA^\Diamond_{\mathsf{incclass}_1}$
and $\RuleA^\Diamond_{\mathsf{incclass}_2}$.
Moreover, for $\varphi,\psi\in \Lambda_{\mathsf{mid}}$, we have that
$\cons(\varphi\GenImp \psi)\in \Omega$ if, and only if, $\varphi\equiv_\Diamond \psi$.
The harder direction is
left-to-right,
and it follows 
by $\RuleA^\Diamond_{\mathsf{incclass}_3}$.

Still, $\Lambda_{\mathsf{mid}}$ might be partitioned into more than two
$\equiv_\Diamond$-equivalence classes.
We avoid this with the
rule $\RuleA^\DMNeg_{\mathsf{just}_2}$,
which, together with the fact proved in the previous paragraph, prevents the existence of more than two $\equiv_\Diamond$-equivalence classes in $\Lambda_{\mathsf{mid}}$.
This concludes the proof of item (2).

Finally, for item (3), 
let  $\varphi\in \Lambda_a$ and $\psi\in \Lambda_b$ with $a,b\in \{\fvalue,\neither,\both,\tvalue\}$.
From left-to-right,
suppose that
$a \leq b$. We want to prove
$\cons(\varphi \GenImp \psi) \in \Omega$.
Note that we cannot have
$a = \both$ and $b = \neither$,
nor
$a = \neither$ and $b = \both$. The only cases we need to
consider, then, are:
\begin{enumerate}
    \item If $a = b$,
    we already have that
    each $\Lambda_c$, with
    $c \in \{ \fvalue,\both,\neither,\tvalue \}$
    forms an $\equiv_\Diamond$-class,
    so $\cons(\varphi \GenImp \psi) \in \Omega$
    by the definition of $\equiv_\Diamond$.
    \item If $a = \fvalue$,
    we obtain $\cons(\varphi \GenImp \psi) \in \Omega$ 
    by $\RuleA^\Diamond_{\geq\fvalue}$.
    \item If $b = \tvalue$,
    we obtain $\cons(\varphi \GenImp \psi) \in \Omega$ 
    by $\RuleA^\Diamond_{\leq\tvalue}$.
\end{enumerate}

From right-to-left, we reason contrapositively.
Suppose that $a \not\leq b$.
We want to conclude that
$\cons(\varphi \GenImp \psi) \in \overline{\Omega}$.
We only need to consider the following cases:

\begin{enumerate}
    \item If $a = \tvalue$,
    we have $b < \tvalue$.
    Thus $\down\psi \in \Omega$.
    By $\RuleA^\Diamond_{\mathsf{incclass}_2}$,
    we have $\cons(\varphi \GenImp \psi) \in \overline\Omega$.
    \item If $a = \both$ or $a = \neither$,
    and $b = \fvalue$,
    we obtain $\cons(\varphi \GenImp \psi) \in \overline\Omega$ 
    by $\RuleA^\Diamond_{\mathsf{incclass}_1}$.
    \item If $\{ a, b \} = \{ \both,\neither \}$, we have
that 
$\cons(\varphi\GenImp \psi),\cons(\psi\GenImp \varphi)\in \overline{\Omega}$ by the definition of the sets
$\Lambda_\both$ and $\Lambda_\neither$.
\qedhere
\end{enumerate}

\end{proof}

Our candidate for
partial valuation is $f:\Lambda\to \SixSet$
given by 
$f(\varphi) \SymbDef a$ if $\varphi\in \Lambda_a$. One can check without
difficulty that this function
is well defined (as a matter of fact, the sets $\Lambda_a$ defined
above are pairwise disjoint).
We will prove that $f$ is indeed
the desired valuation using a
succession of lemmas.
Let us begin with proving that $f$
is a partial homomorphism.

\begin{lemma}\label{lem:imp}

Suppose $\CalcVar_\GenImp, \CalcVar_\Diamond \subseteq \CalcVar$. If $\{\varphi,\psi,\varphi \GenImp \psi\} \subseteq \Lambda$, then $f(\varphi\GenImp \psi)=f(\varphi)\ResImp f(\psi)$.
\end{lemma}
\begin{proof}

By cases on the values of
$f(\varphi)$ and $f(\psi)$:

\begin{enumerate}
    \item If $f(\psi) = \etvalue$,
    we have $\cons\psi, \psi \in \Omega$.
    We want $f(\varphi \GenImp \psi) = \etvalue$,
    which follows from $\RuleA_1^\GenImp$
    and $\RuleA_2^\GenImp$.
    \item If $f(\varphi) = \etvalue$,
    we have $\cons\varphi, \varphi \in \Omega$
    and proceed by cases on the value of
    $f(\psi)$.
    Note that we want to show
    that $f(\varphi \GenImp \psi) = f(\psi)$:
    \begin{enumerate}
        \item The case $f(\psi) = \etvalue$
        was already covered.
        \item If $f(\psi) = \efvalue$,
        it follows 
        from $\RuleA_1^\GenImp$ and $\RuleA_3^\GenImp$.
        \item If $f(\psi) = \tvalue$,
        it follows from 
        $\RuleA_4^\GenImp$ and $\RuleA_5^\GenImp$.
        \item If $f(\psi) = \fvalue$,
        it follows from 
        $\RuleA_4^\GenImp$ and $\RuleA_6^\GenImp$.
        \item If $f(\psi) \in \{ \both, \neither \}$,
        it follows 
        by $\RuleA_4^\GenImp$, $\RuleA_7^\GenImp$
        and $\RuleA_8^\GenImp$
        that
        $f(\varphi \GenImp \psi) \in \{ \both, \neither \}$.
        Then, by 
        $\RuleA_9^\GenImp$
        and Lemma~\ref{diamondrules}
        we have
        $f(\varphi \GenImp \psi) = f(\psi)$.
    \end{enumerate}
    \item If $f(\varphi) = \efvalue$,
    we have $\cons\varphi \in \Omega$
    and $\varphi \in \overline\Omega$.
    We obtain that
    $f(\varphi \GenImp \psi) = \etvalue$
    by $\RuleA_{10}^\GenImp$ and $\RuleA_{11}^\GenImp$.
    \item If $f(\psi) = \efvalue$,
    the case $f(\varphi) \in \{\efvalue,\etvalue\}$
    was already treated, so
    we consider $f(\varphi) \in \{ \fvalue,\both,\neither,\tvalue \}$
    and want to prove
    $f(\varphi \GenImp \psi) = \efvalue$.
    From the available information,
    we have that
    $\cons\psi \in \Omega$,
    $\psi \in \overline{\Omega}$,
    $\cons \varphi \in \overline\Omega$.
    Then for any of the possibilities for $f(\varphi)$, the
    result follows
             by $\RuleA_1^\GenImp$
            and $\RuleA_{12}^\GenImp$.
    \item In case $f(\varphi),f(\psi) \in \{ \fvalue, \both, \neither, \tvalue \}$,
    we make intensive use of Lemma~\ref{diamondrules}.
    \begin{enumerate}
        \item If $f(\varphi) \leq f(\psi)$,
        we want to show $f(\varphi \GenImp \psi) = \etvalue$.
        By Lemma~\ref{diamondrules}, we already have $\cons(\varphi \GenImp \psi) \in \Omega$.
        The result then follows 
        by $\RuleA^\GenImp_{13}$.
        \item If $f(\psi) \leq f(\varphi)$
        and $f(\varphi) \not\leq f(\psi)$,
        we have 
        $\cons(\varphi \GenImp \psi) \in \overline\Omega$ and
        $\cons(\psi \GenImp \varphi) \in \Omega$
        by Lemma~\ref{diamondrules},
        and consider the following subcases:
            \begin{enumerate}
                \item If $f(\varphi) = \tvalue$,
                we want $f(\varphi \GenImp \psi) = f(\psi)$, which follows
                by $\RuleA_{9}^\GenImp$
                and $\RuleA_{14}^\GenImp$.
                \item If $f(\varphi) \in \{ \neither,\both \}$
                and $f(\psi) = \fvalue$,
                we want $f(\varphi \GenImp \psi) = \both$.
                This follows from
                by $\RuleA^\GenImp_{15}$,
                $\RuleA^\GenImp_{16}$,
                $\RuleA^\GenImp_{17}$ and
                $\RuleA^\GenImp_{18}$.
                The first two force
                $f(\varphi \GenImp \psi) \in \{ \both,\neither \}$,
                while the third
                forces $f(\varphi) \neq f(\varphi \GenImp \psi)$.
            \end{enumerate}
        \item Otherwise, $\{f(\varphi), f(\psi)\} = \{ \both, \neither\}$, thus
        we want $f(\varphi \GenImp \psi) = f(\psi)$,
        which is achieved by 
        $\RuleA^\GenImp_{9}$ and
                $\RuleA^\GenImp_{19}$
        together with Lemma~\ref{diamondrules}.\qedhere
    \end{enumerate}
\end{enumerate}
\end{proof}

\begin{lemma}
Suppose $\CalcVar_\DMNeg, \CalcVar_\Diamond\subseteq \CalcVar$. If $\{\varphi,\DMNeg\varphi\} \subseteq \Lambda$, then $f(\DMNeg\varphi)=\DMNeg^{\PPSixResImp}(f(\varphi))$.
\end{lemma}
\begin{proof}
We reason by cases on the value of
$f(\varphi)$:

\begin{enumerate}
    \item If $f(\varphi)=\etvalue$, then $\varphi,\cons \varphi\in \Omega$, so $\DMNeg \varphi\in \overline{\Omega}$ and $\cons\DMNeg \varphi\in \Omega$ by $\RuleA^\DMNeg_2$ and $\RuleA^\DMNeg_3$. 
    \item If $f(\varphi)=\efvalue$, then $\cons \varphi\in \Omega$ and $\varphi\in \overline{\Omega}$, 
so $\DMNeg \varphi\in \Omega$
and $\cons\DMNeg \varphi\in \Omega$ by $\RuleA_1^\DMNeg$
and $\RuleA_3^\DMNeg$.
    \item If $f(\varphi)=\tvalue$, then 
    $\cons\varphi \in \overline\Omega$ and
 $\down \varphi\in \overline{\Omega}$.
 Thus, by $\RuleA^\DMNeg_4$ and $\RuleA^\DMNeg_5$ 
 we have that 
 $\cons\DMNeg\varphi \in \overline\Omega$ and
 $\up \DMNeg \varphi\in \overline{\Omega}$.
    \item If $f(\varphi)=\fvalue$, then 
    $\cons\DMNeg\varphi \in \overline\Omega$
    and
 $\up \varphi\in \overline{\Omega}$.
 Thus by $\RuleA^\DMNeg_4$ and $\RuleA^\DMNeg_6$ we have 
 $\cons\DMNeg\varphi \in \overline\Omega$
 and
 $\down \DMNeg \varphi\in \overline{\Omega}$.
 \item 
  If $f(\varphi) \in \{ \both, \neither \}$, then 
  $\cons \varphi \in \overline\Omega$
  and
 $\up \varphi,\down \varphi\in \Omega$.
 By $\RuleA^\DMNeg_4$, $\RuleA^\DMNeg_7$
 and $\RuleA^\DMNeg_8$, we have
 $\cons\DMNeg\varphi \in \overline\Omega$
 and
 $\up\DMNeg\varphi, \down\DMNeg\varphi \in\Omega$.
 This gives us that
 $f(\DMNeg \varphi) \in \{ \both, \neither \}$.
 Also, since we have $\cons(\varphi \GenImp \DMNeg\varphi) = \down \varphi \in \Omega$,
 we must have $f(\DMNeg \varphi) = f(\varphi)$ by Lemma~\ref{diamondrules}.\qedhere
\end{enumerate}



\end{proof}

\begin{lemma}
Suppose $\CalcVar_\cons \subseteq \CalcVar$. If $\{\FmA,\cons\FmA\} \subseteq \Lambda$, then $f(\cons\varphi)=\cons^{\PPSixResImp}(f(\varphi))$.
\end{lemma}

\begin{proof}
Note that by $\RuleA_{1}^\cons$ we always have
$\cons\cons\varphi\in \Omega$.
Hence, if $f(\varphi)\in \{\efvalue,\etvalue\}$, then
$\cons\varphi\in \Omega$ and thus 
$f(\cons\varphi)=\etvalue=\cons^{\PPSixResImp}(f(\varphi))$.
If $f(\varphi)\notin \{\efvalue,\etvalue\}$, then
$\cons\varphi\in \overline{\Omega}$,
thus
$f(\cons\varphi)=\efvalue=\cons^{\PPSixResImp}(f(\varphi))$.
\end{proof}

\begin{lemma}
Suppose $\CalcVar_\land,\CalcVar_\Diamond \subseteq \CalcVar$. If $\{\varphi,\psi,\varphi\land\psi\} \subseteq \Lambda$, then $f(\varphi\land \psi)=f(\varphi)\land^{\PPSixResImp}\!\! f(\psi)$.
\end{lemma}
\begin{proof}
We proceed by cases on the values of
$f(\varphi)$ and $f(\psi)$:

\begin{enumerate}
    \item If $f(\varphi) = \etvalue$,
    we have $\cons\varphi,\varphi \in \Omega$.
    We want $f(\varphi \land \psi) = f(\psi)$.
        \begin{enumerate}
            \item If $f(\psi) = \etvalue$,
            we have $\cons\psi, \psi \in \Omega$.
            Then we have $\cons(\varphi \land \psi),\varphi \land \psi \in \Omega$
            by 
            $\RuleA_1^\land$ and $\RuleA_2^\land$.
            \item If $f(\psi) = \efvalue$,
            we have $\cons\psi \in \Omega$
            and $\psi \in \overline{\Omega}$.
            Then we have $\cons(\varphi \land \psi) \in \Omega$
            and $\varphi \land \psi \in \overline{\Omega}$
            by 
            $\RuleA_1^\land$ and $\RuleA_3^\land$.
            \item If $f(\psi) \in \{ \fvalue,\both,\neither,\tvalue \}$,
            we have $\cons\psi \in \overline\Omega$
            and
            that
            $f(\varphi \land \psi) \in \{ \fvalue,\both,\neither,\tvalue \}$
            by $\RuleA_4^\land$.
            Then we have
            that $\cons(\psi \GenImp (\varphi \land \psi)),\cons((\varphi \land \psi) \GenImp \psi) \in \Omega$
            by $\RuleA_5^\land$
            and $\RuleA_6^\land$.
            This, together with Lemma~\ref{diamondrules},
            gives the desired result.
        \end{enumerate}
    \item If $f(\psi) = \etvalue$,
    the argument is very similar to the previous
    one.
    The case $f(\varphi) = \etvalue$ was
    already covered.
        \begin{enumerate}
            \item If $f(\varphi) = \efvalue$,
            we have $\cons\varphi \in \Omega$
            and $\varphi \in \overline{\Omega}$.
            Then we have $\cons(\varphi \land \psi) \in \Omega$
            and $\varphi \land \psi \in \overline{\Omega}$
            by 
            $\RuleA_1^\land$ and $\RuleA_7^\land$.
            \item If $f(\varphi) \in \{ \fvalue,\both,\neither,\tvalue \}$,
            we have $\cons\varphi \in \overline\Omega$
            and also
            that
            $f(\varphi \land \psi) \in \{ \fvalue,\both,\neither,\tvalue \}$,
            in view of $\RuleA_8^\land$.
            Then we have
            that $\cons(\varphi \GenImp (\varphi \land \psi)),\cons((\varphi \land \psi) \GenImp \varphi) \in \Omega$
            by $\RuleA_9^\land$ and
            $\RuleA_{10}^\land$.
            Then, by Lemma~\ref{diamondrules},
            we obtain the desired result.
        \end{enumerate}
    \item If $f(\varphi) = \efvalue$,
    we have $\cons\varphi \in \Omega$
    and $\varphi \in \overline{\Omega}$.
    By $\RuleA_{11}^\land$
    and $\RuleA_{7}^\land$,
    we have $\cons(\varphi \land \psi) \in \Omega$
    and $\varphi \land \psi \in \overline{\Omega}$,
    as desired. 
    \item If $f(\psi) = \efvalue$,
    the argument is similar to the one above,
    and follows by $\RuleA_{12}^\land$
    and $\RuleA_{3}^\land$.
    \item If $f(\varphi),f(\psi) \in \{ \fvalue, \both, \neither, \tvalue \}$,
    we have $\cons\varphi, \cons\psi \in \overline{\Omega}$.
    Thus we have $\cons(\varphi \land \psi) \in \overline{\Omega}$
    by $\RuleA_{13}^\land$
    and thus
    $f(\varphi \land \psi) \in \{ \fvalue, \both, \neither, \tvalue \}$,
    as desired.
    Consider then the following cases:
    \begin{enumerate}
        \item If $f(\varphi) \leq f(\psi)$,
        we want to obtain 
        $f(\varphi \land \psi) = f(\varphi)$.
        By Lemma~\ref{diamondrules},
        it is enough to conclude
        $\cons(\varphi \GenImp (\varphi \land \psi)),
        \cons((\varphi \land \psi) \GenImp \varphi) \in \Omega$,
        which follows
        by $\RuleA_{10}^\land$
    and $\RuleA_{14}^\land$.
        \item If $f(\psi) \leq f(\varphi)$,
        we want to obtain 
        $f(\varphi \land \psi) = f(\psi)$.
        Analogously to the previous item,
        by Lemma~\ref{diamondrules},
        it is enough to conclude
        $\cons(\psi \GenImp (\varphi \land \psi)),
        \cons((\varphi \land \psi) \GenImp \psi) \in \Omega$,
        which follows 
        by $\RuleA_{15}^\land$
    and $\RuleA_{6}^\land$.
        \item Otherwise, we have
        $f(\varphi),f(\psi) \in \{ \both, \neither \}$
        and $f(\varphi) \neq f(\psi)$.
        Note that by Lemma~\ref{diamondrules}
        we obtain $\cons(\varphi \GenImp \psi) \in \overline\Omega$.
        We want to obtain here
        $f(\varphi \land \psi) = \fvalue$,
        that is, we want
        to have $\up(\varphi \lor \psi) \in \overline\Omega$, which is achieved 
        by $\RuleA_{16}^\land$.\qedhere
    \end{enumerate}
\end{enumerate}
\end{proof}

\begin{lemma}
Suppose $\CalcVar_\lor,\CalcVar_\Diamond \subseteq \CalcVar$. If $\{\varphi,\psi,\varphi\lor\psi\} \subseteq \Lambda$, then $f(\varphi\lor \psi)=f(\varphi)\lor^{\PPSixResImp}\!\! f(\psi)$.
\end{lemma}
\begin{proof}
We proceed by cases on the values of 
$f(\varphi)$ and $f(\psi)$:
\begin{enumerate}
    \item If $f(\varphi)=\efvalue$, we know that
    $\varphi\in \overline\Omega$ and $\cons \varphi\in {\Omega}$.
    We want to obtain that 
    $f(\varphi \lor \psi)=f(\psi)$,
    since $f(\psi) = \efvalue \lor^{\PPSixResImp}\!\! f(\psi)$.
        \begin{enumerate}
            \item If $f(\psi)=\efvalue$,
            the result follows 
            by $\RuleA_{1}^\lor$
            and $\RuleA_{2}^\lor$.
            \item If $f(\psi)=\etvalue$,
            the result follows by
            $\RuleA_{1}^\lor$
            and $\RuleA_{3}^\lor$.
            \item If $f(\psi) \in \{ \fvalue,\both,\neither,\tvalue \}$,
            we have $\cons\psi \in \overline\Omega$
            and
            we know
            that
            $f(\varphi \lor \psi) \in \{ \fvalue,\both,\neither,\tvalue \}$
            by $\RuleA_4^\lor$.
            Then we have
            that $\cons(\psi \GenImp (\varphi \lor \psi)),\cons((\varphi \lor \psi) \GenImp \psi) \in \Omega$
            by $\RuleA_5^\lor$ and $\RuleA_6^\lor$,
            which, together with Lemma~\ref{diamondrules},
            gives the desired result.
        \end{enumerate}
    \item If $f(\psi)=\efvalue$, we know that
    $\psi\in \overline\Omega$ and $\cons \psi\in {\Omega}$.
    We want to obtain that 
    $f(\varphi \lor \psi)=f(\varphi)$,
    since $f(\varphi) = f(\varphi) \lor^{\PPSixResImp} \efvalue$.
        \begin{enumerate}
            \item The case $f(\varphi)=\efvalue$
            was already covered.
            \item If $f(\varphi)=\etvalue$,
            the result follows
            by $\RuleA_1^\lor$
            and $\RuleA_7^\lor$.
            \item If $f(\varphi) \in \{ \fvalue,\both,\neither,\tvalue \}$,
            we have $\cons\varphi \in \overline\Omega$
            and
            that
            $f(\varphi \lor \psi) \in \{ \fvalue,\both,\neither,\tvalue \}$
            in view of 
            $\RuleA_8^\lor$.
            So we have
            that $\cons(\varphi \GenImp (\varphi \lor \psi)),\cons((\varphi \lor \psi) \GenImp \varphi) \in \Omega$
            by $\RuleA_9^\lor$ and $\RuleA_{10}^\lor$.
            By Lemma~\ref{diamondrules},
            then,
            we
            obtain the desired result.
        \end{enumerate}
    \item If either $f(\varphi) = \etvalue$ or $f(\psi) = \etvalue$,
    we know that either
    $\cons\varphi,\varphi\in \Omega$
    or $\cons\psi,\psi\in \Omega$.
    In any case, we obtain
    $\cons(\varphi \lor \psi), \varphi \lor \psi \in \Omega$
    by
    $\RuleA_{11}^\lor$,
    $\RuleA_{7}^\lor$,
    $\RuleA_{12}^\lor$ and
    $\RuleA_{3}^\lor$, and therefore
    $f(\varphi \lor \psi) = \etvalue$, as desired.
    \item If $f(\varphi),f(\psi) \in \{ \fvalue, \both, \neither, \tvalue \}$,
    we know that
    $\cons \varphi, \cons \psi \in \overline\Omega$, and so
    $\cons(\varphi \lor \psi) \in \overline\Omega$, by 
    $\RuleA_{13}^\lor$,
    thus $f(\varphi \lor \psi) \in \{ \fvalue, \both, \neither, \tvalue \}$.
    Consider now the following cases:
    \begin{enumerate}
        \item If $f(\varphi) \leq f(\psi)$,
        we want to obtain
        $f(\varphi \lor \psi) = f(\psi)$.
        By Lemma~\ref{diamondrules},
        we know that
        $\cons(\varphi \GenImp \psi) \in \Omega$.
        But then
        $\cons(\psi \GenImp (\varphi \lor \psi))$,
        $\cons((\varphi \lor \psi) \GenImp \psi) \in \Omega$,
        by $\RuleA_5^\lor$ and $\RuleA_{14}^\lor$,
        which gives us the desired result
        by invoking Lemma~\ref{diamondrules} again.
        \item If $f(\psi) \leq f(\varphi)$,
        we reason analogously to the previous item, but now using $\RuleA_9^\lor$ and $\RuleA_{15}^\lor$.
        \item Otherwise, we have
        $f(\varphi),f(\psi) \in \{ \both, \neither \}$
        and $f(\varphi) \neq f(\psi)$.
        Note that by Lemma~\ref{diamondrules}
        we obtain $\cons(\varphi \GenImp \psi) \in \overline\Omega$.
        We want to obtain here
        $f(\varphi \lor \psi) = \tvalue$,
        that is, we want
        to have $\down(\varphi \lor \psi) \in \overline\Omega$, and this is achieved 
        by $\RuleA_{16}^\lor$.\qedhere
    \end{enumerate}
\end{enumerate}
\end{proof}

\begin{lemma}
Suppose $\CalcVar_{\top\bot} \subseteq \CalcVar$. If $\top \in \Lambda$, then $f(\top)=\etvalue$; if $\bot \in \Lambda$, then $f(\bot) = \efvalue$.
\end{lemma}
\begin{proof}
    Obvious 
    from the rules of
    $\CalcVar_{\top\bot}$
    and the definition of $f$.
\end{proof}

\begin{lemma}\label{specificrules}
If $\{ \RuleA_{D_\land},\RuleA_{D_\leq}\},\CalcVar_\Diamond \subseteq \CalcVar$, we have $f[\Omega\cap \Lambda]=\up a\cap f[\Lambda]$
 for some $a>\efvalue$;
 and
 if we also have $\RuleA_{D\neq \up \tvalue}\in \CalcVar$, 
 then
 $f[\Omega\cap \Lambda]=\up a\cap f[\Lambda]$
 for some $\efvalue<a\neq \tvalue$.

%



\end{lemma}
\begin{proof}
First of all, we show that
(I):
for $\varphi_1, \ldots, \varphi_n \in \Omega \cap \Lambda$,
if $\bigwedge_i f(\varphi_i) \in f[\Lambda]$,
then
$\bigwedge_i f(\varphi_i) \in f[\Omega \cap \Lambda]$.
By induction on $n$,
consider the base case $n = 2$.
The cases
$f(\varphi) \leq f(\psi)$
and $f(\psi) \leq f(\varphi)$
are obvious, as $f(\varphi) \land f(\psi)$
will coincide either with $f(\varphi)$
or with $f(\psi)$
and they are in $f[\Omega \cap \Lambda]$
by assumption.
The tricky case thus is when
$\{ f(\varphi),f(\psi) \} = \{ \both, \neither \}$.
Suppose that
for some $\theta \in \Lambda$
we have
$f(\theta) = f(\varphi) \land f(\psi) = \fvalue$.
Then, 
by $\RuleA_{D_\land}$,
we must have $\theta \in \Omega$,
and so $f(\theta) = \fvalue \in f[\Omega\cap\Lambda]$.
In the inductive step,
suppose that
$b \SymbDef (f(\varphi_1) \land \ldots \land f(\varphi_n)) \land f(\varphi_{n+1})
\in f[\Lambda]$,
for $\varphi_1,\ldots,\varphi_{n+1} \in \Omega \cap \Lambda$.
Then either
(i) $b = (f(\varphi_1) \land \ldots \land f(\varphi_n))$,
or
(ii) $b = f(\varphi_{n+1})$,
or 
(iii) $b = \fvalue$ and
$\{f(\varphi_1) \land \ldots \land f(\varphi_n),f(\varphi_{n+1})\} = \{ \both,\neither \}$.
In case (i), we use the induction hypothesis.
Case (ii) is obvious.
As for case (iii), we use $\RuleA_{D_\land}$ as we did in
the base case.

Second, we show that (II): for $\varphi \in \Omega \cap \Lambda$ and $\psi \in \Lambda$, if
$f(\varphi) \leq f(\psi)$,
then $\psi \in \Omega$.
By cases on the value of 
$f(\varphi)$ (note that
$f(\varphi) \neq \fvalue$ as $\varphi \in \Omega$):
\begin{enumerate}
    \item If $f(\varphi) = \etvalue$,
    we must have $f(\psi) = \etvalue$,
    thus $\psi \in \Omega$.
    \item If $f(\varphi) \in \{ \fvalue, \both, \neither, \tvalue \}$,
    we have 
    either $f(\psi) = \etvalue$,
    and thus $\psi \in \Omega$,
    or $f(\psi) \in \{ \fvalue, \both, \neither, \tvalue \}$.
    In that case,
    by Lemma~\ref{diamondrules},
    we have $\cons(\varphi \GenImp \psi) \in \Omega$.
    Then $\psi \in \Omega$
    follows 
    by $\RuleA_{D_\leq}$.
\end{enumerate}

Let $a \SymbDef \bigwedge f[\Omega \cap \Lambda]$. 
Clearly, $a \neq \efvalue$.
Since $a \leq b$ for each
$b \in f[\Omega \cap \Lambda]$,
we must have
$f[\Omega \cap \Lambda] \subseteq \up a \cap f[\Lambda]$.
It remains to show that
$\up a \cap f[\Lambda] \subseteq f[\Omega \cap \Lambda]$.
Suppose
that there is $\theta \in \overline\Omega \cap \Lambda$ such that
(a): $a \leq f(\theta)$. By cases:
\begin{enumerate}
    \item If $f(\theta) \leq f(\varphi)$
    for all $\varphi \in \Omega \cap \Lambda$, then 
    $f(\theta) = a$,
    and thus $\theta \in \Omega$
    by (I).
    \item If $f(\theta) \not\leq f(\psi)$
    for some $\psi \in \Omega$,
    we have the following subcases:
        \begin{enumerate}
            \item If $f(\psi) < f(\theta)$,
            by (II)
            we must
            have $\theta \in \Omega \cap \Lambda$.
            \item If $f(\psi) \not\leq f(\theta)$,
            we have $f(\psi) = \both$
            and $f(\theta) = \neither$
            or vice-versa.
            By (a), we have
            $a \in \{ \fvalue, \neither \}$.
            The case $a = \neither$
            was treated in (1),
            so we only consider
            $a = \fvalue$.
            This means that either
            $\both,\neither \in f[\Omega \cap \Lambda]$
            or
            $\fvalue \in f[\Omega \cap \Lambda]$,
            and the result follows from
            (II).
        \end{enumerate}
\end{enumerate}

Now, in case the rule
$\RuleA_{D\neq \up \tvalue}$ is present
in the calculus,
we are able to show that 
if $\varphi, \psi \in \Omega \cap \Lambda$
and $f(\varphi) \lor f(\psi) \in f[\Omega \cap \Lambda]$,
then either $f(\varphi) \in f[\Omega \cap \Lambda]$ or $f(\psi) \in f[\Omega \cap \Lambda]$, which essentially excludes
the possibility of $a = \tvalue$.
Suppose that for $\theta \in \Lambda$,
$f(\theta) = f(\varphi) \lor f(\psi) \in f[\Omega \cap \Lambda]$.
Then either $f(\theta) = f(\varphi)$,
or
$f(\theta) = f(\psi)$
or $f(\theta) = \tvalue$
and $\{f(\varphi),f(\psi)\} = \{ \both,\neither \}$. The first two cases are obvious.
The third one follows
because 
the rule $\RuleA_{D\neq \up \tvalue}$
forces $\varphi \in \Omega$ in this
situation.\qedhere

\end{proof}

{%
We are now finally ready to go back to 
%
the completeness proof we were working on before going through the above series of auxiliary results.
}
Note that the above lemmas show that $f$
is a partial homomorphism over $\PPSixResImp$
 (which can of course be extended to a full homomorphism)
and, in view of the above lemmas, 
by considering a set of designated values of the form $\up a$ for appropriate $a$ we obtain a countermodel for $\FmSetA \sequent_{\PPResImplogicOrderMC} \FmSetB$
or for
$\FmSetA\sequent_{\PPResImplogicPrimeMC}\FmSetB$
as desired.
This finishes
the proof of
Theorem~\ref{the:axiomatizations-leq-up}.
\qedhere
\end{proof}

\subsection{\SetFmla{} axiomatization for $\PPResImplogicSingle$}
\label{sec:set-fmla-axiomat}

The \SetSet{} calculus developed in the preceding
subsection for the logic $\PPResImplogicPrimeMC$
induces a \SetFmla{} logic for its \SetFmla{}
companion $\PPResImplogicSingle$.
We begin by defining this calculus, 
then  indicate why it is complete for
$\PPResImplogicSingle$.
In what follows, given a set of formulas $\FmSetA$, let $\FmSetA \ArbDisj \FmB \SymbDef \{ \FmA \ArbDisj \FmB : \FmA \in \FmSetA \}$
and $\bigvee \{ \FmA_1,\ldots,\FmA_m \} \SymbDef \FmA_1 \ArbDisj (\FmA_2 \ArbDisj \ldots (\ldots \ArbDisj \FmA_n)\ldots)$.

\begin{definition}
Let $\mathsf{R}$ be
a \SetSet{} calculus.
We define $\mathsf{R}^{\lor}$ as the 
\SetFmla{} calculus
\[
\left\{\frac{p}{q \lor q},\frac{p}{p \lor q}, \frac{p \lor q}{q \lor p}, \frac{p \lor (q \lor r)}{(p \lor q) \lor r}\right\} \cup
\left\{\mathsf{r}^\lor : \mathsf{r} \in \mathsf{R}\right\},\text{ where, for each given 
rule $\mathsf{r}=\frac{\FmSetA}{\FmSetB}$ in~$\mathsf{R}$,}\] the rule $\mathsf{r}^\lor$ is set as:
\begin{tabular}{ll}
  $\mathsf{r}$, & in case $\FmSetA$ is empty and $\FmSetB$ is a singleton \\
  $\frac{\FmSetA \lor s}{s}$, & in case $\FmSetB$ is empty \\
  $\frac{\FmSetA \lor s}{(\bigvee \FmSetB) \lor s}$, & otherwise
\end{tabular}. \\
In all cases, $s$ is chosen to be a propositional variable not occurring in the rules that belong to $\mathsf{R}$.
\end{definition}

Using the above recipe is straightforward,
and this gives the reason why we decided to not spell out
the whole axiomatization here.
Before introducing the completeness result,
we define what it means for a \SetFmla{} logic
to have a disjunction.
A \SetFmla{} logic~$\vdash$~over $\Sigma$ 
\emph{has a disjunction} provided that 
$\FmSetA, \FmA\lor\FmB \;\vdash\; \FmC$ {if, and only if,} $\FmSetA,\FmA \;\vdash\; \FmC$ and $\FmSetA,\FmB\;\vdash\;\FmC$ (for $\lor$ a binary connective in $\Sigma$).
The completeness result is immediate
from the fact that  $\PPResImplogicSingle$
has a disjunction, 
in view of the following result:

\begin{lemma}[{\citet*[Thm. 5.37]{shoesmith_smiley:1978}}]
    Let $\CalcVar$ be a \SetSet{} calculus
    over a signature containing a binary
    connective $\lor$.
    If $\vdash_{\CalcVar}$
    has a disjunction,
    then $\vdash_{\CalcVar^\lor} \;=\; \vdash_{\CalcVar}$.
\end{lemma}

\noindent 
From this it immediately follows that: 

\begin{therm}
    $\vdash_{\CalcVar_{\leq}^\lor}
    =
    \PPResImplogicSingle$.
\end{therm}

\section{Algebraic study of $\PPResImplogic$ and $\PPResAsslogic$}
\label{sec:moi}

In this section we look at the class of algebras
that corresponds to $\PPResImplogic$
and  $\PPResAsslogic$
according to  
the general theory of algebraization of logics (\cite{font:2016}).  
In order to facilitate the exposition, here we will write a
$\Sigma$-algebra 
$\AlgA \SymbDef \langle A, \cdot^{\AlgA} \rangle$
as $\langle A; \DefCon_1^\AlgA,\ldots,\DefCon_n^{\AlgA} \rangle$, with $\DefCon_i \in \Sigma$
for each $1 \leq i \leq n$.
We will further omit the superscript 
$\AlgA$ 
from the interpretations of the connectives whenever there is no risk of confusion.

We begin by recalling that, as observed earlier,
the algebra $\PPSixResImp$ is
a {symmetric Heyting algebra} 
in Monteiro's sense.


\begin{definition}[{\citet*[Def.~1.2,~p.~61]{Monteiro1980}}]
\label{d:symhey}
A \emph{symmetric Heyting algebra (SHA)}
is a $\DMSigImp$-algebra $\la A; \land, \lor, \GenImp, \DMNeg, \bot, \top \ra$ such that: 
\begin{enumerate}[(i)]
\item $\la A; \land, \lor, \GenImp, \bot, \top \ra$ is a Heyting algebra.
\item $\la A; \land, \lor, \DMNeg, \bot, \top \ra$ is a De Morgan algebra. 
\end{enumerate}
\end{definition}

As mentioned earlier, SHAs are alternatively known as \emph{De Morgan-Heyting algebras} in the terminology introduced by~\cite{Sankappanavar1987}. The logical counterpart of SHAs is Moisil's
`symmetric modal logic', 
which is the expansion of the Hilbert-Bernays
positive logic (the conjunction-disjunction-implication fragment of intuitionistic logic) by the addition of a De Morgan negation. 
One might expect Moisil's logic to be closely related to 
$\PPResImplogic$. In fact, as we shall see, 
the logic $\PPResAsslogic$ considered earlier
(the $\top$-assertional 
companion of $\PPResImplogic$) 
may be viewed as an axiomatic extension of Moisil's logic; 
whereas we may obtain $\PPResImplogic$ from Moisil's logic 
provided we extend it by appropriate axioms but also drop the \emph{contraposition}
 rule schema ($\MoisilRule{12}$ below). 
 The following is
a Hilbert-style calculus
for Moisil's logic~(see \citet[p.~60]{Monteiro1980}):
$$
\inferx[\MoisilRule{1}]
{}{p \Imp (q \Imp p)}
\quad
\inferx[\MoisilRule{2}]
{}{(p \Imp (q \Imp r)) \Imp ((p \Imp q) \Imp (p \Imp r))}
$$
$$
\inferx[\MoisilRule{3}]
{}{(p \land q) \Imp p}
\quad 
\inferx[\MoisilRule{4}]
{}{(p \land q) \Imp q}
\quad
\inferx[\MoisilRule{5}]
{}{
(p \Imp q) \Imp ((p \Imp r) \Imp (p \Imp (q \land r)))
}
$$
$$
\inferx[\MoisilRule{6}]
{}{p \Imp (p \lor q)}
\quad 
\inferx[\MoisilRule{7}]
{}{q \Imp (p \lor q)}
\quad
\inferx[\MoisilRule{8}]
{}
{(p \Imp r) \Imp ((q \Imp r) \Imp ((p \lor q) \Imp r))}
$$
$$
\inferx[\MoisilRule{9}]
{}{p \Imp \DMNeg\DMNeg p}
\quad 
\inferx[\MoisilRule{10}]
{}{\DMNeg\DMNeg p \Imp p}
$$
$$
\inferx[\MoisilRule{11}]{p, p\Imp q}{q}
\quad 
\inferx[\MoisilRule{12}]{p \Imp q}{\DMNeg q \Imp \DMNeg p}
$$



The Lindenbaum-Tarski algebras of Moisil's logic are precisely the symmetric Heyting algebras~(see \citet[Thm.~2.3,~p.~62]{Monteiro1980}). Using this result,
it is easy to obtain the following:

\begin{proposition}
\label{prop:moisalg}
Moisil's logic is algebraizable  (in the sense of~\cite{BP89})
with the same translations as positive logic (namely, equivalence formulas $\{ x \Imp y, y \Imp x \}$ and defining equation $x \approx \top$).  Its equivalent algebraic semantics is the variety of symmetric Heyting algebras.
\end{proposition}

Having verified that $\PPResAsslogic$ is an axiomatic extension of 
Moisil's logic (see our axiomatization below), 
we will immediately obtain that $\PPResAsslogic$ is algebraizable 
with the  translations mentioned in the preceding proposition; 
the equivalent algebraic semantics
of $\PPResAsslogic$ is then bound to be a sub(quasi)variety of SHAs.
As we will see, 
one may obtain $\PPResAsslogic$ from the above axiomatizations of Moisil's logic by adding
the following axioms
(we 
use $\neg p$ as an abbreviation of $p \Imp \DMNeg (p \Imp p) $
and 
$\cons p$ as an abbreviation of $\neg p \lor \neg  \DMNeg p)$:
$$
\inferx[\PPTopRule{1}]{}{\neg p \Imp \DMNeg \neg \neg p}
\qquad
\inferx[\PPTopRule{2}]{}{\DMNeg \neg \neg p \Imp \neg p}
$$
$$
\inferx[\PPTopRule{3}]{}{({\cons(p_1 \Imp p_2) \land \cons(p_2\Imp p_3) }) 
\Imp  
({\cons p_1 \lor \cons p_4 \lor \cons(p_4\Imp p_3)}  \lor  \cons(p_3\Imp p_2) \lor \cons(p_2\Imp p_1))}
$$

This axiomatization should be compared with the equational presentation given in Definition~\ref{d:PPImpResVar}, and the claimed completeness of the axiomatization will follow
from Theorems~\ref{thm:vargen}  and~\ref{thm:algppass}. Definition~\ref{d:PPImpResVar}(2)
matches $\PPTopRule{3}$. Definition~\ref{d:PPImpResVar}(1) says that the algebra has a PP-algebra reduct:
as observed in~\citet[p.~3150]{marcelino2021}, for an algebra that has a pseudo-complement negation (as all symmetric Heyting algebras do), it is sufficient to impose the equation $\DMNeg \neg \neg x \approx \neg x$ to obtain an involutive Stone
algebra (i.e., modulo the language, a PP-algebra); clearly the equation $\DMNeg \neg \neg x \approx \neg x$
corresponds, via algebraizability, to
$\PPTopRule{1}$ and $\PPTopRule{2}$.


One can also show that $\PPResImplogic$ may be  obtained from the preceding axiomatization by 
taking as axioms all the valid formulas while
dropping
the contraposition rule of Moisil's logic (see, e.g., \citet[p.11]{bou11}).

We proceed to obtain further information on the subclass of symmetric Heyting
algebras that are models of $\PPResAsslogic$ and $\PPResImplogic$. 
%
\cite{Monteiro1980} carried out an extensive study of symmetric Heyting (and related) algebras;
independently,
some of Monteiro's results 
were rediscovered
and  a number of new ones obtained in~\cite{Sankappanavar1987}. From these works we shall recall only a few results needed for our purposes. 

The following example is of special relevance to us because, as we shall
see, the symmetric Heyting algebras we are mostly interested in have the shape
described therein:

\begin{example}[\cite{Sankappanavar1987},~p.~568]
\label{ex:chsum}
Let $\A[D]$ be a finite De Morgan algebra and let $\A[C]^+_n$ and  $\A[C]^-_n$ be two  $n$-element chains, which we view as lattices.
Denoting  $C^+_n \SymbDef \{ c_1, \ldots, c_n \}$, with $c_1 < \ldots < c_n$, let
 $C^-_n \SymbDef \{ \DMNeg c_1, \ldots, \DMNeg c_n \}$, with $\DMNeg c_n < \ldots < \DMNeg c_1$.
Consider the {ordinal sum} of these lattices,  $\A[D]^{\A[C]} : = \A[C]^-_n \oplus \A[D] \oplus \A[C]^+_n$, in which the order and the De Morgan negation are
defined as follows:  for all $d \in D$ and $c_i \in C^+_n$, we let 
$\DMNeg c_i<d<c_i$ and
$\DMNeg^{\A[D]^{\A[C]}} d = \DMNeg^{\A[D]} d$. Then, $\A[D]^{\A[C]} $ is a finite De Morgan algebra, and can therefore be endowed
with the Heyting implication determined by the order,  turning it into a symmetric Heyting algebra. 
\end{example}

{
We now proceed to the axiomatization of $\PPImpResVar$. First, we define an equational class we shall call $\PPImpResTempVar$,
then we show that it coincides with
$\PPImpResVar$.

\begin{definition}
\label{d:PPImpResVar}
Let $\A \SymbDef \la A; \land, \lor, \Imp, \DMNeg, \cons, \bot, \top \ra$
be a symmetric Heyting algebra expanded
with an operation $\cons$.
We say that $\A$
is in the class $\PPImpResTempVar$ if the following conditions are satisfied:
\begin{enumerate}[(i)]
\item 
The reduct $\la A; \land, \lor, \cons, \DMNeg, \bot, \top \ra$ is a PP-algebra.
\item $\A$ satisfies the following equation:
$$
{\cons(x_1 \Imp x_2) \land \cons(x_2\Imp x_3) } 
\leq 
{\cons x_1 \lor \cons x_4 \lor \cons(x_4\Imp x_3)}  \lor  \cons(x_3\Imp x_2) \lor \cons(x_2\Imp x_1).
$$
\end{enumerate}
\end{definition}
}

Our next aim is to check the following results:
$\PPImpResTempVar$ is the variety (and the quasi-variety) generated by $\PPSixResImp$
(Theorem~\ref{thm:vargen}),
and
$\PPImpResTempVar$ is the equivalent algebraic semantics of the algebraizable logic $\PPResAsslogic$
(Theorem~\ref{thm:algppass}).
%
%
%
We shall prove the above by relying on a few lemmas that  also have an independent interest, 
in that they shed some light on the structures of
symmetric Heyting algebras in general and of algebras in $\PPImpResVar$ in particular.

Adopting Monteiro's notation, we shall use the abbreviations $\neg x := x \Imp \DMNeg (x \Imp x)$ and $\De x : = \neg \DMNeg x$
(recall that, on every algebra having a 
PP-algebra reduct (\cite{Gomes2022}), one may take $\cons x \SymbDef \neg x \lor \De x$
and $\De x \SymbDef x \land \cons x$).
In $\PPSixResImp$, we have:
\[
\neg^{\PPSixResImp}(a)
=
\begin{cases}
    \etvalue & \text{if } a = \efvalue,\\
    \efvalue & \text{otherwise.}
\end{cases}
\quad
\text{and}
\quad
\Delta^{\PPSixResImp}(a)
=
\begin{cases}
    \etvalue & \text{ if } a = \etvalue,\\
    \efvalue & \text{ otherwise.}
\end{cases}
\]

Given a SHA $\A$ and a subset $F \subseteq A$, we shall say that $F$ is a
\emph{filter} if it is a non-empty lattice filter of the bounded lattice reduct of
$\A$. A filter $F$ will be called \emph{regular} if $\De a \in F$ whenever $a \in F$, for all
$a \in A$. For example, the algebra $\PPSixResImp$ has only one proper regular filter,
namely the singleton $\{ \etvalue \}$.
On every SHA, the regular filters  form a closure system (hence, a complete lattice) 
and 
may be characterized in a number of alternative ways (see e.g.~\citet[Thm.~4.3,~p.~75, and Thm.~4.11,~p.~80]{Monteiro1980}). In particular, it can be shown that regular filters coincide 
with the~lattice filters~$F$ that further satisfy the contraposition rule, that is: $\DMNeg b \Imp \DMNeg a \in F$ whenever $a \Imp b \in F$. By the algebraizability of
Moisil's logic, this observation yields the following result,
which Sankappanavar proved in a more general (and purely algebraic) context:

\begin{lemma}[\cite{Sankappanavar1987},~Thm.~3.3]
\label{lem:filcon}
    The lattice of congruences of each SHA $\A$ is isomorphic to the lattice of 
    regular filters on $\A$.
\end{lemma}

We now focus on a subvariety of SHAs (to which
$\PPSixResImp$ obviously belongs) where regular filter generation admits a particularly simple description. 
As we will see, the following equation will
play a key role:
\begin{equation}
\label{eq:idempeqname}
\De \De x \approx \De x
\tag{\IdempEqName}
\end{equation}

\begin{lemma}[{\citet*[Thm.~4.17,~p.~82]{Monteiro1980}}]
\label{lem:filgen}
Let $\A$ be a SHA that satisfies 
{\normalfont(\ref{eq:idempeqname})}
and let $B \subseteq A$.
The regular filter $F(B)$ generated by $B$ 
is given by:
$$
F(B) \SymbDef \{ a \in A : \De (a_1 \land \ldots \land a_n) \leq a \text{ for some } a_1, \ldots, a_n \in B \}.
$$    
\end{lemma}

\begin{lemma}[{\citet*[Cor.~4.8]{Sankappanavar1987}}]
\label{lem:simpsi}
Let $\A$ be a SHA that satisfies
{\normalfont(\ref{eq:idempeqname})}.
The following are equivalent:
\begin{enumerate}[(i)]
\item $\A$ is directly indecomposable.
\item $\A$ is subdirectly irreducible.
\item $\A$ is simple.
\end{enumerate}
\end{lemma}




The subvariety of SHAs defined by the equation
{\normalfont(\ref{eq:idempeqname})}
is dubbed $\SDH$ by~\cite{Sankappanavar1987},
who observes that it is a discriminator\footnote{{Note that this notion of `discriminator variety' is 
unrelated to the above-mentioned `discriminator for a
monadic matrix' (see Theorem~\ref{fact:classical-imp-axiomat}).}}
variety~(see \citet[Def.~IV.9.3]{burris2011}).
This entails that $\PPImpResTempVar$ is also a discriminator variety. The discriminator term 
for $\PPImpResTempVar$
is the following:
$$
t(x,y,z):=( \De (x \Leftrightarrow  y) 
\land z )\lor (\DMNeg 
\De (x \Leftrightarrow y) 
\land x 
)
$$
where  $x \Leftrightarrow y \SymbDef (x \Imp y) \land (y \Imp x)$.
To see this, consider Theorem~\ref{thm:vargen} and note that on $\PPSixResImp$
we have,
for all $c,d  \in \SixSet$,
\[
\De^{\PPSixResImp} (c \Leftrightarrow d) 
=
\begin{cases}
    \etvalue & \text{ if } c = d,\\
    \efvalue & \text{ otherwise.}
\end{cases}
\]

%
Recall that every variety of algebras is generated by its subdirectly irreducible members, which in our case (since $\PPImpResTempVar \subseteq \SDH$)
are simple, by Lemma~\ref{lem:simpsi}. By Lemma~\ref{lem:filcon}, this means that every such algebra $\A$
with a top element $\top$ has (as we have seen of $\PPSixResImp$)
a single non-trivial regular filter, namely the singleton $\{ \top \}$.
The following is now the main lemma we need:

\begin{lemma}
\label{lem:varax}
 Every simple algebra $\A \in \PPImpResTempVar$  is (isomorphic to)
 a subalgebra of $\PPSixResImp$.
\end{lemma}

\begin{proof}
In the proof we shall use the equations
$\De x \leq x$ and $\De (x \lor y) = \De x \lor \De y$, which may be easily verified in $\PPSixResImp$.
If $\A$ is simple, then $\A$ has only one non-trivial congruence, corresponding to 
the regular filter $\{ \top \}$ --- which in this case must be prime, as we now argue.
In fact, since $\A$ satisfies $\De \De x \approx \De x$, we have $\De a = \bot$ whenever $a \neq \top$; 
otherwise, 
indeed,
the regular filter 
$F(a) = \{ b \in A : \De a \leq b \}
$ would be proper (Lemma~\ref{lem:filgen}). Now, assume $a_1 \lor a_2 = \top$ for some $a_1, a_2 \in A$. 
We have $\De (a_1 \lor a_2) = \De a_1 \lor \De a_2 = \top = \De \top$. Thus, if $a_1 \neq \top$, then $\De a_1 = \bot$ and $a_2 \geq \De a_2 =   \bot \lor  \De a_2 = \De a_1 \lor \De a_2 = \top$,
so $a_2 = \top$. Thus $\{\top \}$ is a prime filter (entailing, since $\A$ has a De Morgan algebra reduct, that $\{ \bot\}$ is a prime ideal); hence the last elements of $\A$ form a chain $C^+$, and the first elements of
$\A$ form a chain $C^-$. This means that $\A$ has the shape described in Example~\ref{ex:chsum}, except that it need not be finite. 


Let $D : = A - \{ \bot, 
\top \}$. We claim that any chain of elements in $D$ must have length at most $3$.
To see this, notice that, for all $a \in A$, we have
$a \in D$ if, and only if, 
 $\cons a = \bot$.  In fact, for $a \in D$, we have $\neg a = \bot = \De a$, so
 $\cons a = \neg a \lor \De a = \bot$. Notice also that, for $a_1, a_2 \in D$,
 we have $\cons (a_1 \Imp a_2) = \bot$ if (and only if) $a_1 \not \leq a_2$.

Now assume, by way of contradiction, that there are elements $a_1, a_2, a_3, a_4 \in D$ such that $a_1 < a_2 < a_3 < a_4$, forming a four-element chain.
Then the inequality in the second item of Definition~\ref{d:PPImpResVar} would fail, for we would have:
$$
{\cons(a_1 \Imp a_2) \land \cons(a_2\Imp a_3) } = \cons \top \land \cons \top= \top \not \leq 
\bot = \cons a_1 \lor \cons a_4 \lor \cons(a_4 \Imp a_3)  \lor  \cons( a_3\Imp a_2) \lor \cons( a_2\Imp a_1).
$$
Thus, all chains in $D$ have at most three elements. Hence, as $\A$ is a distributive lattice, it
is easy to verify that
$\A$ must be finite, with a coatom (call it $\tvalue$) and an atom $\fvalue$.
This easily entails that $\A$ must be isomorphic to one of the subalgebras of $\PPSixResImp$.
\end{proof}

Theorem~\ref{thm:vargen} below, the first of the announced goals, is then essentially an immediate corollary
of Lemma~\ref{lem:varax}:

\begin{therm}
    \label{thm:vargen}
$\PPImpResTempVar$ is the variety (and the quasi-variety) generated by $\PPSixResImp$.
\end{therm}
\begin{proof}
    From Lemma~\ref{lem:varax}
    and, for the quasi-variety result,
    see e.g.~\citet*[Thm.~1.3.6.ii]{ClDa98}.
\end{proof}

We are also ready to prove the second of the goals:

\begin{therm}
\label{thm:algppass}
$\PPImpResTempVar$ is the equivalent algebraic semantics of the algebraizable logic $\PPResAsslogic$.
\end{therm}
\begin{proof}
    We know from the algebraizability of Moisil's logic (with respect to all SHAs) that 
    $\PPResAsslogic$ is algebraizable, with the same translations, with respect to a sub(quasi)variety of SHAs which is axiomatized by  the equations that translate the new
    axioms. Recalling that the equation
    $\neg x \approx \DMNeg \neg \neg x$ guarantees that a Heyting algebra has an involutive Stone algebra reduct~(see \citet[Remark~2.2]{cignoli1983}),
    is easy to verify that the translations of the new axioms  are indeed equivalent to the equations introduced in Definition~\ref{d:PPImpResVar}. 
\end{proof}

Another consequence of Lemma~\ref{lem:varax}
that has a logical impact is the following.
Recall that the universes  of subalgebras of $\PPSixResImp$
are the same as those of $\PPSix$, minus the five-element chain. Let us denote the corresponding 
algebras by $ \PPFourResImp, \PPThreeResImp$ and $\PPTwoResImp$.

\begin{corollary}
\label{cor:subvar}
 The (proper, non-trivial) subvarieties of $\PPImpResTempVar$  are precisely the following (the axiomatizations are 
 obtained by adding the mentioned equations to
 the axiomatization of $\PPImpResTempVar$):
 \begin{enumerate}
 \item $\Variety (\PPFourResImp)$, axiomatized 
 by 
 the equation 
 $ x \land \DMNeg x \leq y \lor \DMNeg y$.

 \item $\Variety (\PPThreeResImp)$, axiomatized 
 by 
 $ x \lor (x  \Imp   (y \lor \neg y  ) )  \approx \top $.

 \item $\Variety (\PPTwoResImp)$, axiomatized 
 by 
 $x \lor \DMNeg x \approx \top$.
 
 \end{enumerate}
 
\noindent Note that $\Variety (\PPFourResImp)$ 
and  $\Variety (\PPThreeResImp)$ are incomparable, while
$\Variety (\PPTwoResImp)$, which is (up a choice of language) just the variety of Boolean algebras,
is included in both of them. 
\end{corollary}
\begin{proof}
    All claims are established by easy computations. The main observation we need is that, by J\'onsson's Lemma (see \citet[Cor.~IV.6.10]{burris2011}),
    for $\A \in \{  \PPFourResImp, \PPThreeResImp, \PPTwoResImp\}$,
    the subdirectly irreducible (here meaning simple) algebras in each 
    $\Variety (\A)$  are in $ \HOp \SOp (\A)$, and $\HOp\SOp (\A) = \SOp (\A)$.
\end{proof}

By Theorem~\ref{thm:algppass}, the logic $\PPResAsslogic$ is complete with respect to 
the class of all matrices $\langle \A, \{ \top^{\AlgA} \} \rangle$
such that $\A \in \PPImpResVar$.
But, by Theorem~\ref{thm:vargen}, 
we know that the single matrix
$\langle \PPSixResImp, \{ \etvalue \} \rangle$ suffices. Thus:

\begin{proposition}
    \label{prop:1assert-single-mat}
    $\PPResAsslogic$ is determined
    by $\langle \PPSixResImp, \{ \etvalue \} \rangle$.
\end{proposition}

\noindent 
This observation may be used to verify the following equivalence,
which holds for all $\FmSetA\cup\{\FmA \} \subseteq \LangSet{\Sigma}$:
$$
\FmSetA  \vdash_\PPResAsslogicSingle \FmA 
\quad \text{ if, and only if, } \quad 
\De \FmSetA \vdash_\PPResImplogicSingle \FmA 
$$
where
$ \De \FmSetA : = \{ \De \FmB : \FmB \in \FmSetA \}$.
From the latter, relying on the DDT (recall Section~\ref{ss:impcon}) for $\PPResImplogic$, we can obtain the DDT
for $\PPResAsslogic$:
for all $\FmSetA, \{\FmA, \FmB \} \subseteq \LangSet{\Sigma}$,
$$
\FmSetA, \FmA \vdash_\PPResAsslogicSingle \FmB 
\quad \text{ iff } \quad 
\FmSetA \vdash_\PPResAsslogicSingle \De \FmA  \GenImp \FmB.
$$

The class of algebraic reducts of reduced matrices for
$\PPResImplogic$ is also $\PPImpResVar$.
In fact, 
it can be shown that 
$\PPResImplogic$
is complete with respect to the class of
all matrices 
$\langle \A, D \rangle$
with $\A \in \PPImpResVar$
and $D$ a non-empty lattice filter of $\A$.
The reduced models of $\PPResImplogic$ may be characterized as follows:

\begin{proposition}
\label{prop:redmodPPResImplogic}
Given $\A \in \PPImpResVar$, we have that a matrix $\la \A, D \ra $ is a reduced model of $\PPResImplogic$ if, and only if,
$D$ is a lattice filter that contains exactly one 
regular filter (namely~$\{ \top^\AlgA \}$). 
\end{proposition}

\begin{proof}
We shall prove both implications by contraposition. 
Assume first that $\la \A, D \ra $ is not reduced. 
Then, by the characterization of the Leibniz congruence given in Proposition~\ref{prop:equi}, there are elements $a,b \in A$ such that
 $a \neq b$ and
$$
\De (a \Imp b), \De (b \Imp a) \in D.
$$ Assuming $a \not \leq b$, we have $a \Imp b \neq \top$.
Recall that the regular filter  generated by the element $a \Imp b$ is
$$
F(\{ a \Imp b \}) = \{ c \in A : \De (a \Imp b) \leq c \}.
$$
From the assumption $\De (a \Imp b) \in D$, we have  
$F(\{ a \Imp b \}) \subseteq D$. So $D$ contains a regular filter
distinct from $\{\top \}$.

Conversely, assume $D$ contains a regular filter 
$F \neq \{ \top \}$. Then there is $a \in F$ such that
$a < \top $. This means that 
$\De (a \Imp \top ) = \De \top = \top  \in F$ and (since $F$ is regular)
$\De (\top \Imp a) = \De a \in F$. Then, by the characterization of
Proposition~\ref{prop:equi}, the pair $(a,\top)$ is identified by the Leibniz congruence of $D$, which would make $\la \A, D \ra $ not reduced.
\end{proof}

\section{On interpolation for $\PPResImplogicSingle$ and $\PPResAsslogic$,
and amalgamation for
$\PPImpResTempVar$}
\label{sec:interp-amalg}
We begin by defining three
basic notions of
interpolation according to the terminology of~\citet[Def. 3.1]{pigo98}.
In this section we will focus only in
\SetFmla{} logics, since this
is the framework in which 
the properties of interpolation are commonly formulated
and investigated
in the literature.


\begin{definition}
A \SetFmla{} $\Sigma$-logic $\vdash$ has the
\begin{enumerate}
    \item \emph{extension interpolation property} (EIP) if
    having
$\FmSetA,\FmSetB \vdash \FmA$
implies
that there is 
$\FmSetC \subseteq \LangSetProp{\Sigma}{\Props({\FmSetB \cup \{ \FmA\})}}$
such that
$\FmSetA \vdash \FmB$ for
all $\FmB \in \FmSetC$
and
$\FmSetC,\FmSetB \vdash \FmA$.
    \item \emph{Craig interpolation property} (CIP) if,
    whenever $\Props(\FmSetA) \cap \Props(\FmA) \neq \varnothing$,
    having
$\FmSetA \vdash \FmA$
implies
that there is 
$\FmSetC \subseteq \LangSetProp{\Sigma}{\Props({\FmSetA})
\cap \Props(\FmA)}$
such that
$\FmSetA \vdash \FmB$ for
all $\FmB \in \FmSetC$
and
$\FmSetC \vdash \FmA$.
    \item \emph{Maehara interpolation property} (MIP) if,
    whenever
    $\Props(\FmSetA)\cap \Props({\FmSetB \cup \{ \FmA\})} \neq \varnothing$,
    having
$\FmSetA,\FmSetB \vdash \FmA$
implies
that there is 
$\FmSetC \subseteq \LangSetProp{\Sigma}{\Props(\FmSetA)\cap \Props({\FmSetB \cup \{ \FmA\})}}$
such that
$\FmSetA \vdash \FmB$ for
all $\FmB \in \FmSetC$
and
$\FmSetC,\FmSetB \vdash \FmA$.
\end{enumerate}
\end{definition}

\noindent Note that the (MIP)
implies
the (CIP) --- just take $\Psi = \varnothing$.
It also implies the
(EIP) when the logic
has theses (that is, formulas
$\FmA$ such that $\varnothing \vdash \FmA$) on a single
variable and 
every formula without variables is logically equivalent to some constant in the signature.
Note that for the logics
$\PPResImplogicSingle$ and $\PPResAsslogic$
both conditions hold good.

Let us see now 
which of these interpolation properties
the logic $\PPResImplogicSingle$ satisfies.

\begin{therm}
The logic
    $\PPResImplogicSingle$
    has the (EIP) but does not have
    the (CIP)
    nor the (MIP).
\end{therm}
\begin{proof}
Since  $\PPResImplogicSingle$ is finitary, we may consider only finite candidates for $\FmSetB$ and $\FmSetC$.
To see that $\PPResImplogicSingle$
    has (EIP), assume that 
  $\FmSetA,\FmSetB\vdash_\PPResImplogicSingle \FmA$.
The property is satisfied if we choose $\FmSetC \SymbDef \{\bigwedge \FmSetB \GenImp \FmA\}$, since $\PPResImplogicSingle$ has the DDT
(see Section~\ref{ss:impcon}).

    Given that (MIP)
    implies
   (CIP),
    it is enough to show
    that the latter fails.
First of all, note that
    $(p \land \DMNeg p \land q \land \DMNeg q \land \DMNeg\cons(p \GenImp q)) 
    \lor s
    \vdash_\PPResImplogicSingle
    (r \lor \DMNeg r) \lor s
    $ (recall the proof of Theorem~\ref{fact:non-single-matrix-pp6imp}).
    Based on the definition of
    (CIP), let
    $\FmSetA \SymbDef
    \{ 
    (p \land \DMNeg p \land q \land \DMNeg q \land \DMNeg\cons(p \GenImp q)) 
    \lor s
    \}$
    and $\FmA \SymbDef (r \lor \DMNeg r) \lor s$.
    Note that 
    $\Props(\FmSetA)\cap \Props(\FmA) = \{ s \}$.
Assume there is such $\FmSetC$,
then $\FmB(s) \SymbDef \bigwedge \FmSetC$
is a formula on a single variable~$s$.
We will see now that we cannot have both 
(i) $\FmSetA \vdash_\PPResImplogicSingle \FmB(s)$ 
and
(ii) $\FmB(s) \vdash_\PPResImplogicSingle 
 \FmA
$.
Let us
fix  
$v(p) \SymbDef \both$, $v(q)\SymbDef\neither$,
$v(r)\SymbDef\both$ and $v(s)\SymbDef\efvalue$ thus making $v(\FmSetA)=\{ \fvalue \}$ and
$v(\FmA)=\both$.
Note that  $v(\FmB)\in \{\efvalue,\etvalue\}$.
However, (i) fails when $v(\FmB)=\efvalue$, and (ii) fails when $v(\FmB)=\etvalue$.\qedhere
\end{proof}

Let us now take a look
at the situation for $\PPResAsslogicSingle$.
Recall that, in the previous section, we presented the DDT
for this logic, which demands a (derived) connective different from~$\Imp$ to play the role of implication. 
In the same way as in the 
proof of the above theorem, this DDT guarantees that the (EIP) holds for this logic.
We show however that (CIP) and (MIP) also hold.
In what follows,
recall from Proposition~\ref{prop:1assert-single-mat} that $\PPResAsslogicSingle$ is determined
by the single matrix $\langle \PPSixResImp, \{ \etvalue \} \rangle$.

\begin{therm}
   The logic $\PPResAsslogicSingle$
    has the (EIP), the (CIP)
    and the (MIP).
\end{therm}
\begin{proof}
 It is enough to show it has the (MIP).
Since $\PPResAsslogic$ is finitary, it suffices to consider finite sets.
 Given $\FmSetA,\FmSetB \vdash_\PPResAsslogicSingle \FmA$,
 let $\Props(\FmSetA)\cap \Props(\FmSetB\cup \{\FmA\})=\{p_1,\ldots,p_k\}\neq \varnothing$.
 Consider $\InterpValSet \SymbDef \{v\in\Hom (\LangSetProp{\Sigma}{\{p_1,\ldots,p_k\}}, \PPSixResImp) : v(\FmSetA)=\etvalue\}$.
 For
  each  $v\in\InterpValSet$ and $1\leq i\leq k$, let
 \begin{align*}
     \FmA^v_i(p)&
     \SymbDef
     \begin{cases}
     p\land \cons p & \mbox{ if }v(p_i)=\etvalue\\
     \DMNeg \down p  & \mbox{ if }v(p_i)=\tvalue\\
     \up p\land \down p \land \DMNeg \cons p  & \mbox{ if }v(p_i)\in \{\both,\neither\}\\
     \DMNeg \up p  & \mbox{ if }v(p_i)=\fvalue\\
     \DMNeg p\land \cons p  & \mbox{ if }v(p_i)=\efvalue
 \end{cases}\\
 I^v_\both& \SymbDef \{i:v(p_i)=\both\}\\
 J^v_\neither& \SymbDef \{j:v(p_j)=\neither\}
 \end{align*}
 and
   \begin{align*}
 \FmB_v&\SymbDef\bigwedge_{1\leq i\leq k}\FmA^v_i(p_i)\land \bigwedge_{i\in I_\both^v,j\in J_\neither^v}\DMNeg\cons(p_i\GenImp p_j)\\
 \FmC&\SymbDef\bigvee_{v\in \InterpValSet}\FmB_v
 \end{align*}

If we set $\equiv \; \subseteq \SixSet\times \SixSet$ such that
$a \equiv b$ iff $a=b$ or $\{a,b\}=\{\both,\neither\}$ we obtain that, for all $v'\in \Hom (\LangAlg{\Sigma}, \PPSixResImp)$,
 $$v'(\FmA^v_i(p))=
 \begin{cases}
     \etvalue & \text{if }  v'(p)\equiv v(p_i) \\
     \efvalue &  \text{otherwise.}\\
 \end{cases}
 $$
 Thus, if
 $v'(\FmB_v)=\etvalue$ then $v(p_i)\equiv v'(p_i)$ for every 
 $1\leq i\leq k$. 

Let us show that $\FmC\vdash_\PPResAsslogicSingle \FmA$.
 For every $v'\in \Hom (\LangAlg{\Sigma}, \PPSixResImp)$
such that $v'(\FmC)=\etvalue$,
there is $v\in \InterpValSet$
such that
$v(p_i)\equiv v'(p_i)$ for every 
 $1\leq i\leq k$.
 Also, $\{\both,\neither\}=\{v(p_i),v(p_j)\}$ if, and only if,
 $\{\both,\neither\}=\{v'(p_i),v'(p_j)\}$, and therefore 
 $v(\FmB)=v'(\FmB)$ for every $\FmB$ such that $\Props(\FmB)\subseteq\{p_1,\ldots,p_k\}$.
Thus, without loss of generality, we proceed considering that  
$v(p_i)=v'(p_i)$ for $1\leq i\leq k$.

Considering $v''\in \Hom (\LangAlg{\Sigma}, \PPSixResImp)$ such that 
$$v''(p)\SymbDef
\begin{cases}
    v(p) & p\in P{\setminus}\Props(\FmSetA)\\
    v'(p) &p\in \Props(\FmSetA){\setminus}\{p_1,\ldots,p_k\} 
\end{cases}$$
we have that $v''(\FmSetA)=v''(\FmSetB)=\etvalue$,
thus from $\FmSetA,\FmSetB\vdash_\PPResAsslogicSingle \FmA$
we conclude that $v''(\FmA)=v'(\FmB)=\etvalue$ and therefore
 $\FmC \vdash_\PPResAsslogicSingle \FmA$.
Finally, 
to see that $\FmSetA\vdash_\PPResAsslogicSingle \FmC$,
note that if $v(\FmSetA)=\{\etvalue\}$ then,
by definition, $v\in \InterpValSet$                                             and therefore                    
 $v(\FmB_v)=v(\FmC)=\etvalue$.\qedhere
\end{proof}


When a \SetFmla{} logic satisfies some of the
above interpolation properties
and we know it is algebraizable,
the class of algebras corresponding to the equivalent algebraic semantics
satisfies so-called
`amalgamation properties',
which we formulate below
based on \cite{pigo98}.

Let $\mathsf{K}$ be a class of $\Sigma$-algebras and $\AlgA,\AlgB \in \mathsf{K}$.
If $\AlgA$ is a subalgebra of $\AlgB$
and $S \subseteq B$,
$\AlgB$ is said to be a \emph{$\mathsf{K}$-free
extension of $\AlgA$ over $S$}
if for every $\AlgC \in \mathsf{K}$, every homomorphism
$h : \AlgA \to \AlgC$
and every $f : S \to C$,
there is a unique homomorphism
$g : \AlgB \to \AlgC$
such that $g \upharpoonleft A = h$
and $g \upharpoonleft S = f$,
where $\upharpoonleft$ denotes domain restriction.
Write
$f : \AlgA \rightarrowtail \AlgB$
to denote that $f$ is an injective homomorphism,
and 
$f : \AlgA \hookrightarrow \AlgB$
to denote that $f$ is a \emph{free injection
over $\mathsf{K}$}, meaning that
$f$ is an injection and
$f(\AlgB)$ is a \emph{$\mathsf{K}$-free extension}
of $\AlgB$ over some set $S$ of elements.

\begin{definition}[\cite{pigo98}, Def. 5.2]
    A class of algebras $\mathsf{K}$
    has the
    \begin{enumerate}
        \item \emph{ordinary amalgamation property} (OAP) when
        for all $\AlgA,\AlgB,\AlgC \in \mathsf{K}$,
        and all
        homomorphisms
        $f : \AlgC \rightarrowtail \AlgA$
        and
        $g : \AlgC \rightarrowtail \AlgB$,
        there exists 
        $\AlgD \in \mathsf{K}$
        and homomorphisms
        $h : \AlgA \rightarrowtail \AlgD$
        and
        $k : \AlgB \rightarrowtail \AlgD$
        such that
        $hf = kg$.
        \item \emph{flat amalgamation property} (FAP) when
        for all $\AlgA,\AlgB,\AlgC \in \mathsf{K}$,
        and all
        homomorphisms
        $f : \AlgC \rightarrowtail \AlgA$
        and
        $g : \AlgC \hookrightarrow \AlgB$,
        there exists 
        $\AlgD \in \mathsf{K}$
        and homomorphisms
        $h : \AlgA \hookrightarrow \AlgD$
        and
        $k : \AlgB \rightarrowtail \AlgD$
        such that
        $hf = kg$.
        \item \emph{Maehara amalgamation property} (MAP) when,
        for all $\AlgA,\AlgB,\AlgC \in \mathsf{K}$,
        and all
        homomorphisms
        $f : \AlgC \rightarrowtail \AlgA$
        and
        $g : \AlgC \to \AlgB$,
        there exists 
        $\AlgD \in \mathsf{K}$
        and homomorphisms
        $h : \AlgA \to \AlgD$
        and
        $k : \AlgB \rightarrowtail \AlgD$
        such that
        $hf = kg$.
    \end{enumerate}
\end{definition}
\noindent 
The reader is referred to
\cite[Thm. 5.3]{pigo98}
for the relationships between
the above notions of amalgamation.
In particular, we have that
(MAP) implies both (FAP) and (OAP).


\begin{therm}
   $\PPImpResTempVar$
   has the
   (OAP),
   the
   (FAP), 
   the (MAP).
\end{therm}
\begin{proof}
Clearly,
it suffices to show 
that $\PPImpResTempVar$
has the (MAP).
By
Theorem~
\ref{thm:algppass},
we know that
$\PPImpResTempVar$ is the equivalent algebraic semantics of $\PPResAsslogic$.
We just saw that this logic
satisfies (MIP);
hence, by~\cite[Cor. 5.27]{pigo98},
$\PPImpResTempVar$
satisfies (MAP).
\end{proof}

\section{Conclusions and future work}
\label{sec:conclusions}
The present paper has initiated the study of implicative
expansions of logics of perfect paradefinite algebras
by considering classic-like and Heyting implications, both in the \SetSet{} and in the
\SetFmla{} framework.
We investigated semantical characterizations (via classes of algebras and logical matrices) as well as proof-theoretical ones
(via \SetSet{} and \SetFmla{} Hilbert-style calculi) of these expansions.
For the expansions with a Heyting implication, 
the new
connective introduced a further challenge, for
 the resulting logic ($\PPResImplogic$)
cannot be 
characterized by a single  logical matrix;
however, we have proved that it 
can be characterized by a single finite PNmatrix.
Over such expansions we also studied  properties of interpolation and amalgamation for the corresponding algebraic models.
We indicate below a few  directions that we believe could prove worthwhile 
pursuing in future research.

\paragraph*{Alternative expansions of the logics 
of PP-algebras.}

In Section~\ref{sec:classical-implication} we have briefly considered 
other logics that may be obtained by conservatively adding an implication to the logics of PP-algebras
$\PPlogicPrimeMC$
and
$\PPlogic$. 
We did not explore these alternatives much further, preferring instead
 (from Section~\ref{sec:intuitionistic-implication} onward) to focus our attention on logics  
 having a more straightforward connection to existing frameworks
 (Moisil's logic and symmetric Heyting algebras).
Nevertheless, we feel that such alternative systems may  deserve further study,
and the connection we noted with the recent work by~\cite{coniglio2023sixvalued} 
provides further motivation for this project. Additionally, it would be interesting
to systematically investigate  the logics determined by
refinements of 
the 
PNmatrix 
     $\PartialMat_{\mathsf{up}}$, providing axiomatizations,
     classifying them within the  hierarchies of algebraic logic
     (e.g.,~which among them are algebraizable?)
     and studying the corresponding classes of algebras as done in 
     Section~\ref{sec:moi} for $\PPResImplogic$. 
Lastly, further enlarging the scope, one might recall from Proposition~\ref{A1-is-classiclike} how condition $\CondAOne$ corresponds to classic-like implications (Definition~\ref{def:classicalimplication}), and ask whether such condition 
could be relaxed, perhaps requiring the new implication to be axiomatized 
by the usual inference rules for Heyting implications.
 In this way we could circumvent the limitation of Theorem~\ref{thm:noexist}
 exploring  alternative self-extensional expansions of the logics of PP-algebras.





    

\paragraph*{Extensions of $\PPResImplogic$.}
Our preliminary investigations suggest that the landscape of extensions of the base logics
considered in the present paper is quite interesting and complex. Concerning the
 \emph{finitary} $\SetFmla$ extensions of $\PPResImplogic$ we may affirm the following:
\begin{itemize}
    \item By~\cite[Thm.~3.7]{Jansana2006}, the finitary $\SetFmla$ self-extensional extensions of 
    $\PPResImplogicSingle$ are in one-to-one correspondence with the subvarieties of $\PPImpResVar$.
    Thus, by
    Corollary~\ref{cor:subvar}, there are only three of them, all of them axiomatic
    (none of them being conservative expansions of $\PPlogic$). 
    These logics may be axiomatized, relatively to $\PPResImplogicSingle$,
    by adding axioms corresponding to the equations
    in Corollary~\ref{cor:subvar}, to wit: 
\begin{enumerate}
 \item the logic of $\Variety (\PPFourResImp)$ is  axiomatized by $(p \land \DMNeg p) \Imp (q \lor \DMNeg q)$;

 \item the logic of $\Variety (\PPThreeResImp)$ by $ p \lor (p  \Imp   (q \lor \neg q  ) )$; and

 \item the logic of $\Variety (\PPTwoResImp)$, which is (up to the choice of language) just classical logic, by $p \lor \DMNeg p$.
 \end{enumerate}

 

    \item The number of {axiomatic} $\SetFmla$ extensions of $\PPResImplogic$ 
    (all logics which are obviously finitary)
    is larger but also finite,
    for  each axiomatic extension may be be characterized by 
    the submatrices of the original matrices that satisfy the axioms.\footnote{
    This result can be easily established  using
     \cite[Prop.~3.1]{OnAxRexp} together with the observation that the total components of the PNmatrix 
     $\PartialMat_{\mathsf{up}}$ 
     are deterministic;  thus,  the submatrices of $\PartialMat_{\mathsf{up}}$ which are images of valuations that satisfy every instance of an axiom are sound with respect to that particular axiom.}



    \item As we have seen (cf.~Proposition~\ref{prop:1assert-single-mat}), the assertional logic $\PPResAsslogic$ 
    is itself a (non-axiomatic, non-self-extensional) extension of $\PPResImplogic$. 
    Algebraizability of $\PPResAsslogic$ (Theorem~\ref{thm:algppass}) entails that its axiomatic $\SetFmla$ extensions are in one-to-one correspondence with the subvarieties of $\PPImpResVar$,
    so again we can conclude that there are only three of them.
    

    \item In contrast to the preceding results, we conjecture that 
    it may be possible to construct 
    countably many distinct (non-axiomatic) $\SetFmla$ 
    finitary 
    extensions of $\PPResAsslogic$. 
    The theory of algebraizability would then tell us that 
    the variety $\PPImpResTempVar$ has at least countably many sub-quasivarieties.
    \end{itemize}

Now considering the 
$\SetSet$ extensions of 
    $\PPResImplogicOrderMC$,
    we can say the following:

\begin{itemize}
    \item Similarly to
    the $\SetFmla$ setting, all \emph{axiomatic} $\SetSet$ extensions of 
    $\PPResImplogicOrderMC$ 
    are characterized by the sets  of (sub)matrices in $ \{ \langle \PPSixResImp, \up a \rangle : a \in  \mathcal{V} \}$ that satisfy the corresponding axioms, hence their cardinality is also finite.
    Their $\SetFmla$ companions
    are obtained 
     by adding the same axioms to $\PPResImplogic$.

\item In consequence, there are \emph{at least} 
    three self-extensional $\SetSet$  extensions of 
    $\PPResImplogicOrderMC$, 
    namely,
    the ones determined by the matrices which satisfy the corresponding axioms.
    (There may be more,
    but each of them will 
    have as $\SetFmla$ companion 
    one of the three self-extensional $\SetFmla$ 
    extensions of $\PPResImplogic$ mentioned earlier.)

    \item The methods  used in the previous sections may be used to obtain analytic $\SetSet$ axiomatizations  for all the 
    logics determined by
    sets of (sub)matrices in $ \{ \langle \PPSixResImp, \up a \rangle : a \in  \mathcal{V} \}$.

\end{itemize}

The preceding considerations suggest that 
a complete description of the lattice of all extensions of 
our base logics
    is well beyond the scope of the present work, and will have to be pursued in future investigations. 
    We summarize such a research programme
    in terms of the following couple of problems:

\begin{problem}
Describe the lattice of all ($\SetFmla$) extensions of $\PPResImplogic$
and the lattice of all ($\SetSet$) extensions of 
$\PPResImplogicOrderMC$. 
\end{problem}

\begin{problem}
Look at the same problem again, but now restricting one's attention, on the one hand, to the 
sublattice consisting of all the 
$\SetFmla$ extensions of $\PPResAsslogic$, and on the other hand,
to the 
sublattice consisting of all the 
$\SetSet$ extensions of 
the $\SetSet$ 
logic 
determined by the matrix $\langle \PPSixResImp, \Upset{\etvalue} \rangle$.
Due to the algebraizability of 
$\PPResAsslogic$, for finitary $\SetFmla$ logics the problem may be rephrased as: describe
the lattice of all subquasivarieties of $\PPImpResVar$.
\end{problem}


\paragraph*{Fragments of the language of $\PPResImplogic$.}

A close inspection of the methods  employed in the previous sections to axiomatize 
$\PPResImplogic$ and its $\SetSet$ companions suggests that
these may also be applied so as to obtain analytic axiomatizations 
for the logics corresponding to those fragments of the language over the connectives in
$\{ \land, \lor, \Imp, \DMNeg, \cons, \bot, \top \}$ that
are sufficiently rich to express an appropriate set of separators.
Some of these, we believe, have intrinsic logical and algebraic interest, and may deserve further study. Let us single out, for instance, the fragments corresponding to 
the connectives $\{ \Imp, \cons \}$, 
$\{ \Imp, \De \}$ (recall that $\De x : = \DMNeg x \Imp \DMNeg (x \Imp x)$) and $\{ \Imp, \DMNeg \}$. The first of them may be of interest
in the study of implicative fragments of Logics of Formal Inconsistency, while the second 
could be studied in the setting of implicative fragments of algebras with modal operators. The study of the third could lead to an interesting generalization of Monteiro's results on symmetric Heyting
algebras and their logic.

\section*{Acknowledgments}
Vitor Greati acknowledges support from the FWF project P33548.
Sérgio Marcelino's research was done under the scope of project FCT/MCTES through national funds and when applicable co-funded by EU under the project UIDB/50008/2020.
João Marcos acknowledges research support received from CNPq.
Umberto Rivieccio acknowledges support from the 2023-PUNED-0052
grant ``Investigadores tempranos UNED-SANTANDER''
and  from the I+D+i research project PID2022-142378NB-I00 ``PHIDELO'', funded by the Ministry of Science and Innovation of Spain.


   

\bibliographystyle{abbrvnat}
\bibliography{refs}

\end{document}